%% file: main.tex
\documentclass[12pt]{article}
\usepackage{subfiles}
\usepackage{natbib}
\def\mytitle{Valid Post-Selection Inference in Robust Q-Learning}
\title{\mytitle}
\author{
  \small{
  Jeremiah Jones$^{1*}$\href{j.miahjones+research@gmail.com}, Ashkan Ertefaie$^2$,
  James R. McKay$^3$,} \\
  \small{David W. Oslin$^{4,5}$,
  Robert L. Strawderman$^6$}  \\ 
  \footnotesize{$^1$Eli Lilly and Company, Indianapolis, IN} \\
  \footnotesize{$^2$  Department of Biostatistics, Epidemiology and Informatics,
University of Pennsylvania}\\
  \footnotesize{$^3$Center on the Continuum of Care in the Addictions,} \\
  \footnotesize{Department of Psychiatry, Perelman School of Medicine, University of Pennsylvania} \\
  \footnotesize{$^4$VISN 4 Mental Illness Research and Education Center, Crescenz VA Medical Center} \\
  \footnotesize{$^5$Treatment Research Center and Center for Studies of Addictions,} \\
  \footnotesize{Department of Psychiatry, Perelman School of Medicine, University of Pennsylvania} \\
  \footnotesize{$^6$Department of Biostatistics and Computational Biology, University of Rochester} \\
}
\date{August 2025}

\usepackage{jjones}
\usepackage[nodisplayskipstretch]{setspace}
\usepackage{placeins}
\usepackage{rotating}
\usepackage{subcaption}
\usepackage{mathrsfs} 

\newcommand{\thebib}{
  \bibliographystyle{biom}
  \bibliography{uposi-rql}
}

\usepackage{bbm}

\begin{document}


\maketitle

\begin{abstract}
  Q-learning facilitates the development of an optimal adaptive treatment strategy through stagewise regression on a pre-specified set of tailoring variables and confounders. Semiparametric robust Q-learning eliminates the residual confounding that can occur when parametric working models for  confounding influences are misspecified. However, in the presence of many potential tailoring variables, constructing an optimal adaptive treatment strategy using either approach may lead to including extraneous variables that contribute little or no benefit while increasing implementation costs, thereby placing an undue burden on patients. Using data-driven selection processes to identify a smaller set of informative prognostic factors is straightforward; however, proper statistical inference must account for this selection process. In this paper, we adapt the Universal Post-Selection Inference (UPoSI) procedure to the semiparametric Robust Q-learning method. UPoSI, introduced for use with linear models, allows for very general variable selection mechanisms. Our approach addresses the unique challenges stemming from the use of UPoSI with semiparametric multistage decision methods. Theoretical and simulation results demonstrate the validity of the proposed confidence regions. We illustrate our proposed methods through an application to adaptive treatment strategy estimation for substance 
  abuse.
\end{abstract}





\doublespacing

\section{Introduction}
\label{sec:intro}

Data-driven individualized adaptive strategies have drawn substantial  attention in recent years in economics, statistics and health research \citep{meier2012glp, kitagawa2018should, watts2020optimizing, xu2022estimating}. The key goal is to optimize the expected value of a specified outcome by tailoring treatments to  individuals based on their ongoing performance and characteristics. The quality of the constructed strategies, however, can be severely hampered by confounding bias or inclusion of many spurious variables. 

Existing methods focused on finding optimal individualized treatment strategies (i.e., strategies leading to the best outcome) can be categorized as either direct \citep{zhao2012estimating,zhangEstimatingOptimalTreatment2012} or indirect methods \citep{ertefaieRobustQLearning2021,wallaceDoublyrobustDynamicTreatment2015,schulteQandAlearningMethods2014,murphyOptimalDynamicTreatment2003}. The former requires modeling an outcome and the latter relies on modeling the treatment assignment mechanism.  Direct methods  
can be 
inefficient and fail to provide reasonable inference for the parameters that define non-smooth decision rules (e.g., indicator or max operators) due to slow rates of convergence \citep{chakrabortyStatisticalMethodsDynamic2013}. Indirect methods do not  suffer from these shortcomings at the cost of relying on the correctly specified outcome model. This can be an important limitation in observational settings where model misspecification can lead to residual confounding and lack of causal interpretability.

Importantly, the form of the decision rules arising from  indirect methods do not depend on the full outcome model, but on the so-called ``blip'' \citep{robinsOptimalStructuralNested2004} or ``contrast'' function \citep{zhangEstimatingOptimalTreatment2012,schulteQandAlearningMethods2014}. Motivated in part by this fact, \citet{ertefaieRobustQLearning2021} proposed an indirect semiparametric method, or Robust Q-learning, to mitigate the residual confounding bias issue by leveraging the Robinson-Speckman transformation \citep{robinsonRootnconsistentSemiparametricRegression1988,speckmanKernelSmoothingPartial1988}. Robust Q-learning enables semiparametric estimation of an optimal strategy within a class of decision rules defined using a finite dimensional vector of parameters and a pre-specified set of tailoring variables.

Data-driven selection of tailoring variables has received recent attention. \cite{qian2011performance} demonstrated that excess variables included in the decision rule can negatively impact the quality of 
individualized 
treatment rules,
using an $\ell_1$-penalized least-squares approach was proposed to adapt to the underlying sparsity.
\citet{wallaceModelSelectionEstimation2019} developed an information criterion for the G-estimation framework, although such methods do not scale well with dimension. 
\citet{shi2018high} proposed an adaptive lasso technique in the A-learning setting \citep{luVariableSelectionOptimal2013, schulteQandAlearningMethods2014}, while
\citet{bianVariableSelectionRegressionbased2021} introduced a Lasso-type penalty in the framework of \citet{wallaceDoublyrobustDynamicTreatment2015}. \cite{shi2016robust} propose a two-stage concave penalization estimator to accommodate high-dimensional estimation and selection in both the outcome and mediator models. However, none of these techniques provide valid inference for the coefficients corresponding to the selected variables. 
Classical inference frameworks fail in the presence of variable selection because of the bias that is introduced by only performing inference on those coefficients 
estimated to be further away from zero than others \citep{leebModelSelectionInference2005,leebSparseEstimatorsOracle2008,berkValidPostselectionInference2013}. 
For a recent review of methods introduced
to handle such problems, see  \cite{kuchibhotlaValidPostselectionInference2020}, which motivates the study of post-selection inference from a replicability perspective  \citep{leeExactPostselectionInference2016}.

In an individualized decision making problem, the ability to carry out valid statistical inference is also crucial after selecting a model using a data driven approach, as it enables investigators to statistically evaluate whether an
optimized treatment strategy is significantly better than the other ones.  If the null hypothesis cannot be rejected for some treatments (e.g., treatment 1 and 2 are equally beneficial), then a caregiver may consider other factors such as treatment cost, side effects and patient preferences to choose among them. Recently, \cite{zhao2022selective} proposed a post selection inference method for effect modifier selection in a 
single-stage setting by generalizing the selective inference framework \citep{leeExactPostselectionInference2016}. One of the key  challenges in providing inference after selection in multi-stage decision making is to account for the effect of the  randomness of the selected models throughout the stages which also impacts the target parameters.

We propose a new post-selection inference framework in multi-stage individualized decision making problems that decouples the variable selection problem from the post-selection inference one. In practice, this allows the analyst to use a variety of formal and informal approaches to select the variables used in estimating adaptive treatment strategies. We then link the variable selection problem to the primary goal of estimating an optimal treatment strategy through the framework of
Robust Q-learning. This approach has several 
advantages: (1) existing regression-based selection methods are easily extended to this setting; (2) confounding control functions can be modeled nonparametrically using data-adaptive methods while preserving the possibility of
root-$n$ inference; and (3) in contrast to the method of \citet{wallaceDoublyrobustDynamicTreatment2015}, this framework does not require a correctly-specified blip function in order to target a meaningful projection of the blip function.

The main contributions of this article are summarized as follows. First, Universal Post-Selection Inference \citep[][hereafter referred to as ``UPoSI'']{kuchibhotlaValidPostselectionInference2020} is generalized to Robust Q-learning to handle the subtleties arising in multi-stage settings. Our procedure provides strong asymptotic coverage guarantees for the selected parameter.  Second, we propose a version of the perturbation bootstrap,
proposed earlier for single-stage lasso-type variable selection methods \citep{dasPerturbationBootstrapAdaptive2019,minnierPerturbationMethodInference2011}, for the UPoSI setting that handles
multiple stages of model selection.
Third, we establish the theoretical properties of the proposed methods under both  fixed- and random-design settings. 
Simulation studies are used
to examine the finite sample performance of our methods and demonstrate notable improvements over selective inference in settings where both approaches are applicable. Data from the "Extending Treatment Effectiveness of Naltrexone" multi-stage randomized trial is used to illustrate our proposed methods.


\section{Notation}
\label{sec:notation-prob-statement}

For simplicity, we consider a two-stage study where binary treatment decisions are made at each time point. 
Suppose $n$ i.i.d. trajectories of $\bO \defined (\bX_1, A_1, \bX_2, A_2, Y)$ from an unknown distribution $P_0$
are observed.
For each stage $\ell=1,2$, the candidate tailoring variables $\bX_{\ell} \in \Xcal_{\ell} \subseteq \R^{ q_{\ell} }$ are assumed to precede the Stage $\ell$ binary treatment, $A_{\ell} \in \Acal_{\ell} \defined \{0,1\}$, and $Y$ represents the continuous outcome observed after both stages. 
Let all of the history preceding $A_{\ell}$ 
be represented by $\history{\ell},$ so that $\history{2} \defined (\bX_1^{\top}, A_1, \bX_2^{\top})^{\top}$ and $\history{1} \defined \bX_1$. These take values in $\historycal{2} \defined \Xcal_1 \times \Acal_{1} \times \Xcal_2$ and $\historycal{1} \defined \Xcal_1$, respectively.
%
Let $Y^*(a_1, a_2)$ denote the potential outcome of $Y$ if the treatments are set to $\{A_1=a_1,~A_2=a_2\}$. 
Define $\mu_{\ell A0}(\obshistory{\ell}) \defined \Eop(A_{\ell} \vert \history{\ell} = \obshistory{\ell} ), \ell = 1,2$ as the treatment propensities in each stage.
We make the commonly used assumptions for studying causal effects in this setting: stable unit treatment value assumption, positivity, and sequential ignorability at each stage given 
$\history{\ell}$ \citep{murphyOptimalDynamicTreatment2003, robinsOptimalStructuralNested2004}.


For vectors or column matrices, we use $\| \cdot \|_q$ to represent the $\ell_q$ norm. For real square matrices, we use $\| \cdot \|_{\infty}$ to represent the maximal element of the matrix. We will also make use of the $L_2(P_0)$ norm for random functions: $\LtwoPzero[ ]{h} = \{\Eop_{P_0} h^2\}^{1/2}$. 
The minimal and maximal eigenvalues of a matrix $\bA$ 
will be denoted by $\lambda_{min}(\bA)$ and $\lambda_{max}(\bA)$, respectively.
Finally, we will also use the shorthand $a \vee b$ and $a \wedge b$ to represent the minimum and maximum, respectively, of the scalar variables $a$ and $b$.

\section{Submodel Selection in Robust Q-learning}
\label{sec:var-sel}

The observed data in $\history{\ell}$ defines a maximal amount of information which can be used for treatment decisions at Stage $\ell$. In this work, we will consider sparse linear models as working models for the contrast functions. We may consider the $p_{\ell}-$dimensional vector $\historybasis{\ell}$ to be a fixed-dimensional transformation of the vector $\history{\ell},$ such that $\historybasis{\ell}$ is fixed when conditioning upon $\history{\ell}.$ If an analyst decides \textit{a priori} on a dictionary of possible terms to include in a working model (e.g., all main effects and one-way interactions), we may view such a dictionary as a ``full model'' represented through the entire vector $\historybasis{\ell}$ in each stage. Submodels are created by subsetting this full vector. We may identify this subsetting operation with the indices of $\historybasis{\ell}$ used to create the sub-vector. That is, the ``full model'' at Stage $\ell$ may be identified with the object $\model_{\ell}^F := \{1,\ldots,p_{\ell}\}$ and submodels of $\model_{\ell}^F$ are those sets $\model_{\ell} \subseteq \model_{\ell}^F$. The maximal set of all possible submodels is given by the power set $\modelspace_{\ell} := \{ \model_{\ell}: \model_{\ell} \subseteq \model_{\ell}^F \}$. We will often restrict the set of submodels being studied in each stage to only those submodels with a certain level of sparsity; the set of all $C_{\ell}-$sparse models based on the data in $\historybasis{\ell}$
will be denoted $\modelspace_{\ell}(C_{\ell}) \defined \{ \model_{\ell} \in \modelspace_{\ell} : |\model_{\ell}| \leq C_{\ell} \},$
where $|\model_{\ell}|$
denotes set cardinality.

We will use the notation  $\model_{\ell} \in \modelspace_{\ell}$ to refer to an arbitrary (fixed) submodel. That is, $\model_{\ell}$ denotes a particular specification of the elements of $\historybasis{\ell}$ that will be used to assess effect modification and thereby tailor future treatment. A data-dependent model will be denoted as $\hat\model_{\ell}.$
Let $\bA(\model_{\ell})$ represent a sub-matrix or sub-vector of $\bA$ corresponding to the model indices $\model_{\ell}$. For example, when $\historybasis{1}$ is a 5-dimensional vector and $\model_{1}=\{1,2\}$, then $\historybasis{1}(\model_{1})$ represents a 2-dimensional sub-vector of $\historybasis{1}$ that contains its first 2 elements. Similarly, if $\bA$ is a $5 \times 5$ real matrix, $\bA(\model_{1})$ represents the $2 \times 2$ sub-matrix of $\bA$ with entries corresponding to first two rows and columns. 
Based on the description of $\modelspace_{\ell}$, we may think of any \textit{a priori}-specified model $\model_{\ell} \in \modelspace_{\ell}$ as representing the set 
of indices $j$ corresponding to the covariates $\historybasis{\ell}\subj$ to be included in said model.

\citet{ertefaieRobustQLearning2021} described the goal of \dtr estimation as identifying an optimal decision rule $d$ maximizing the value function $V(d)$ over the space of all rules $\Dcal \defined \Dcal_1 \times \Dcal_2$, where $\Dcal_j$ is the space of all decision rules mapping $\Xcal_j$ to a treatment decision in stage j. 
The expected value to patients may not be the only objective worth pursuing in estimating a rule, as alternative \dtr[ies] might achieve a similar value while requiring less invasive or less expensive data collection. 
In this situation, we might say that there exists $d'$ making use of only $\model_1$ and $\model_2$, with $V(d') \approx V(d^*).$ In this scenario, the expected regret of $d'$ might be considered small in relation to the more pragmatic benefits of sparsity.
Consequently, we might consider an optimization restricted to the space of rules which make use of only $\model_1$ and $\model_2,$ if it were known ahead of time which models were likely to yield small regrets.

In practice, an analyst may not be able to anticipate which submodels are inferior to the full models. As such, making use of the data to adapt to the underlying distribution is an attractive option. 
Random model selection is involves random variables $\hat{\model}_1,\hat{\model}_2$ that (i) depend at least in part on the observed data and (ii) take their values in the space $\modelspace_1\times \modelspace_2$.
In this interpretation, we may view each $\model_{\ell} \in \modelspace_{\ell}$ as a potential realization of the random variable $\hat{\model}_{\ell}$. Thus, model selection acts to select the ``relevant'' subspace $\Dcal_{\hat{\model}_1 \hat{\model}_2} \subset \Dcal$ and targeted parameters identifying the optimal strategy in that subspace. Here, we use $\Dcal_{\model_1 \model_2}$ for any  pair of models $(\model_1, \model_2) \in \modelspace_1 \times \modelspace_2$ to represent the set of parametric decision rules that use $\historybasis{1}(\model_1)$ in the first stage and $\historybasis{2}(\model_2)$ in the second.

As a departure from the fixed-model case, the data-dependent choice of subspace $\Dcal_{\hat{\model}_1 \hat{\model}_2}$ complicates subsequent inference \citep{leebSparseEstimatorsOracle2008}.
Our presentation of the submodel-restricted subspaces also illuminates a fundamental incompatibility with ``full model-based'' approaches to inference.
That is, estimation or inference that focus on subsets of a ``full parameter''---such as de-biased or de-sparsified approaches \citep{zhangConfidenceIntervalsLow2014,vandegeerAsymptoticallyOptimalConfidence2014}---are focused on different targets that are not optimal over the restricted space. This motivates our adoption of the ``post-selection'' viewpoint of parameter estimation and inference.

\section{Robust Q-learning with Fixed Submodels}
\label{sec:rob-trans-itr}

While the details of the  underlying models and derivation of Robust Q-learning are fully presented in \Cref{sec:rql-overview}, we briefly define the models here to orient the reader to the notation and main components used throughout the paper. The Robust Q-learning algorithm \citep{ertefaieRobustQLearning2021} applies the Robinson-Speckman transformation \citep{robinsonRootnconsistentSemiparametricRegression1988,speckmanKernelSmoothingPartial1988} to enable nonparametric modeling of the nuisance functions within the Q-learning framework. Specifically, for a two stage problem, we first consider the following second stage model 
\[Y - \mu_{2Y0}(\history{2}) = \left\{ A_2 - \mu_{2A0}(\history{2}) \right\} \historybasis{2}(\model_2)^\top \btheta_{20,\model_2} + \epsilon_{2,\model_2},\]
where  $\mu_{2Y0}(\obshistory{2}) \defined \Eop(Y ~|~ \history{2} = \obshistory{2})$. Then,
we define the first stage pseudo-outcome as the $\model_2-$dependent random variable $Y_{1 \model_2} \defined Y + \xi{\left\{ A_2, \historybasis{2}(\model_2) ; \btheta_{20,\model_2} \right\}}$ where $\xi{(a_2, \bx; \btheta)} \defined \bx^\top \btheta \left\{ \Ind{(\bx^\top \btheta > 0)} - a_2 \right\}.$ Next, we specify the first-stage model as 
\[Y_{1 \model_2} - \mu_{1Y \model_2 0}(\history{1}) = \left\{ A_1 - \mu_{1A0}(\history{1}) \right\} \historybasis{1}(\model_1)^\top \btheta_{10,\model_1 \model_2} + \epsilon_{1,\model_1 \model_2},\]
where $\mu_{1Y \model_2 0}(\obshistory{1}) = 
\Eop(Y_{1 \model_2} ~|~ \history{1} = \obshistory{1})$. The residual terms $\epsilon_{2,\model_2}$ and $\epsilon_{1,\model_1 \model_2}$ are defined in \Cref{sec:rql-overview}. Importantly, observe that the second-stage model $\model_2$ directly impacts the pseudo-outcome $Y_{1 \model_2}$ and its corresponding conditional expectation $\mu_{Y1 \model_2 0}(\history{1}),$  and subsequently influences our derivation of the first-stage quantities. This interdependence will
play a critical role when we consider the possibility
of using variable selection at each stage of Q-learning.

 In this section, we will focus on developing the practical aspects of using linear working models of the form $\historybasis{2}(\model_2)^\top \btheta_{2 \model_2}$ and $\historybasis{1}(\model_1)^\top \btheta_{1 \model_1}$ to capture the conditional average treatment effect functions in stages 2 and 1, respectively. As part of this process, the conditional expectations of treatment and outcome must be estimated at each stage. 
We adopt the use of $K-$fold cross-fitting for estimation by general statistical learners ~\citep{klaassen1987consistent, zheng2011cross,chernozhukovDoubleDebiasedMachine2018f}, as described below.

\subsection{Estimation via Cross-fitting}
\label{sec:cf-risk}

Let $K$ represent some fixed number of folds and $\Pcal_K$ represent a partition of $\{1,\ldots,n\}$ into $K$ indexing sets of roughly-equal size; i.e. $\Pcal_K = \{ \Ik : k=1,\ldots,K \}$ with $\cup_{k=1}^K \Ik = \{1,\ldots,n\}$ and $\Ik \cap \mathbf{I}_{k'} = \varnothing$ for $k\neq k'.$
Let $\bm{D}_{\mathbf{I}} \defined \{ \bO_i : i \in \mathbf{I} \}$ for any indices $\mathbf{I}.$ Then using $\mathbf{I}^c$ for the set complement, we write $\crossfitDataC$ to represent the observed data outside of $\Ik.$ 
We demonstrate cross-fitting by applying it to $\mu_{2Y0}(\cdot),$ the conditional expectation of $Y$ given $\history{2}$. To estimate the value of this function at $\history{2i},$ for $i \in \Ik,$ use $\crossfitDataC$ to train an estimator $\hat\mu_{2Y}(\cdot; \crossfitDataC)$ for the whole function and obtain a prediction at $\history{2i}$ for every $i\in\Ik.$ We change 
the held-out fold $\Ik$ until all of the necessary predictions have been made. Because $\Pcal_K$ is a partition, the sum over $i=1,\ldots,n$ may equivalently be written as a double-sum over $k=1,\ldots,K$ and $i \in \Ik.$

Employing this cross-fitting strategy, we obtain predictions for each $k=1,\ldots,K,~ i \in \Ik$ using the trained functions \( \hat\mu_{2 Y}(\history{2i}; \crossfitDataC) \) and \( \hat\mu_{2 A}(\history{2i}; \crossfitDataC),\) where this last function estimates the propensity of treatment $A_2$ given $\history{2}$ for subject $i$. This gives rise to $\hat{R}_{2 n, \model_2}(\btheta_{2, \model_2}),$ which is written as the least-squares objective function
\begin{multline}
  \hat{R}_{2 n, \model_2}(\btheta_{2, \model_2}) \defined \frac{1}{n} \sum_{k=1}^K \sum_{i \in \Ik} \big[ Y_i - \hat\mu_{2 Y}(\history{2i}; \crossfitDataC) 
  - \left\{ A_{2i} - \hat\mu_{2 A}(\history{2i}; \crossfitDataC) \right\} \historybasis{2i}(\model_2)^\top \btheta_{2, \model_2} \big]^2.
  \label{eq:cf-risk-stg-2}
\end{multline}
The minimizer of this function, $\hat\btheta_{2n, \model_2},$ satisfies the equations:
\begin{align}
  \label{eq:norm-eq-cf-2}
  \bzero &= \hat\bG_{2n}(\model_2) - \hat\bH_{2n}(\model_2) \hat\btheta_{2n, \model_2},
\end{align}
where the vector $\hat\bG_{2n}(\model_2) $ and matrix $\hat\bH_{2n}(\model_2) $ are defined as
  \begin{align}
    \hat\bG_{2n}(\model_2)  \defined \frac{1}{n} \sum_{k=1}^K \sum_{i \in \Ik} \historybasis{2i}(\model_2)  \left\{ A_{2i} - \hat\mu_{2A}(\history{2i}; \crossfitDataC) \right\}
    \times \left\{ Y_i - \hat\mu_{2Y}(\history{2i}; \crossfitDataC) \right\}
    \label{eq:cf-risk-quant-2}
    \\
    \hat\bH_{2n}(\model_2)  \defined \frac{1}{n} \sum_{k=1}^K \sum_{i \in \Ik} \left\{ A_{2 i} - \hat\mu_{2 A}(\history{2 i}; \crossfitDataC) \right\}^2 \left( \historybasis{2i}(\model_2)  \right)^{\otimes 2}.
    \nonumber
  \end{align}

With the second-stage estimator $\hat\btheta_{2n, \model_2}$ created, we move to the first stage. The ``blip function,''  $\xi{(a_2, \bx; \btheta)} \defined \bx^\top \btheta \left\{ \Ind{(\bx^\top \btheta > 0)} - a_2 \right\},$ can be used to create the first-stage outcome $\hat{Y}_{1 \model_2} \defined Y + \xi{\{ A_2, \historybasis{2}(\model_2); \hat\btheta_{2n,\model_2} \}},$ and repeat the outcome modeling process in the first stage. The estimate in the first stage solves the equations \( \bzero = \hat\bG_{1n,\model_2}(\model_1) - \hat\bH_{1n}(\model_1) \hat\btheta_{1n, \model_1 \model_2}, \) making use of the matrix
$\hat\bH_{1n}(\model_1)  \defined \frac{1}{n} \sum_{k=1}^K \sum_{i \in \Ik} \{ A_{1 i} - \hat\mu_{1 A}(\history{1 i}; \crossfitDataC) \}^2 ( \historybasis{1i}(\model_1))^{\otimes 2},$ and vector
$\hat\bG_{1n,\model_2}(\model_1)  \defined \frac{1}{n} \sum_{k=1}^K \sum_{i \in \Ik} \historybasis{1i}(\model_1)  \{ A_{1i} - \hat\mu_{1A}(\history{1i}; \crossfitDataC) \} \{ \hat{Y}_{1 \model_2 i} - \hat\mu_{1Y \model_2}(\history{1i}; \crossfitDataC) \}$.
These quantities used the cross-fitted predictions $\hat\mu_{1 Y \model_2}(\history{1i}; \crossfitDataC)$ and $\hat\mu_{1 A}(\history{1i}; \crossfitDataC),$ based on the same data-splitting used in the second stage. We note that the conditional expectation of the pseudo-outcome $\hat\mu_{1 Y \model_2}(\cdot)$ depends on $\model_2$ due to its inclusion in the blip function, although the propensity $\hat\mu_{1 A}$ does not.

\subsection{The Perturbation Bootstrap with Cross-fitting}
\label{sec:cf-perturb-risk}

The development of the cross-fitted empirical functions lends itself to a perturbation bootstrap approach. The perturbation bootstrap has previously been used as an inference method when tied to specific, nearly unbiased variable selection techniques like the adaptive lasso \citep{dasPerturbationBootstrapAdaptive2019,minnierPerturbationMethodInference2011}. Below, we will develop a more general method for post-selection inference that leverages the bootstrap to strongly control false coverage rates. 

To fix ideas, let $\bootWgt \sim P_{\bootWgt}$ be an analyst-specified random variable satisfying $\Eop(\bootWgt) = 1,~ \Eop(\bootWgt-1)^2=1$ with $\bootWgt_1,\ldots,\bootWgt_n$ i.i.d. from $P_{\bootWgt}.$ The perturbation bootstrap estimators $\hat\btheta^b_{2n,\model_2}$ and $\hat\btheta^b_{1n,\model_1 \model_2}$ solve equations similar to \eqref{eq:norm-eq-cf-2}. The quantities $\hat\bH_{1n}(\model_1),$ $\hat\bH_{2n}(\model_2)$, $\hat\bG_{2n}^b(\model_2)$, and $\hat\bG_{1n,\model_2}^b(\model_1)$ are defined similarly to their counterparts, except the $i^{th}$ term is multiplied by the random variable $\omega_i$. For $\hat\bG_{1n,\model_2}^b(\model_1),$ the pseudo-outcome $\hat{Y}_{1 \model_2}^b$ is also perturbed based on the boostrapped parameter $\hat\btheta^b_{2n,\model_2}$ from Stage 2. A key feature is that the same $\omega_i$ is used for observation $i$ between both stages to appropriately incorporate the correlation between stages.
More detail for these pseudo-outcomes is presented in \Cref{app:defining-pseudo-outcomes}.

\section{UPoSI for Population Parameters}
\label{sec:uposi}

\subsection{Adaptation of UPoSI to Robust Q-learning}
\label{sec:uposi-dtr}

The UPoSI procedure was presented in \citet{kuchibhotlaValidPostselectionInference2020} as an assumption-lean approach for performing inference on parameters after selection.
In this framework, the post-selection inference problem is formulated as providing coverage guarantees for confidence regions that are constructed for parameters, like those defined in \citet{ertefaieRobustQLearning2021}, after a random model selection event takes place. Unlike selective inference methods \citep{leeExactPostselectionInference2016}, the UPoSI framework is agnostic to the specific random model selection mechanism. In fact, the resulting inference is valid simultaneously over all plausible models. To retain focus, more detail on this perspective is given in \Cref{app:equiv-sel-sim}. 
We complete this section by generalizing the arguments of \citet{kuchibhotlaValidPostselectionInference2020} to Robust Q-learning.

Fix any pair of models $\model_1 \in \modelspace_1$ and $\model_2 \in \modelspace_2$. Then the following inequalities arise by simply adding and subtracting components of the empirical and population versions of the normal equations and using elementary inequalities:
\begin{align}
  \label{eq:uposi-inequality-2}
  \Vert \hat{\bH}_{2n}(\model_{2}) \{ \hat\btheta_{2n, \model_{2}} - \btheta_{20,\model_2} \} \Vert_\infty &\leq \hat{D}^{G}_{2n} + \hat{D}_{2n}^{H} \Vert\btheta_{20, \model_{2}}\Vert_1 \\
  \Vert \hat{\bH}_{1n}(\model_{1}) \{ \hat\btheta_{1n,\model_{1} \model_{2}} - \btheta_{10,\model_1 \model_2} \} \Vert_\infty &\leq \hat{D}^{G}_{1n,\model_2} 
   + \hat{D}_{1n}^{H}\Vert\btheta_{10,\model_{1} \model_{2}}\Vert_1.
  \label{eq:uposi-inequality-1}
\end{align}
The quantities on the RHS are related to components defined in \Cref{sec:cf-risk}:
\begin{equation}
  \label{eq:uposi-random-vars}
  \begin{aligned}
    \hat{D}_{2n}^G &\defined \| \hat\bG_{2n} - \bG_{20} \|_{\infty} & 
    \hat{D}_{2n}^H &\defined \| \hat\bH_{2n} - \bH_{20} \|_{\infty} \\ 
    \hat{D}_{1n,\model_2}^G &\defined \| \hat\bG_{1n,\model_2} - \bG_{10,\model_2} \|_{\infty} &
    \hat{D}_{1n}^H &\defined \| \hat\bH_{1n} - \bH_{10} \|_{\infty}
  \end{aligned}, 
\end{equation}
where $\hat\bG_{2n} = \hat\bG_{2n}(\model_2^F),$ 
$\hat\bH_{2n} = \hat\bH_{2n}(\model_2^F)$ and so on
are as defined earlier.

The random variables defined in \eqref{eq:uposi-random-vars} above are free of any selected models with the exception of $\hat{D}_{1n,\model_2}$, which depends on $\model_2$ due to the Stage 1 pseudo-outcome's dependence on the Stage 2 selected model. Consequently, \eqref{eq:uposi-inequality-2} holds simultaneously over all possible models $\model_2 \in \modelspace_2$ while \eqref{eq:uposi-inequality-1} holds (for each $\model_2$) simultaneously over all $\model_1 \in \modelspace_1$. 
If the joint distribution of the RHS quantities in each of \eqref{eq:uposi-inequality-1,eq:uposi-inequality-2} were known, we could plug-in the appropriate quantiles to turn the almost-sure inequalities into probabilistic ones that were valid, uniformly over $\modelspace_{\ell},$ at a specified confidence level.
In practice, these quantiles are unknown but may nonetheless be estimated by a perturbation bootstrap procedure. Using the bootstrap definitions of \Cref{sec:cf-perturb-risk}, the perturbation boostrap analogues of \eqref{eq:uposi-random-vars} are
\begin{equation}
  \label{eq:uposi-random-vars-boot}
  \begin{aligned}
    \hat{D}_{2n}^{Gb} &\defined \| \hat\bG_{2n}^b - \hat\bG_{2n} \|_{\infty} & 
    \hat{D}_{2n}^{Hb} &\defined \| \hat\bH_{2n}^b - \hat\bH_{2n} \|_{\infty} \\ 
    \hat{D}_{1n,\model_2}^{Gb} &\defined \| \hat\bG_{1n,\model_2}^b - \hat\bG_{1n,\model_2} \|_{\infty} &
    \hat{D}_{1n}^{Hb} &\defined \| \hat\bH_{1n}^b - \hat\bH_{1n} \|_{\infty}.
  \end{aligned}  
\end{equation}


The UPoSI regions may be constructed from inequalities \cref{eq:uposi-inequality-1,eq:uposi-inequality-2}: 
\begin{align}
  \hat{\Rcal}_{1 n, \model_1 \model_2} := \Big\{ \btheta \in \R^{|\model_1|} : \Vert \hat{\bH}_{1n}(\model_{1}) \{ \hat\btheta_{1n,\model_{1} \model_{2}} - \btheta \} \Vert_\infty \leq 
  \hat{C}_{1n,\model_2}^G(\alpha) + \hat{C}_{1n}^{H}(\alpha) \Vert\hat\btheta_{1n,\model_{1} \model_{2}}\Vert_1 \Big\} \label{eq:uposi-1-dagger-star} \\
  \hat{\Rcal}_{2 n, \model_2} := \Big\{ \btheta \in \R^{|\model_2|} : \Vert \hat{\bH}_{2n}(\model_{2}) \{ \hat\btheta_{2n, \model_{2}} - \btheta \} \Vert_\infty \leq
  \hat{C}^{G}_{2n}(\alpha) + \hat{C}_{2n}^{H}(\alpha) \Vert\hat\btheta_{2n, \model_{2}}\Vert_1 \Big\}. \label{eq:uposi-2-dagger-star}
\end{align}
The functions of $\alpha$ on the RHS are multivariate quantiles based on the bootstrap. Specifically, these are defined as any pair of numbers which satisfy \(
  P\big( \hat{D}_{1n,\model_2}^{Gb} \leq \hat{C}_{1n,\model_2}^{G}(\alpha),~ \hat{D}_{1n}^H \leq \hat{C}_{1n}^{H}(\alpha)\big) \geq 1-\alpha  
\) in the case of the first region, and \(
  P\big( \hat{D}_{2n}^{Gb} \leq \hat{C}_{2n}^{G}(\alpha),~ \hat{D}_{2n}^{Hb} \leq \hat{C}_{2n}^{H}(\alpha) \big) \geq 1-\alpha
\)
in the case of the second. We state the simultaneous coverage result below and delay discussion of theoretical assumptions to \Cref{sec:cond}.

\begin{theorem}[Validity of the confidence regions]
  \label{thm:uposi-coverage}
  Under \Cref{assump:boundedness,assump:bounded-model,assump:rate-assumps,assump:model-selection-consistency,assump:regularity} presented in \Cref{sec:uposi-theory}, the confidence regions \cref{eq:uposi-1-dagger-star,eq:uposi-2-dagger-star} satisfy:
  \begin{multline}
    \label{eq:uposi-coverage-statements}
    \liminf_{n\rightarrow\infty} \Prob\left( \btheta_{10,\hat\model_1 \hat\model_2} \in \hat{\Rcal}_{1 n, \hat\model_1 \hat\model_2} \right) \geq 1-\alpha
    \text{ and } 
    \liminf_{n\rightarrow\infty} \Prob\left( \btheta_{20,\hat\model_2} \in \hat{\Rcal}_{2 n, \hat\model_2} \right) \geq 1-\alpha.
  \end{multline}
\end{theorem}


\subsection{Discussion of Population-level Inference}
\label{sec:uposi-discuss}

Since the design is explicitly acknowledged to be random in the setting of \dtr[ies], the previous sections arguably present the most natural way to formulate the target parameters and associated methods for inference. However, this formulation does not guarantee desirable behavior of the UPoSI regions. For instance, the inequalities \eqref{eq:uposi-inequality-1,eq:uposi-inequality-2} that inform the size of the regions both depend on the $\ell_1$ norm of the unknown parameter. This dependence is problematic for two reasons:
(1) the size of the UPoSI regions is not equivariant under scaling transformations, and (2) the uncertainty reflected in the regions grows with size of the target parameter.
Transforming the variables to some unitless representation (e.g., scaling each column by the standard deviation) can ameliorate the first problem.
The second problem is less-easily handled. In practice, this behavior can produce regions that are much too conservative in settings where one might expect high power.

The opposite is also true: we can imagine a hypothesis-testing scenario where the region $\hat\Rcal_{1n,\model_1 \model_2}$ 
is used to test the hypothesis that our selected rule space is completely spurious. This can be formalized through a test of $H_0: \btheta_{10,\hat\model_1 \hat\model_2} = \bzero.$ In this setting, 
one replaces $\hat\btheta_{1n,\hat\model_1 \hat\model_2}$ in 
\eqref{eq:uposi-1-dagger-star} by its hypothesized value $\bzero,$ eliminating the second term on the RHS of the inequality. The size of the resulting region is then determined by $C_{1n,\hat\model_2}^G(\alpha),$ leading to a test that rejects if $\| \hat\bH_{2n}(\hat\model_1) \hat\btheta_{1n,\hat\model_1 \hat\model_2} \|_{\infty} > C_{1n,\hat\model_2}^G(\alpha).$ In our simulations, we have found that this test has relatively high power even when the regions yield very wide confidence intervals due to issue \#2 above. In the next section, we present an alternative formulation of the UPoSI regions which is more advantageous in terms of region size.

\section{Conditioning on the Design}
\label{sec:cond}

\subsection{Conditional vs. Population Inference}
\label{sec:cond-intro}

Our previous exposition explicitly 
accounts for the random nature of the histories $\history{1}, \history{2}$ in each stage. In the OLS problem that was handled therein, \citet{kuchibhotlaValidPostselectionInference2020} noticed similar phenomena to those discussed in \Cref{sec:uposi-discuss}, and argued along similar lines to propose fixed-design UPoSI regions that were smaller than the population versions, although an explicit derivation for the conditional perspective in a random setting was not provided. As we demonstrate later in \Cref{sec:uposi-conditional}, viewing the inference problem conditionally upon the design elements 
can also result in smaller confidence regions.
Alternative methods dealing with inference after selection are typically derived from with this inferential viewpoint, considering the design to be either deterministic or conditionally fixed \citep{berkValidPostselectionInference2013,leeExactPostselectionInference2016,tianAsymptoticsSelectiveInference2017}.
Such ``conditional'' approaches provide interesting comparisons in terms of power or confidence interval length, as well as false coverage rates. 


\subsection{Defining the Conditional Targets}
\label{sec:cond-targets}

In \Cref{sec:uposi-dtr}, the targets were defined implicitly by taking expectations of the random quantities involved in the first-order equations like \eqref{eq:norm-eq-cf-2}. The linear working model in Stage $\ell$ has a design which depends on both the history variable $\history{\ell}$ as well as the treatment $A_{\ell}$ for each subject. Fixing these design elements can be achieved by conditioning on all such variables. Let
\(
  \designElements{\ell} \defined \{ \history{\ell 1}, A_{\ell 1}, \ldots,\history{\ell n}, A_{\ell n} \}
\)
define the set of the design elements in Stage $\ell.$
In order to obtain the design-conditional target in Stage 2, first let $\bHOracle{2}$ represent the expectation $\Eop(\hat\bH_2 | \designElements{2})$ if the unknown functions $\hat\mu_{2Y}$ and $\hat\mu_{2A}$ were replaced by the true $\mu_{2Y0}$ and $\mu_{2A0}$, respectively. Use a similar method to define $\bG_{2n}^{cond}$ based on $\hat\bG_{2n}$.
Then the conditional target $\btheta_{2n,\model_2}^{cond}$ solves the equation
\begin{equation}
  \label{eq:norm-eq-true-cond-2}
  \bzero = \bG_{2n}^{cond}(\model_2) - \bHOracle{2}(\model_2) \btheta_{2n,\model_2}^{cond},
\end{equation}
ensuring it minimizes a squared error that only is random through $\designElements{2}.$

A similar argument may be used in the first stage as well, which would lead to a target $\btheta_{1n,\model_1 \model_2}^{cond},$ defined through the relation
\begin{equation}
  \label{eq:norm-eq-true-cond-1}
  \bzero = \bG_{1n, \model_2}^{cond}(\model_1) - \bHOracle{1}(\model_1) \btheta_{1n,\model_1 \model_2}^{cond}.
\end{equation}
\noindent
We subscript the targets $\btheta_{2n,\model_2}^{cond}$ and $\btheta_{1n,\model_1 \model_2}^{cond}$ by $n$ in order to remind the reader that these quantities depend on the underlying data through the design elements.

One important difference from the conditional targets generated by parametric Q-learning is that the role of confounding has been significantly diminished with the use of Robust Q-learning. As an illustration, if $\Delta_2$ is the conditional average treatment effect, then
\begin{equation*}
  \btheta_{2n,\model_2}^{cond} = \argmin_{\btheta_{2,\model_2}} \big[ \frac{1}{n} \sum_{i=1}^n \{ A_{2i} - \mu_{2 A0}(\history{2i}) \}^2
  \times \{ \Delta_2(\history{2i}) - \historybasis{2i}(\model_2)^\top \btheta_{2, \model_2} \}^2 \big].
\end{equation*}
As such, we might view these targets as noisy representations of the population-level targets. 

However, these conditional targets also differ from the population targets through a potential dependence on the conditioning set $\designElements{\ell}.$ As discussed in \citet{bujaModelsApproximationsConsequences2019},
the parameters may possibly vary for different realizations of $\designElements{1}$ and $\designElements{2}.$ This work further demonstrates that such issues are avoided when the models being considered contain the true model. Translating this to our setting, such a result might occur if, e.g. $\Delta_2(\history{2}) \equiv \historybasis{2}(\model_2)^{\top} \btheta_{20,\model_2}$ for some $\model_2$ and the randomly-selected model $\hat\model_2$ contains $\model_2.$ In this case, the population parameter $\btheta_{20,\hat\model_2}$ and the conditional target $\btheta_{2n,\hat\model_2}^{cond}$ are equivalent. Conversely, misspecifications could occur either due to 
a poor choice of design variables or an incorrect $\hat\model_2$.

\subsection{Confidence Regions for Conditional Targets}
\label{sec:uposi-conditional}

Regions for these conditional targets may be derived using similar arguments to those in \Cref{sec:uposi-dtr}. 
That is, \(
  \Vert \hat{\bH}_{2n}(\model_{2}) \{  \hat\btheta_{2n, \model_{2}} - \btheta_{2n,\model_2}^{cond} \} \Vert_\infty \leq D_{2n}^{cond}
\) and \(
  \Vert \hat{\bH}_{1n}(\model_{1}) \{ \hat\btheta_{1n,\model_{1} \model_{2}} - \btheta_{1n,\model_1 \model_2}^{cond} \} \Vert_\infty \leq D_{1n,\model_2}^{cond}
\) both hold, where \(
  D_{2n}^{cond} := \| \hat\bG_{2n} - \bG_{2n}^{cond} \|_{\infty}
\) and \(
  D_{1n,\model_2}^{cond} := \| \hat\bG_{1n,\model_2} - \bG_{1n, \model_2}^{cond} \|_{\infty}.
\)
We illustrate the derivation of the first inequality as an example. Taking the first-order equation for $\hat\btheta_{2n, \model_{2}}$ and subtracting \eqref{eq:norm-eq-true-cond-2}, we obtain
\(
  \hat\bH_{2n}(\model_2) \{ \hat\btheta_{2n, \model_{2}} - \btheta_{2n,\model_2}^{cond} \} + \{ \bHOracle{2}(\model_2) - \hat\bH_{2n}(\model_2) \} \btheta_{2n,\model_2}^{cond} = \{ \hat\bG_{2n} - \bG_{2n}^{cond} \}(\model_2).
\)
Suppose for a moment that the second term LHS is negligible (formalized by \Cref{lem:unif-gram-inv} in \Cref{app:lemmas-for-thm-1}). Apply the $\ell_{\infty}$ norm to both sides, drop the second term LHS, and apply the inequality $\| \bv(\model_2) \|_{\infty} \leq \| \bv \|_{\infty}$ for any $p_2-$dimensional vector $\bv$ and submodel $\model_2$ to arrive at the stated inequality. Let $\hat{C}_{2 n}^{cond}(\alpha)$ and $\hat{C}_{1 n,\model_2}^{cond}(\alpha)$ represent the upper $1-\alpha$ quantiles of $D_{2n}^{Gb}$ and $D_{1n,\model_2}^{Gb}$, respectively. The conditional regions are defined as
\begin{align}
  \hat{\Rcal}_{1 n, \model_1 \model_2}^{cond} := \Big\{ \btheta \in \R^{|\model_1|} : \Vert \hat{\bH}_{1n}(\model_{1}) \{ \hat\btheta_{1n,\model_{1} \model_{2}} - \btheta \} \Vert_\infty 
  \leq \hat{C}_{2 n}^{cond}(\alpha) \Big\} \label{eq:uposi-1-cond} 
  \\
  \hat{\Rcal}_{2 n, \model_2}^{cond} := \Big\{ \btheta \in \R^{|\model_2|} : \Vert \hat{\bH}_{2n}(\model_{2}) \{ \hat\btheta_{2n, \model_{2}} - \btheta \} \Vert_\infty 
  \leq \hat{C}_{1 n,\model_2}^{cond}(\alpha) \Big\}. \label{eq:uposi-2-cond}
\end{align}

\subsection{Theory for Conditional Targets}
\label{sec:uposi-cond-theory}

In the first-stage estimation problem, a general model selection technique may involve arbitrary second-stage models, and could possibly choose a least-favorable $\model_2$ in terms of Stage 1 coverage. Dealing with this possibility requires generalizing the UPoSI procedure to hold simultaneously over models from both stages. \Cref{assump:model-selection-consistency} instead requires the second-stage selected model $\hat\model_2$ to ``settle'' in some sense. This allows us to treat the variation added by $\hat\model_2$ to be of a lower order than the variation added by the selection of $\hat\model_1.$ Functionally, this assumption motivates the bootstrap description of $\hat{D}_{1n,\model_2}^{Gb}$ in \eqref{eq:uposi-random-vars-boot}, as the first-stage impact of variation from Stage 2 is assessed only within the context of a particular model---i.e., $\hat\model_2.$

\begin{assumption}
  \label{assump:model-selection-consistency}
  The model  $\hat\model_j$
  selected in Stage $j$ 
  takes values in $\modelspace_j(C_j),j=1,2.$ 
  Moreover, $\hat\model_j$ converges to some $\model_2^* \in \modelspace_2(C_2)$ in the sense \(
    P(\hat\model_2 = \model_2^*) \rightarrow 1.
  \) 
\end{assumption}
\noindent
The remaining regularity \cref{assump:boundedness,assump:bounded-model,assump:rate-assumps,assump:regularity} are available in \Cref{sec:uposi-theory}. The regions contain their respective conditional targets at the appropriate rates, simultaneously over all $\model_1 \in \modelspace_1,~ \model_2 \in \modelspace_2$:

\begin{theorem}[Validity of the conditional regions]
  \label{thm:uposi-cond-valid}
  Under the conditions of \Cref{thm:uposi-coverage},
  \begin{multline}
    \label{eq:uposi-cond-coverage-statements}
    \liminf_{n\rightarrow\infty} \Prob{\left( \btheta_{1n,\hat\model_1 \hat\model_2}^{cond} \in \hat{\Rcal}_{1 n, \hat\model_1 \hat\model_2}^{cond} ~\big|~ \designElements{1} \right)} \geq 1-\alpha
    \text{ and }
    \liminf_{n\rightarrow\infty} \Prob{\left( \btheta_{2n,\hat\model_2}^{cond} \in \hat{\Rcal}_{2 n, \hat\model_2}^{cond} ~\big|~ \designElements{2} \right)} \geq 1-\alpha.
  \end{multline}
\end{theorem}



\citet{kuchibhotlaValidPostselectionInference2020} created coordinate-wise confidence intervals by creating the smallest hyperrectangle that enclosed the UPoSI confidence regions. In general, the confidence interval lengths take a product form \(
  L = \nu \times C,
\)
where $\nu$ is a norm related to the design matrix and $C$ represents the critical value. Define $\be_j$ as a conformable vector of 0 with a 1 in the $j^{th}$ position. In this context, $\nu=\| \be_j^\top \{ \hat\bH_{2n}(\model_2) \}^{-1} \|_1$ and $C=C_{2n}^{cond}(\alpha)$ in Stage 2, and $\nu=\| \be_j^\top \{ \hat\bH_{1n}(\model_1) \}^{-1} \|_1,$ and $C=C_{1n,\model_2}^{cond}(\alpha)$ for the Stage 1. The corresponding half-lengths, $\hat{L}_{2j \model_2}^{cond}$ and $\hat{L}_{1j \model_1 \model_2}^{cond},$ are shown to lead to valid simultaneous coverage events $\hat{\mathcal{I}}_{2n} \defined \bigcap_{j=1,\ldots,|\hat\model_2|} \be_j^\top \left| \hat\btheta_{2n,\hat\model_2} - \btheta_{2n,\hat\model_2}^{cond} \right| \leq \hat{L}_{2j \hat\model_2}^{cond}$ and $\hat{\mathcal{I}}_{1n} \defined \bigcap_{j=1,\ldots,|\hat\model_1|} \be_j^\top \left| \hat\btheta_{1n,\hat\model_1 \hat\model_2} - \btheta_{1n,\hat\model_1 \hat\model_2}^{cond} \right| \leq \hat{L}_{1j \hat\model_1 \hat\model_2}^{cond}$, as stated in the following corollary. 

\begin{corollary}[Validity of the conditional intervals]
  \label{thm:upsoi-ci-valid-cond}
  Under the conditions of \Cref{thm:uposi-coverage}, $\liminf_{n\rightarrow\infty} \Prob(\hat{\mathcal{I} }_{jn} ~|~ \designElements{j}) \geq 1-\alpha$ for $j=1,2$.
\end{corollary}

\section{Simulation Study}
\label{sec:simulation-uposi}

We generated 1000 datasets simulated from various settings. The outcome was modeled as:
\(
  Y = \eta_1(\bX_1) + A_1 \delta_1(\bX_1) + \eta_2(\bX_2) + A_2 \delta_2(\bX_2) + \epsilon
\)
with the treatment for stages $k=1$ and $2$ generated as \(
  A_k \sim Bern[ \text{expit}\{ \psi_A(\bX_k) \} ],
\)
with different functional forms specified for $\eta_1, \eta_2, \delta_1, \delta_2,$ and $\psi_A$. The covariates $\bX_1 \in \R^{p_1}$ and $\bU \in \R^{p_1}$ were generated as $U(-1,1)$ random variables, with $\bX_2 = \bX_1 + \gamma A_1 + \bU$. The errors $\epsilon$ were specified as i.i.d. standard Normal variates. Due to this specification, we may make the identification $\Delta_2 \equiv \delta_2$. However, it is generally difficult to obtain an analytic form for the Stage 1 conditional average treatment effect, $\Delta_{1 \model_2}$. To simplify the problem, we focus on inference for the first stage parameters under settings in which we make the simplification $\delta_2(\bX_2) \equiv 1,$ which automatically satisfies \Cref{assump:regularity}. Further, the coefficient $\gamma$ is set to 0 in this scenario, so that $\bX_2 \independent A_1 ~|~ \bX_1.$ This implies $\Delta_{1 \model_2} \equiv \delta_1,$ since the third and fourth terms of the outcome model have the same expectations conditional on $\{ \history{1}, A_1=a \}$ when $a=0$ and $a=1.$ 

We report the results of six simulation scenarios at four different sample sizes. Scenarios A, B, and C make the previously-described modifications for the first-stage parameters, while scenarios D, E, and F do not. We summarize the configurations in \Cref{tab:sim-setup}. Here, we list different functional forms considered: linear function $f_l$, quadratic function $f_q$, a highly-nonlinear function with interactions $f_n$, and the constant functions 1 and 0.
Because the true $\Delta_{\ell}$ functions in each stage are linear, the conditional targets and unconditional targets are very similar, and in fact are equivalent as long as the selected model contains the true model. In our simulations, the two versions of these parameters had very small differences, but the size of the $\hat{\Rcal}_{\ell}$ intervals were much larger when compared to those of $\hat{\Rcal}_{\ell}^{cond}.$
Specific forms for $f_l,~f_q,$ and $f_n$ are given in \Cref{app:sim-function-forms}.

We constructed UPoSI confidence intervals 
(i.e., with half-length $\hat{L}_{2j \model_2}$ in Stage 2), labeled ``UPoSI'' in the simulation results. We compare these with the selective intervals (``SI'') constructed based on polyhedral inference, as well as naive confidence intervals which ignore the selection process by applying the perturbation bootstrap as if the model were not data-driven (``Naive''). 
To examine the conservativeness of the intervals, we use the median confidence interval length of each method at each sample size; to understand the coverage properties of the intervals, we use the FCR.

The nuisance parameters were estimated with super learning implemented in the \texttt{R} package \texttt{SuperLearner} \citep{polleySuperLearnerSuperLearner2019}. We used two different methods for model selection: the least angle regression (LAR) and forward selection (FS) algorithms with a fixed model size of five, both implemented in the \texttt{selectiveInference R} package. In our simulations, the LAR results were more favorable to SI than those under FS. As such, we focus the conversation by relegating the FS results to \Cref{fig:sim-results-fs} in \Cref{app:sim-results}.
We use the conditional-design framework of \Cref{sec:cond} throughout our simulations; to our knowledge, the corresponding polyhedral inference methods have not been developed for
the case of a random design. The SI methods are used as implemented in the aforementioned \texttt{R} package. 
It should be noted that the SI procedure is only known to control FCR for standard lasso-penalized linear models. To our knowledge, no results exist establishing the properties of SI in the current setting for either stage due to (i) the need to estimate and account for 
unknown nuisance functions and (ii) the highly likely presence of heteroskedasticity in the Stage 1 problem, an inferential complication not yet incorporated into current software.
In contrast, both variations of the UPoSI intervals strongly control FCR through \eqref{eq:uposi-cond-coverage-statements}.

\begin{table}
  \centering
  \caption[Simulation scenarios studied in \Cref{sec:simulation-uposi}.]{
    \label{tab:sim-setup}
    Simulation scenarios studied in \Cref{sec:simulation-uposi}. The Stage 1 scenarios have been modified for simple identification of $\Delta_{1\hat\model_2}.$}
  \begin{tabular}{cccccccc}
    \toprule
    Stage & Scenario & $\eta_1$ & $\delta_1$ & $\eta_2$ & $\delta_2$ & $\psi_A$ & $\gamma$ \\
    \midrule
    & A & $f_q$ & $f_l$ & 0 & 1 & $f_l$ & 0 \\
    & B & $f_q$ & $f_l$ & 0 & 1 & $f_n$ & 0 \\
    \multirow{-3}{*}{1}
    & C & 0 & $f_l$ & 0 & 1 & 0 & 0 \\
    \midrule
    & D & $f_l$ & $f_l$ & $f_q$ & $f_l$ & $f_l$ & 1 \\
    & E & $f_l$ & $f_l$ & $f_q$ & $f_l$ & $f_n$ & 1 \\
    \multirow{-3}{*}{2}
    & F & 0 & 0 & 0 & 0 & 1 & 1 \\
    \bottomrule
  \end{tabular}
\end{table}

\begin{figure}
  \centering
  \begin{subfigure}[b]{\textwidth}
  \includegraphics[width=0.95\textwidth,page=1]{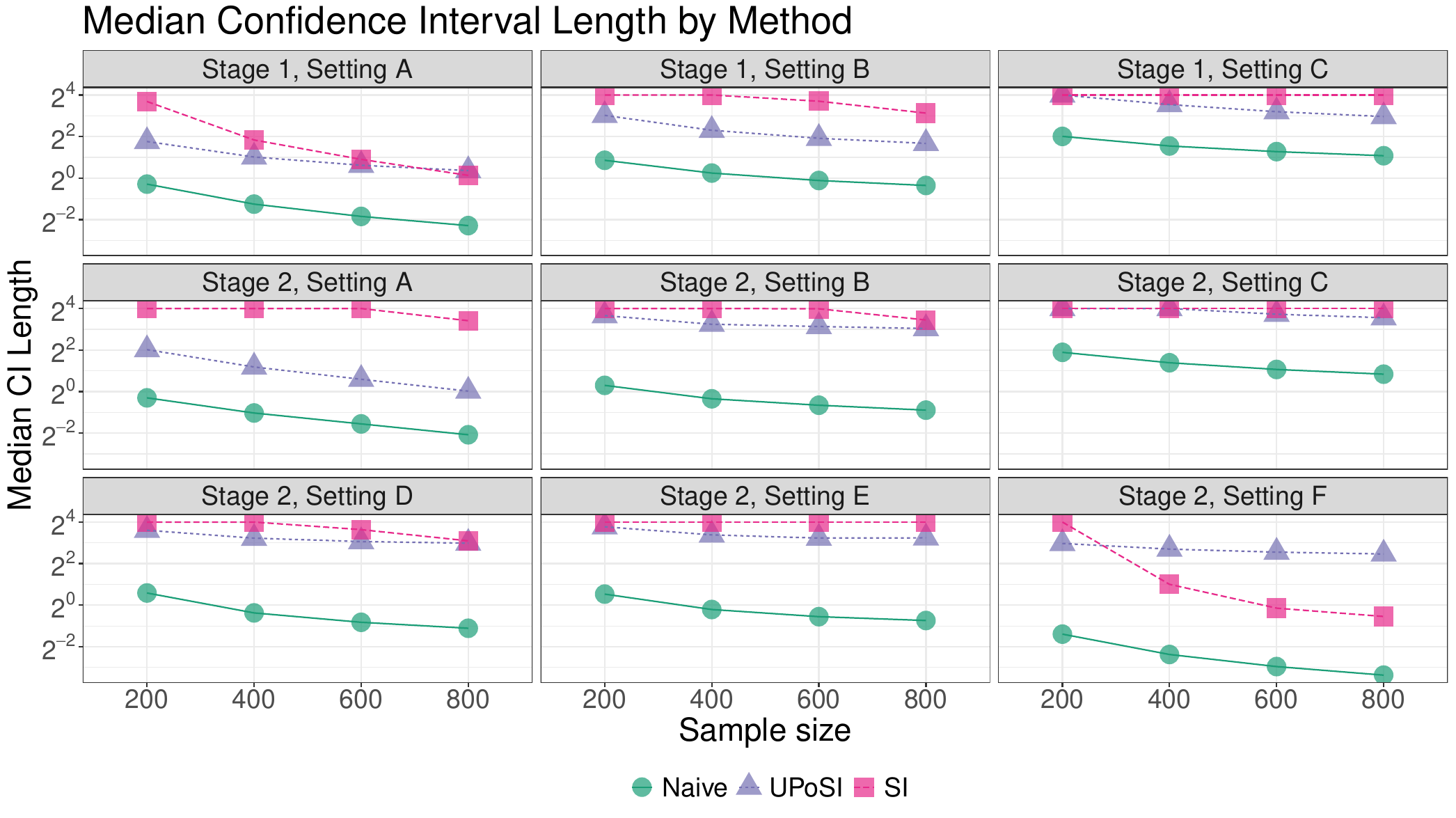}
  \end{subfigure}
  \begin{subfigure}[b]{\textwidth}
  \includegraphics[width=0.95\textwidth,page=2]{simulation-results-lar-conditional.pdf}
\end{subfigure}
\caption[Confidence interval performance for each method: LAR]{
  \label{fig:sim-results-lar}
  Confidence interval performance for each method, grouped by the stage of Robust Q-learning and sample size when using LAR.
  Top: Median confidence interval length; Bottom: False coverage rates.
}
\end{figure}

The results are presented in \Cref{fig:sim-results-lar}.
We can immediately see that the Naive intervals are quite small compared to the other methods. However, they do not control the FCR and are therefore invalid in the presence of selection. The SI intervals tend to be larger than the corresponding conditional UPoSI intervals, although there are a few situations in which the median length of the SI intervals is smaller. Importantly, the SI median interval length also may be infinite; in plotting results, we therefore cap the median lengths at 16 to make visual comparisons possible.
Finally, the UPoSI intervals strongly control FCR. 

\FloatBarrier
\section{Data Application}
\label{sec:extend-analysis}

We apply our method to the Extending Treatment Effectiveness of Naltrexone (ExTENd) trial \citep{murphy2007developing, mckay2010randomized, mckay2011extended, lei2012smart,qian2012dynamic}. The trial dataset is comprised of $n=250$ subjects enrolled in  a sequentially-randomized trial examining the effect of extending naltrexone with behavioral interventions. A more substantive overview can be found in \citet{lei2012smart}. Aligning with the previous analysis, the outcome $Y$ 
represents the proportion of abstinence days during the study, $A_1$ represents either the lenient $(A_1=1)$ or stringent $(A_1=0)$ definition of nonresponse, and $A_2$ 
takes on different meaning depending on the response status of the patient. If the patient is a responder, an alternative of telephone disease management is offered as an extension $(A_2=1)$ or the patient continues naltrexone alone $(A_2=0)$.
If the patient is a nonresponder, an alternative of combined behavioral intervention is offered, with the patient either switching to this alone $(A_2=0)$ or augmenting naltrexone $(A_2=1)$.
We estimated $\mu_{2Y0}$ and $\mu_{1Y \hat\model_2 0}$ using \texttt{xgboost} implemented in the \texttt{mlr3verse} R package with $K=10$ cross-fitting folds and tuned the models via the \texttt{mlr3hyperband} package, using 5 inner folds of cross-validation. 
The propensities were estimated with correctly-specified logistic regression models, stratified by second-stage responder status.

We consider two different analyses of the ExTENd trial: Analysis (I), which considers two-way interactions among variables in each stage; and, Analysis (II), which uses noise variables to augment the available data. For Analysis (I), we created the libraries $\historybasis{1}$ and $\historybasis{2}$ by centering the continuous variables available in each stage. In the second stage, we 
respectively considered two-way interactions of the available history with the response and non-response indicators, hence allowing for separate linear decision boundaries within each response group. Each response group was allowed a separate intercept term to capture effects of different treatments. In the first stage, we considered two-way interactions among all baseline variables. This resulted in a full model of size 20 in Stage 2 and 16 in Stage 1.
We used the minimax concave penalty to select variables due to its oracle properties, removing the overall treatment effect terms from the penalty. Cross-validation provided the preliminary tuning parameter selection, which removed all variables in both stages. The final model was selected by manually reducing the tuning parameter to include a larger model. Since our inferential procedure is simultaneously valid over a broad class of model selection mechanisms, we also compared this model to that provided by lasso. We found that both selection methods generated the same model in this case. Selective Inference does not accommodate the procedure that checks a selected model against an alternative mechanism, so we only compared our method (UPoSI) to the naive bootstrap (Naive). 
%
Analysis (II) was completed similarly to (I) but with certain key differences. To augment the data, we simulated 50 additional variables in $\history{1}$ as i.i.d. standard Normal variates. These have no relationship to the outcome or any additional variables, yet may be spuriously selected by the model selection mechanism, which specifies a model size \textit{a priori}. We use these results to judge the anticonservativeness of each of the previous two inference methods, as well as the selective inference (SI) methodology. 
Additional details regarding the setup are available in \Cref{app:extend-results}.

The Analysis (I) results are summarized in \Cref{tab:extend}.
The resulting Stage 2 tailoring model included the intercept terms for each response group as well as ocds$_0$ among non-responders. For Stage 1, only ocds$_0$ and the intercept were included. 
The Naive bootstrap inference found the first-stage ocds$_0$ term to be significant at the .05 level. After adjusting for selection, our technique finds none of these individual terms significant.
A selection-unadjusted F-test of overall significance gives p-values of .04 and .29 in Stages 1 and 2, respectively, whereas the UPoSI-based test described in \Cref{sec:uposi-discuss} yields p-values of .17 and .86, respectively.
Since this is an exploratory analysis of an existing data set rather than an analysis confirming a pre-existing hypothesis, a researcher might consider using a larger value of $\alpha$ for deciding if these effects deserve further study. For example, allowing $\alpha=.2$ for exploration of novel hypotheses, the researcher may decide to study the use of baseline ocds for determining whether a patient should stringently or leniently be switched from naltrexone to the behavioral intervention.
However, our procedure cautions against declaring these trends as statistically significant at the .05 level, since we determined an interesting hypothesis using the data.
In Analysis (II), the Naive inference method incorrectly identifies 8 out of 10 noise variables as significant. UPoSI and SI both correctly include zero in the interval, with the UPoSI intervals significantly smaller than the SI counterparts. Additional results are available in \Cref{app:extend-results}.

\begin{table}
  \centering
  \caption[Inference on the selected tailoring variables in the ExTENd study.]{
    Inference on the selected tailoring variables in the ExTENd study. Presence of $^*$ after the inference method indicates a 95\% interval that does not include zero.
  }
  \label{tab:extend}
  \begin{tabular}[t]{rlrlrrr}
    \toprule
    \multicolumn{4}{c}{ } & \multicolumn{3}{c}{Confidence Intervals} \\
    \cmidrule{5-7}
    Stage & Variable & Est. & Method & Lower Limit & Upper Limit & Length\\
    \midrule
    &  &  & Naive & -0.095 & 0.153 & 0.248\\
    
    & \multirow{-2}{*}{NR:(Intercept)} & \multirow{-2}{*}{0.030} & UPoSI & -0.170 & 0.229 & 0.400\\
    
    &  &  & Naive & -0.034 & 0.005 & 0.040\\
    
    & \multirow{-2}{*}{NR:ocds$_0$} & \multirow{-2}{*}{-0.015} & UPoSI & -0.074 & 0.045 & 0.119\\
    
    &  &  & Naive & -0.036 & 0.089 & 0.125\\
    
    \multirow{-6}{*}{2} & \multirow{-2}{*}{R:(Intercept)} & \multirow{-2}{*}{0.026} & UPoSI & -0.089 & 0.141 & 0.231\\
    \cmidrule{1-7}
    &  &  & Naive & -0.055 & 0.065 & 0.120\\
    
    & \multirow{-2}{*}{(Intercept)} & \multirow{-2}{*}{0.005} & UPoSI & -0.094 & 0.104 & 0.198\\
    
    &  &  & Naive* & 0.003 & 0.023 & 0.020\\
    
    \multirow{-4}{*}{1} & \multirow{-2}{*}{ocds$_0$} & \multirow{-2}{*}{0.013} & UPoSI & -0.003 & 0.028 & 0.031\\
    \bottomrule
  \end{tabular}
\end{table}

\section{Conclusion}

Accounting for selection in \dtr estimation is important, as demonstrated in our simulations. However, several challenges present themselves when one is allowed to choose the decision rule space. We presented a method based on the Universal Post-Selection Inference framework, which strongly controls the probability that the selective target does not belong to a confidence region. This method was shown to be valid even in the presence of nuisance parameters estimated data-adaptively under some conditions, and a bootstrap procedure was shown to complement this procedure and lead to asymptotically valid inference. We derived an improvement to the UPoSI framework
that results in less conservative confidence intervals for each parameter and demonstrated this improvement in our
simulation studies.

Several challenges still remain. These results were proved in a fixed-dimensional asymptotic regime. We conjecture that many theoretical properties would also hold if the dimensions $\log p_1,~\log p_2$ diverge slower than $n^{c}$ for some $c < 1/2.$ Determining the appropriate rates for cross-fitting may require careful analysis. Additionally, it is unclear if the SI intervals that were compared in this setting admit the cross-fitting estimation strategy. Indeed, it is unclear what alternative methods have provable coverage properties in this setting, either for the conditional or population targets.
Finally, the UPoSI-based intervals can be large compared to naive intervals, especially when the design is considered random rather than fixed. Further refinements focused on improving inferential power are desirable.

\thebib
\section*{Acknowledgements}

This work was dissertation research of the first author during his time at the University of Rochester Department of Biostatistics and Computational Biology.

\newpage

\include{uposi-app-11-9-22}



\end{document}

%% file: uposi-app-11-9-22.tex





\appendix

\begin{center}
  \LARGE{
    \textbf{Supplement to ``\mytitle''}
  }
\end{center}

\setcounter{section}{0}
\renewcommand\thesection{S\arabic{section}}

\setcounter{equation}{0}
\renewcommand\theequation{S\arabic{equation}}

\setcounter{assumption}{0}

  



\section{Review of Robust Q-learning in this Scenario}

\subsection{The Centering Approach in Two Stages}
\label{sec:rql-overview}



The Robust Q-learning algorithm \citep{ertefaieRobustQLearning2021} applies the Robinson-Speckman transformation \citep{robinsonRootnconsistentSemiparametricRegression1988,speckmanKernelSmoothingPartial1988} to Q-learning.
In particular, these authors express a saturated nonparametric model in its centered form:
\begin{equation}
  \label{eq:rob-model-true}
  Y - \mu_{2Y0}(\history{2}) = \left\{ A_2 - \mu_{2A0}(\history{2}) \right\} \Delta_2(\history{2}) + \varepsilon_{2},
\end{equation}
where $\mu_{2Y0}(\obshistory{2}) \defined \Eop(Y ~|~ \history{2} = \obshistory{2})$ and $\mu_{2A0}(\history{2})$ is the propensity defined in \Cref{sec:notation-prob-statement}.
This model exhibits a number of features. The mean functions $\mu_{2A0},~\mu_{2Y0}$ are relatively easy to estimate and are disentangled from the contrast model. If $\mu_{2Y0} \text{ and } \mu_{2A0}$ were known exactly, an analyst could impose a blip model, say $\historybasis[\top]{2} \btheta_{2},$ in place of $\Delta_2(\history{2}) .$ If we additionally require that the model only makes use of the data in $\model_2 \in \modelspace_2(C_2),$ 
we can rewrite \eqref{eq:rob-model-true} as
\begin{align}
  \label{eq:rob-model-y2-model-2}
  Y - \mu_{2Y0}(\history{2}) =& \left\{ A_2 - \mu_{2A0}(\history{2}) \right\} \historybasis{2}(\model_2)^\top \btheta_{20,\model_2} + \epsilon_{2,\model_2},
  \\
  \label{eq:model2-eps}
  \epsilon_{2,\model_2} =& \varepsilon_{2} + \left\{ A_2 - \mu_{2A0}(\history{2}) \right\}
  \times \left\{ \Delta_2(\history{2}) - \historybasis{2}(\model_2)^\top \btheta_{20,\model_2} \right\},
\end{align}
where $\btheta_{20,\model_2}$ is a non-random population-level parameter which is defined in \eqref{eq:norm-eq-true-2} in \Cref{sec:rql-param-oracle}.
In this case, the imposition of a model for 
the second-stage blip function $\Delta_2(\history{2})$ as well as the selection of a particular $\model_2$ induces a change in the residuals. Nonetheless, these residuals are uncorrelated with any function of $\history{2}$ because
$\Eop(\epsilon_{2,\model_2} ~|~ \history{2})=0$ for any $\model_2 \in \modelspace_2(C_2).$

In general, estimation in Q-learning uses  
backwards induction for optimizing the outcome
at each stage.
The developments above show how the 
combination of centering (i.e., through $\mu_{2Y0}(\history{2})$), 
tailoring variable specification (i.e., through $\model_2$)
and an associated linear model $\historybasis{2}(\model_2)^\top \btheta_{20,\model_2}$ for 
$\Delta_2(\history{2})
= \Eop(Y ~|~ \history{2}, A_2=1) - \Eop(Y ~|~ \history{2}, A_2=0)$ modify \eqref{{eq:rob-model-true}} 
in pursuit of an optimal second-stage decision rule.
Similar developments are required to develop
the first-stage decision rule, and rely
on an analogous sequence of modifications
for an appropriately constructed pseudo-outcome.
In particular, we may define a pseudo-outcome
dependent on the second-stage model, or
\begin{align}
  \label{eq:y1-s2-def}
  Y_{1 \model_2} \defined Y + \xi{\left\{ A_2, \historybasis{2}(\model_2) ; \btheta_{20,\model_2} \right\}},
\end{align}
where $\xi{(a_2, \bx; \btheta)} \defined \bx^\top \btheta \left\{ \Ind{(\bx^\top \btheta > 0)} - a_2 \right\}.$
Next, we posit the saturated nonparametric model
$  Y_{1 \model_2} - 
  \mu_{1Y \model_2 0}(\history{1}) = \left\{ A_1 - \mu_{1A0}(\history{1}) \right\} \Delta_{1,\model_2}(\history{1}) + \varepsilon_{1, \model_2},$
where $\mu_{1Y \model_2 0}(\obshistory{1}) = 
\Eop(Y_{1 \model_2} ~|~ \history{1} = \obshistory{1}),$ the error $\varepsilon_{1, \model_2}$ obeys $\Eop(\varepsilon_{1, \model_2} ~|~ \history{1}, A_1)=0,$ and
 \(\Delta_{1,\model_2}(\history{1}) \defined \Eop(Y_{1 \model_2} ~|~ \history{1}, A_1=1) - \Eop(Y_{1 \model_2} ~|~ \history{1}, A_1=0) \) is
 the first-stage blip function.
Finally, arguing similarly to 
\eqref{eq:rob-model-y2-model-2},
we 
obtain the equivalent model
\begin{align}
  Y_{1 \model_2} - \mu_{1Y \model_2 0}(\history{1}) =& \left\{ A_1 - \mu_{1A0}(\history{1}) \right\} \historybasis{1}(\model_1)^\top \btheta_{10,\model_1 \model_2} + \epsilon_{1,\model_1 \model_2},
  \label{eq:rob-model-y1-model-1}
\end{align}
where $\epsilon_{1,\model_1 \model_2}$ has a representation similar to \eqref{eq:model2-eps}
and the model-based residuals satisfy $\Eop(\epsilon_{1,\model_1 \model_2} ~|~ \history{1})=0$ for all $\model_1 \in \modelspace_1(C_1),~ \model_2 \in \modelspace_2(C_2).$ 
Importantly, observe that the second-stage model $\model_2$ directly impacts the pseudo-outcome $Y_{1 \model_2}$ and its corresponding conditional expectation $\mu_{Y1 \model_2 0}(\history{1}),$  and subsequently influence our derivation of the first-stage quantities. This interdependence will
play a critical role when we consider the possibility
of using variable selection at each stage of Q-learning.

\subsection{The Submodel Parameters and their Oracle Estimators}
\label{sec:rql-param-oracle}

In this section, we first explore how estimation might 
proceed if the conditional expectation functions
$\mu_{2A0}(\history{2})$ and
$\mu_{2Y0}(\history{2})$ 
were known. Such estimation might start from the least-squares objective function, which for any $\model_2 \in \modelspace_2$ may be viewed as a random function of an $\R^{|\model_2|}-$valued argument:
\begin{align}
  R_{2 n, \model_2}(\btheta_{2, \model_2}) \defined & \frac{1}{n} \sum_{i=1}^n \big[ Y_i - \mu_{2 Y0}(\history{2i}) 
  - \left\{ A_2 - \mu_{2 A0}(\history{2i}) \right\} \historybasis{2i}(\model_2)^\top \btheta_{2, \model_2} \big]^2.
  \label{eq:oracle-risk-stg-2}
\end{align}
In this expression, $\btheta_{2, \model_2}$ is the
only free parameter 
and has dimension depending on the size of $\model_2.$
Subtracting $R_{2 n, \model_2}(\bzero)$ from both
sides of \eqref{eq:oracle-risk-stg-2} and simplifying, 
we obtain 
\begin{align}
  R_{2 n, \model_2}(\btheta_{2, \model_2}) - R_{2 n, \model_2}(\bzero) =& - 2 \tilde\bG_{2n}(\model_2)^\top \btheta_{2, \model_2}
  + \btheta_{2, \model_2}^\top \bHOracle{2}(\model_2) \btheta_{2, \model_2},
  \label{eq:oracle-risk-decomp-2}
\end{align}
which makes use of the following quantities related to the gradient and Hessian:
  \begin{align}
    \tilde\bG_{2n}(\model_2) &\defined \frac{1}{n} \sum_{i=1}^n \left\{ A_{2 i} - \mu_{2 A 0}(\history{2i})\right\} \left\{ Y_i - \mu_{2 Y 0}(\history{2i}) \right\} \historybasis{2i}(\model_2) 
    \nonumber
    \\
    \bHOracle{2}(\model_2) &\defined \frac{1}{n} \sum_{i=1}^n \left\{ A_{2 i} - \mu_{2 A 0}(\history{2i})\right\}^2 \left( \historybasis{2i}(\model_2) \right)^{\otimes 2}.
    \label{eq:oracle-risk-quant-2}
  \end{align}
The objective function \eqref{eq:oracle-risk-decomp-2}, equivalently \eqref{eq:oracle-risk-stg-2}, is easily
seen to be convex and thus has a minimizer depending on its quadratic behavior.  We 
also see 
that 
the choice of model $\model_2$ serves to subset the full vector $\tilde\bG_{2n} = \tilde\bG_{2n}(\model_2^F)$ and matrix $\bHOracle{2} = \bHOracle{2}(\model_2^F);$ 
hence, the model $\model_2$ only impacts 
\eqref{eq:oracle-risk-decomp-2}
through the indicated subsetting operations.
The representation \eqref{eq:oracle-risk-decomp-2} recovers the normal equations for least-squares estimators by equating its gradient with the $|\model_2|-$dimensional zero vector. 
In particular, the oracle estimator $\tilde\btheta_{2n, \model_2}$ satisfies
\begin{equation}
 \label{eq:norm-eq-oracle-2}
  \bzero = \tilde\bG_{2n}(\model_2) - \bHOracle{2}(\model_2) \tilde\btheta_{2n, \model_2}.
\end{equation}

Taking the expectation of $R_{2 n, \model_2}$ translates the empirical squared-error criterion into its population analogue. The best-fitting population parameter, $\btheta_{20,\model_2},$ minimizes this expected error. As such, we apply similar arguments: taking expectations on both sides of \eqref{eq:oracle-risk-decomp-2},
the parameter $\btheta_{20,\model_2}$ satisfies
\begin{align}
  \label{eq:norm-eq-true-2}
  \bzero = \bG_{20}(\model_2) - \bH_{20}(\model_2) \btheta_{20, \model_2},
\end{align}
where the components $\bG_{20}(\model_2)$ and $\bH_{20}(\model_2)$ of the quadratic function are defined as
\begin{equation}
  \label{eq:expected-risk-quant-2}
  \begin{aligned}
    \bG_{20}(\model_2) &\defined \Eop{ \left[\left\{ A_{2} - \mu_{2 A 0}(\history{2}) \right\}^2 \Delta_{2}(\history{2}) \historybasis{2}(\model_2) \right]}
    \\
    \bH_{20}(\model_2) &\defined \Eop{ \left[ \left\{ A_{2} - \mu_{2 A 0}(\history{2}) \right\}^2 \left( \historybasis{2}(\model_2) \right)^{\otimes 2} \right]}.
  \end{aligned}
\end{equation}
Similarly to before,  
let $\bG_{20} = \bG_{20}(\model_2^F)$
and $\bH_{20} = \bH_{20}(\model_2^F);$
then, analogously to 
\eqref{eq:norm-eq-oracle-2}, 
the model specification $\model_2$ 
subsets these full vectors in
arriving at \eqref{eq:norm-eq-true-2}.

There is an additional complication arising due to the use of a model-dependent pseudo-outcome in the first stage. Specifically, the pseudo-outcome in \eqref{eq:y1-s2-def} depends on $\model_2$ through both $\historybasis{2}(\model_2)$ and the unknown target parameter $\btheta_{20,\model_2}.$ Replacing
$\btheta_{20,\model_2}$
with the oracle estimate 
$\tilde\btheta_{2n, \model_2}$ as the parameter in the blip function used in \eqref{eq:y1-s2-def}, we obtain
an oracle-observed pseudo outcome $\tilde{Y}_{1 \model_2}.$ 
Although 
$\mu_{1Y \model_2 0}(\obshistory{1}) \neq 
\Eop(\tilde Y_{1 \model_2} ~|~ \history{1} = \obshistory{1})$ in general, we can nevertheless posit a useful objective function in Stage 1 for a pair of models $\model_1$ and $\model_2$ that is a function of an $\R^{|\model_1|}-$valued argument:
\begin{align}
  \nonumber
  R_{1 n, \model_1 \model_2}(\btheta_{1, \model_1}) \defined \frac{1}{n} \sum_{i=1}^n \big[ \tilde{Y}_{1 \model_2 i} - \mu_{1 Y \model_2 0}(\history{1i})  - \left\{ A_{1i} - \mu_{1 A0}(\history{1i}) \right\} \historybasis{1i}(\model_1)^\top \btheta_{1, \model_1} \big]^2.
\end{align}
Identical arguments to those used in the second-stage calculations show that
\begin{align}
  R_{1 n, \model_1 \model_2}(\btheta_{1, \model_1}) - R_{1 n, \model_1 \model_2}(\bzero) = - 2 \tilde\bG_{1n, \model_2}(\model_1)^\top \btheta_{1, \model_1} 
 + \btheta_{1, \model_1}^\top \bHOracle{1}(\model_1) \btheta_{1, \model_1},
  \label{eq:oracle-risk-decomp-1}
\end{align}
where
  \begin{multline}
    \tilde\bG_{1n, \model_2}(\model_1) \defined \frac{1}{n} \sum_{i=1}^n \left\{ A_{1i} - \mu_{1 A 0 i}(\history{1i}) \right\} 
    \times
    \left\{ \tilde{Y}_{1 \model_2 i} - \mu_{1 Y \model_2 0}(\history{1i}) \right\} \historybasis{1i}(\model_1)
    ,
    \nonumber
    \\
    \bHOracle{1}(\model_1) \defined \frac{1}{n} \sum_{i=1}^n \left\{ A_{1i} - \mu_{1 A 0 i}(\history{1i}) \right\}^2 \left( \historybasis{1i}(\model_1) \right)^{\otimes 2}.
  \end{multline}
Consequently, the oracle estimator $\tilde\btheta_{1n,\model_1 \model_2}$ 
satisfies 
\begin{align}
  \label{eq:norm-eq-oracle-1} 
  \bzero &= \tilde\bG_{1n,\model_2}(\model_1) - \bHOracle{1}(\model_1) \tilde\btheta_{1n, \model_1 \model_2}.
\end{align}
The dependence of the pseudo-outcomes on $\model_2$
means that the $p_{1}-$dimensional vectors $\tilde\bG_{1n, \model_2} = \tilde\bG_{1n, \model_2}(\model_1^F)$ and $\bG_{10, \model_2} = \bG_{10, \model_2}(\model_1^F)$ also depend on $\model_2$. This is a departure from the second-stage problem, in which the impact of the model $\model_2$ only serves to subset the 
$p_{2}-$dimensional vectors
$\tilde\bG_{2n}$ and $\bG_{20}.$ That is, the second-stage vector $\tilde\bG_{2n}$ is fixed for a given realization of the data and is only subsetted by $\model_2$ in deriving $\tilde\btheta_{2n,\model_2},$ whereas
the first-stage vector $\tilde\bG_{1n,\model_2}$ changes based on $\model_2$ irrespective of any $\model_1$. The latter creates additional complexities for post-selection inference.

The population estimator corresponding to
\eqref{eq:norm-eq-oracle-1} is
considerably harder to characterize due to the fact that
$\mu_{1Y \model_2 0}(\obshistory{1}) 
= \Eop(Y_{1 \model_2} ~|~ \history{1} = \obshistory{1})\neq 
\Eop(\tilde Y_{1 \model_2} ~|~ \history{1} = \obshistory{1}).$ However, the limiting
form of this population estimator has the
same appealing form as \eqref{eq:norm-eq-true-2}
under certain assumptions;
in particular, \Cref{lem:pop-risk-stg1} in \Cref{app:sec-stg1-risk-deriv} establishes that the target of interest 
for a given $\model_1$ and $\model_2$ satisfies
\begin{align}
  \label{eq:norm-eq-true-1}
  \bzero &= \bG_{10,\model_2}(\model_1) - \bH_{10}(\model_1) \btheta_{10, \model_1 \model_2},
\end{align}
where
\begin{equation}
  \label{eq:oracle-expected-risk-quant-1b}
  \begin{aligned}
     \bG_{10,\model_2}(\model_1) &\defined \Eop{ \left[\left\{ A_{1} - \mu_{1 A 0}(\history{1}) \right\}^2 \Delta_{1,\model_2}(\history{1}) \historybasis{1}(\model_1) \right]}
    \\
    \bH_{10}(\model_1) &\defined \Eop{ \left[ \left\{ A_{1} - \mu_{1 A 0}(\history{1}) \right\}^2 \left( \historybasis{1}(\model_1) \right)^{\otimes 2} \right]}.
  \end{aligned}
\end{equation}

\subsection{Theoretical Results}
\label{sec:uposi-theory}


We begin by stating our assumptions. The first assumption is restated from the main paper for convenience:
\begin{assumption}
  The model  $\hat\model_j$
  selected in Stage $j$ 
  takes values in $\modelspace_j(C_j),j=1,2.$ 
  Moreover, $\hat\model_j$ converges to some $\model_2^* \in \modelspace_2(C_2)$ in the sense \(
    P(\hat\model_2 = \model_2^*) \rightarrow 1.
  \) 
\end{assumption}

For the second assumption, the expression \(
  \Lambda_{\ell}(C_{\ell}) \defined \min_{\model_{\ell} \in \modelspace_{\ell}(C_{\ell})} \lambda_{min}(\bH_{\ell 0})
\)
for $\ell = 1,2$ represents the minimal eigenvalue of $\bH_{\ell 0}$ over all subsets of size bounded by $C_{\ell}.$ 
This assumption allows the full matrices $\bH_{\ell 0}$ to be collinear, but requires uniform invertibility over sparse subsets.
\begin{assumption}
  \label{assump:bounded-model}
  For $\ell=1,2$, there exists $0<C_{\ell} < p_{\ell}$ such that
  \(
    \Lambda_1(C_1) \wedge \Lambda_2(C_2) > c_0 > 0.
  \)
\end{assumption}
\noindent
Similarly, we assume that several quantities are uniformly bounded by some constant.
\begin{assumption}
  \label{assump:boundedness}
  The following quantities are uniformly bounded:
  \(
    \| \history{1} \|_{\infty} \vee \| \history{2} \|_{\infty} \vee |\Delta_2(\history{2})| \vee \max_{\model_2 \in \modelspace_2(C_2)}| \Delta_{1,\model_2}(\history{1}) | \leq C.
  \)
\end{assumption}

The next assumption imposes requirements on estimation rates
for nuisance parameters.

\begin{assumption} \label{assump:rate-assumps}
  The cross-fitting setup is used for some fixed $K>1$ and the cross-fitted learners satisfy: 
  (i) \(
    \LtwoPzero[ ]{\hat\mu_{2 Y} - \mu_{2Y0}} = \LtwoPzero[ ]{\hat\mu_{1 Y \hat\model_2} - \mu_{1 Y \hat\model_2 0}} = o_p(1),
  \) (ii) \(
    \LtwoPzero[ ]{\hat\mu_{1 A} - \mu_{1A0}} = \LtwoPzero[ ]{\hat\mu_{2 A} - \mu_{2A0}} = o_p(n^{-1/4}),
  \)
  and (iii) \(
    \LtwoPzero[ ]{\hat\mu_{1 A} - \mu_{1A0}} \LtwoPzero[ ]{\hat\mu_{1 Y \hat\model_2} - \mu_{1 Y \hat\model_2 0}} = o_p(n^{-1/2}).
  \)
  %
\end{assumption}
\noindent
We draw attention to the fact that in a randomized trial setting, the propensity model is known. Consequently, (ii) is automatically satisfied, as the propensities can either be estimated with certainty, or with the parametric rate $O_p(n^{-1/2}).$ In this case, (iii) follows by (i), showing that application to a randomized trial only requires consistency of the outcome models. This is a relatively mild condition satisfied by several types of nonparametric learners \citep{zheng2011cross,chernozhukovDoubleDebiasedMachine2018f}.

Finally, we require an assumption ensuring regular behavior. An alternative assumption on the unknown distribution \citep[][Corollary 1]{ertefaieRobustQLearning2021} or method modifications \citep[e.g.,][]{laberDynamicTreatmentRegimes2014} also address the nonregularity occurring in the first stage. 

\begin{assumption}
  \label{assump:regularity}
  The targeted rule in Stage 2 yields unique treatment decisions almost surely, in the sense that $\history{2}$ satisfies
  \(
    P\left( \historybasis{2}(\model_2)^\top \btheta_{20,\model_2^*} = 0 \right) = 0.
  \)
\end{assumption}
\noindent


Our proof includes the stronger statement that the second probability statement is valid simultaneously over all $\hat\model_2 \in \modelspace_2(C_2)$, and the first over all $\hat\model_1 \in \modelspace_1(C_1)$. Hence, these regions are asymptotically simultaneously valid over models in the sense explored by \citet{kuchibhotlaValidPostselectionInference2020}. Although the proof relies on \Cref{assump:model-selection-consistency}, 
our simulation results demonstrate that the effect of second-stage selection plays a minor role in the settings we study.

These probability statements involve covering the full vector of the parameter, essentially treating each element of the selected parameter as part of a family and controlling the Family-Wise Coverage Rate. More information on this coverage criterion is discussed in \Cref{app:equiv-sel-sim} in \eqref{eq:fvfcp-app}. One consequence is that the False Coverage Rate (FCR) is controlled by our proposed regions, which follows from standard arguments regarding Family-Wise Error Rate and False Discovery Rate.

Next, we show that the perturbation bootstrap approach to estimating the quantiles of the distribution is valid.
To do so, recall the definitions of the UPoSI random variables in \eqref{eq:uposi-random-vars} as well as their bootstrap analogues in \eqref{eq:uposi-random-vars-boot}. The proof of this theorem appears in \Cref{sec:pf-thm-uposi-bootstrap-valid}.

\begin{theorem}[Validity of the perturbation bootstrap]
  \label{thm:uposi-bootstrap-valid}
  Under the conditions of \Cref{thm:uposi-coverage}, the following distributional approximations hold:
  \begin{multline*}
    \sup_{a, b \geq 0} \bigg| \Prob \left( \sqrt{n} \hat{D}_{2n}^G \leq a, ~ \sqrt{n} \hat{D}_{2n}^H \leq b \right) 
    - \Prob \left( \sqrt{n} \hat{D}_{2n}^{Gb} \leq a, ~ \sqrt{n} \hat{D}_{2n}^{Hb} \leq b \right) \bigg| \rightarrow 0 \\
    \sup_{a, b \geq 0} \bigg| \Prob \left( \sqrt{n} \hat{D}_{1n,\hat\model_2}^G \leq a, ~ \sqrt{n} \hat{D}_{1n}^H \leq b \right) 
    - \Prob \left( \sqrt{n} \hat{D}_{1n,\hat\model_2}^{Gb} \leq a, ~ \sqrt{n} \hat{D}_{1n}^{Hb} \leq b \right) \bigg| \rightarrow 0.
  \end{multline*}
\end{theorem}
\noindent
This distributional approximation is useful for obtaining the appropriate quantiles for the UPoSI regions. For example, the RHS of the inequality in \eqref{eq:uposi-2-dagger-star} depends on
\(
  C^{G}_{2n}(\alpha) + C_{2n}^{H}(\alpha) \Vert\hat\btheta_{2n, \model_{2}}\Vert_1.
\)
Using \Cref{thm:uposi-bootstrap-valid}, fix some $c > 0$ and set $a=c, ~b=\Vert\hat\btheta_{2n, \hat\model_{2}}\Vert_1 c$ inside the sup norm. Then we may treat the quantity $\hat{D}_{2n} \defined \hat{D}_{2n}^{H} + \Vert\hat\btheta_{2n, \hat\model_{2}}\Vert_1 \hat{D}_{2n}^{G}$ as a univariate random variable and be assured that its quantiles are uniformly approximated by $\hat{D}_{2n}^b \defined \hat{D}_{2n}^{Hb} + \Vert\hat\btheta_{2n, \hat\model_{2}}\Vert_1 \hat{D}_{2n}^{Gb}$.

Finally, the preceding results imply that the UPoSI intervals are valid post-selection, in that they satisfy a similar simultaneous coverage condition to \eqref{eq:uposi-coverage-statements}. 

\begin{corollary}[Validity of the confidence intervals]
  \label{thm:upsoi-ci-valid}
  Under the conditions of \Cref{thm:uposi-coverage}, $\liminf_{n\rightarrow\infty} \Prob(\hat{\mathcal{I}}_{2n}) \geq 1-\alpha$ and $\liminf_{n\rightarrow\infty} \Prob(\hat{\mathcal{I}}_{1n}) \geq 1-\alpha$ where $\hat{\mathcal{I}}_{2n} \defined \bigcap_{j=1,\ldots,|\hat\model_2|} \left| \be_j^\top (\hat\btheta_{2n,\hat\model_2} - \btheta_{20,\hat\model_2}) \right| \leq \hat{L}_{2j \hat\model_2}$ and $\hat{\mathcal{I}}_{1n} \defined \bigcap_{j=1,\ldots,|\hat\model_1|} \left| \be_j^\top (\hat\btheta_{1n,\hat\model_1 \hat\model_2} - \btheta_{10,\hat\model_1 \hat\model_2}) \right| \leq \hat{L}_{1j \hat\model_1 \hat\model_2}$. 
\end{corollary}
\noindent An interesting feature of these confidence intervals is that
the family of constructed intervals are simultaneously correct at the coverage rate $1-\alpha.$ An analyst might view the method as an alternative to Scheffé, except that these apply to individual coefficients rather than contrasts, and naturally accommodate model selection. This corollary is stated without proof, as it generally follows the argumentation in \citet{kuchibhotlaValidPostselectionInference2020}.


A natural question is whether the bootstrap distributions of $\hat{D}_{2n}^{Gb}$ and $\hat{D}_{1n,\model_2}^{Gb}$ are valid approximations to those of $D_{2n}^{cond}$ and $D_{1n,\model_2}^{cond}$, respectively, on the relevant conditioning sets.
This is the case, as stated in the following theorem, proved in \Cref{sec:pf-thm-upsoi-cond}.
\begin{theorem}[Conditional perturbation bootstrap]
  \label{thm:uposi-cond}
  Under the conditions of \Cref{thm:uposi-coverage}, the following distributional results hold:
  \begin{align}
    \sup_{a \geq 0} \bigg| \Prob \left( \sqrt{n} D_{2n}^{cond} \leq a ~\big|~ \designElements{2} \right) 
    \label{eq:uposi-mult-boot-cond-2}
    - \Prob \left( \sqrt{n} \hat{D}_{2n}^{Gb} \leq a ~\big|~ \designElements{2} \right) \bigg| \rightarrow 0
    \\
    \sup_{a \geq 0} \bigg| \Prob \left( \sqrt{n} D_{1n,\hat\model_2}^{cond} \leq a ~\big|~ \designElements{1} \right) 
    \label{eq:uposi-mult-boot-cond-1}
    - \Prob \left( \sqrt{n} \hat{D}_{1n,\hat\model_2}^{Gb} \leq a ~\big|~ \designElements{1} \right) \bigg| \rightarrow 0.
  \end{align}
\end{theorem}

\section{Simulation Addedum}
\label{app:sim-add}

\subsection{Simulation Description}
\label{app:sim-function-forms}

Each of these functions only depend on the first five elements (at most) of the argument. Let $\bX\subj[1]$ represent the first element, $\bX\subj[2]$ the second, and so on. Without loss of generality, let these functions be defined as functions of a vector $\bX \in \R^5,$ understanding that the higher-dimensional functions will be mapped to these functions through the first five coordinates. Then letting $\bbeta = (2,2,1,.1,.1)$, we define
\begin{align*}
  f_l(\bX) &= \bX^\top \bbeta \\
  f_q(\bX) &= 0.5\left\{ \bX^\top \diag(\bbeta) \bX + \bX^\top \bbeta - 2 \right\} \\
  f_n(\bX) &= 0.5 \sin\Big[ \pi \bX\subj[1] \bX\subj[2] \Big] + 2 \Big[\bX\subj[3] - .5\Big]^2 - 1.
\end{align*}
The ``constant functions'' are represented by 0 and 1, taking the stated value over their ranges.

\subsection{Stage 1 Blip Function}

Under the simulation model, there is some subtletly in calculating $\Delta_{1 \model_2}(\history{1})$ for each $\model_2.$ We restate the model below:
\begin{align*}
  Y = \eta_1(\bX_1) + A_1 \delta_1(\bX_1) + \eta_2(\bX_2) + A_2 \delta_2(\bX_2) + \epsilon \\
  A_k \sim Bern[ \text{expit}\{ \psi_A(\bX_k) \} ] ~ \text{ for } k=1,2,
\end{align*}
The $\Delta_{1 \model_2}(\history{1})$ function is identified as the contrast in $Y_{1 \model_2}$, given $\history{1}$ and two different values of $A_1$:
\begin{equation*}
  \Delta_{1 \model_2}(\history{1}) = \Eop(Y_{1 \model_2} ~|~ \history{1}, A_1 = 1) - \Eop(Y_{1 \model_2} ~|~ \history{1}, A_1 = 0).
\end{equation*}
Using the simulation model for $Y$ and the definition of $Y_{1 \model_2},$ we can write
\begin{align}
  \label{eq:app-sim-1}
  Y_{1 \model_2} &= \eta_1(\bX_1) + A_1 \delta_1(\bX_1) + \eta_2(\bX_2) + \Ind{\left\{ \historybasis{2}(\model_2)^\top \btheta_{20,\model_2} > 0 \right\}} \historybasis{2}(\model_2)^\top \btheta_{20,\model_2} + \epsilon \\
  \label{eq:app-sim-2}
  &+ A_2 \left\{ \delta_2(\bX_2) - \historybasis{2}(\model_2)^\top \btheta_{20,\model_2} \right\}.
\end{align}
If the conditional independence $\bX_2 \independent A_1 | \bX_1$ holds ($\gamma=0$), then 
\begin{equation*}
  \Eop{\left\{ \text{\eqref{eq:app-sim-1}} ~|~ \history{1}, A_1=1 \right\}} - \Eop{\left\{ \text{\eqref{eq:app-sim-1}} ~|~ \history{1}, A_1=0 \right\}} = \delta_1(\bX_1).
\end{equation*}
We wish to derive conditions to make a similar contrast that will make
\eqref{eq:app-sim-2} equal to zero; this would imply that $\Delta_{1, \model_2}(\history{1}) \equiv \delta_1(\bX_1).$
In particular, note that
\begin{multline*}
  \Eop{\left\{ \text{\eqref{eq:app-sim-2}} ~|~ \history{1}, A_1=a \right\}} = \Eop{\left[ \Eop\left\{ \text{\eqref{eq:app-sim-2}} ~|~ \history{1}, A_1=a, \bX_2 \right\} ~|~ \history{1}, A_1=a \right]} \\
  = \Eop{\left[ \left\{ \delta_2(\bX_2) - \historybasis{2}(\model_2)^\top \btheta_{20,\model_2} \right\} \Eop\left( A_2 ~|~ \history{1}, A_1=a, \bX_2 \right) ~|~ \history{1}, A_1=a \right]}.
\end{multline*}
Under the simulation setup, $\Eop\left( A_2 ~|~ \history{1}, A_1=a, \bX_2 \right) = \Eop\left( A_2 ~|~ \bX_2 \right).$ Consequently, if the model $\model_2$ is correctly-specified, then $\delta_2(\bX_2) - \historybasis{2}(\model_2)^\top \btheta_{20,\model_2}=0.$ When $\delta_2(\bX_2) \equiv 1,$ then any model containing the intercept is correctly-specified.

An alternative is if $\historybasis{2}(\model_2)$ only depends upon $\bX_2$ under the previous conditional independence. In that case, all the random variation inside the conditional expectation depends only on $\bX_2,$ so that using the independence assumption results in \[
  \Eop{\left\{ \text{\eqref{eq:app-sim-2}} ~|~ \history{1}, A_1=1 \right\}} - \Eop{\left\{ \text{\eqref{eq:app-sim-2}} ~|~ \history{1}, A_1=0 \right\}} = 0.
\]
A way to enforce this is through using a Markov modeling assumption, so that the transformation $\historybasis{2}$ is a function only of $\bX_2.$ This does not restrict $\delta_2.$

\FloatBarrier

\newpage
\subsection{Additional Simulation Results}
\label{app:sim-results}

\begin{figure}[!ht]
  \centering
  \resizebox*{!}{0.4\textheight}{
  \begin{subfigure}[b]{\textwidth}
  \includegraphics[width=0.95\textwidth,page=1]{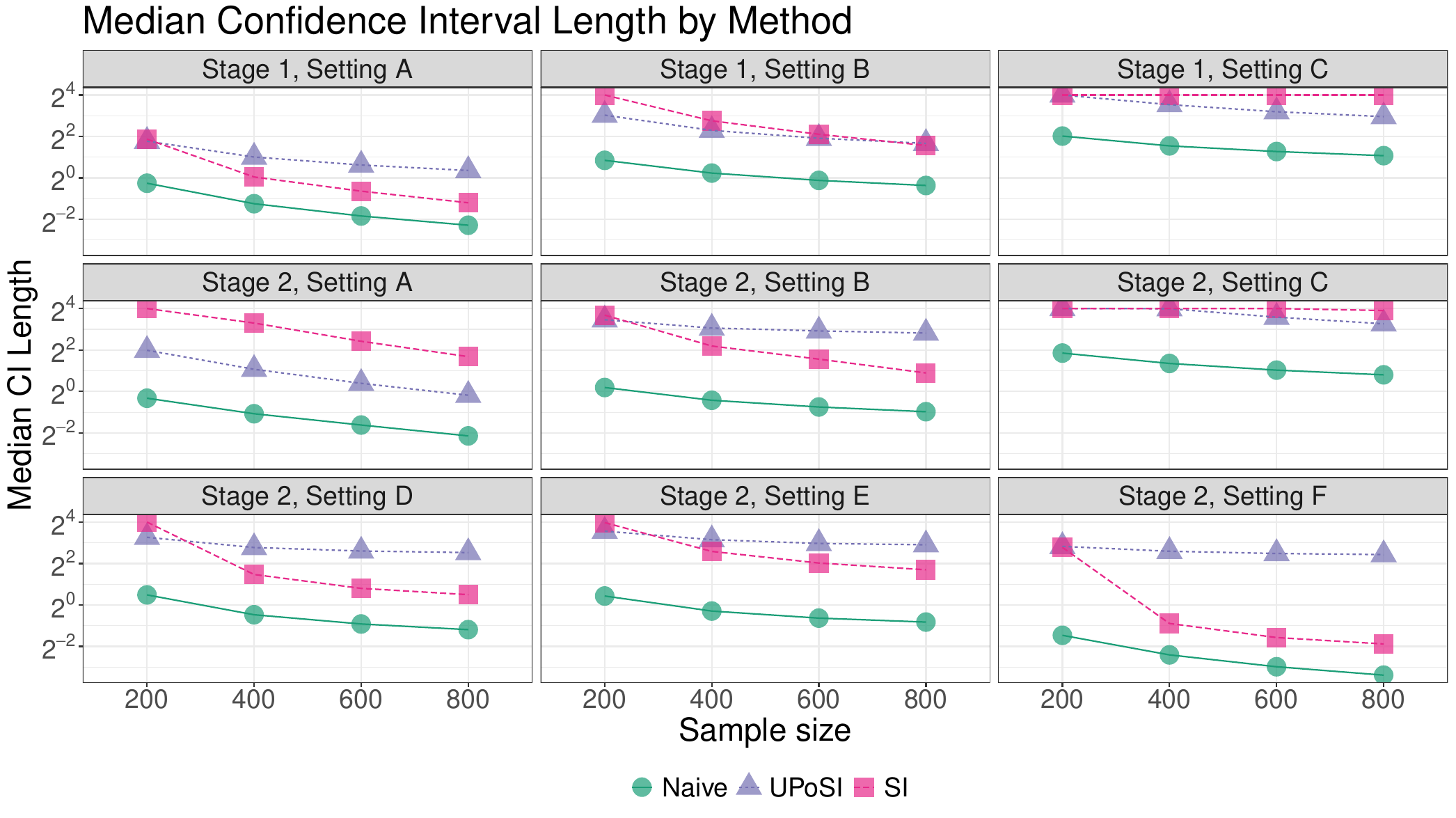}
  \end{subfigure}
}\\
 \resizebox*{!}{0.4\textheight}{
  \begin{subfigure}[b]{\textwidth}
  \includegraphics[width=0.95\textwidth,page=2]{simulation-results-fs-conditional.pdf}
\end{subfigure}
}
\caption[Confidence interval performance for each method: FS]{
  \label{fig:sim-results-fs}
  Confidence interval performance for each method, grouped by the stage of Robust Q-learning and sample size when using FS.
  Top: Median confidence interval length; Bottom: False coverage rates.
}
\end{figure}

\FloatBarrier
\section{Additional ExTENd Study Results}
\label{app:extend-results}

For Analysis (I), we generated $\historybasis{1}$ using main effects and two-way interactions among the baseline variables. All continuous variables were pre-transformed by centering the column by its mean. A similar procedure was used to create $\historybasis{2}$, which included two-way interactions of the responder and nonresponder group indicator variables with the remaining variables.

For Analysis (II), we generated $\historybasis{1} = \history{1}$ and $\historybasis{2} = (R \bZ^{\top}, (1-R) \bZ^{\top})^{\top}$, where $R$ represents the indicator for the response prior to Stage 2 re-randomization and $\bZ$ represents all of the variables in $\history{2}$ other than $R$.
In total, that leads to a full model of size 132 in the second stage and 56 in the first. To allow the selective inference comparison, we select models using least-angle regression. In Stage 1, we use a fixed model size of 5 for the selected tailoring variables; in Stage 2, we select three variables separately among responders and nonresponders, letting us interpret these selections as interactions. The final inferential procedure for SI estimates the error variance using the pooled data.

\Cref{fig:extend-aug} shows the Naive inference method tend finds all but 2 of the noise terms to be significant at the .05 level, showing the deficiency of naive inference. Compared to SI, our proposed method gives much smaller confidence interval lengths. None of the selected terms are found to be seignificant after controlling for selection. 

\begin{figure}
  \centering
  \includegraphics[width=0.8\linewidth]{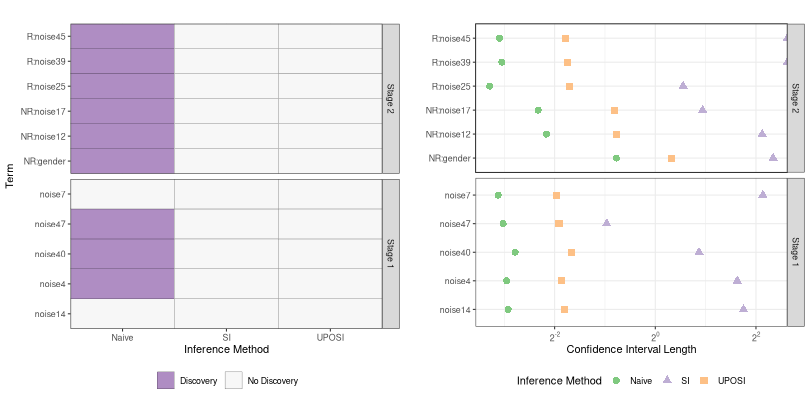}
  \caption[Performance on augmented data.]{Comparison of naive inference, selective inference (SI), and the proposed method (UPOSI) on the ExTENd data. Left: Discovery of nonzero coefficients. Right: Confidence interval length.}
  \label{fig:extend-aug}
\end{figure}

\begin{table}
  \centering
  \includegraphics[width=0.6\linewidth]{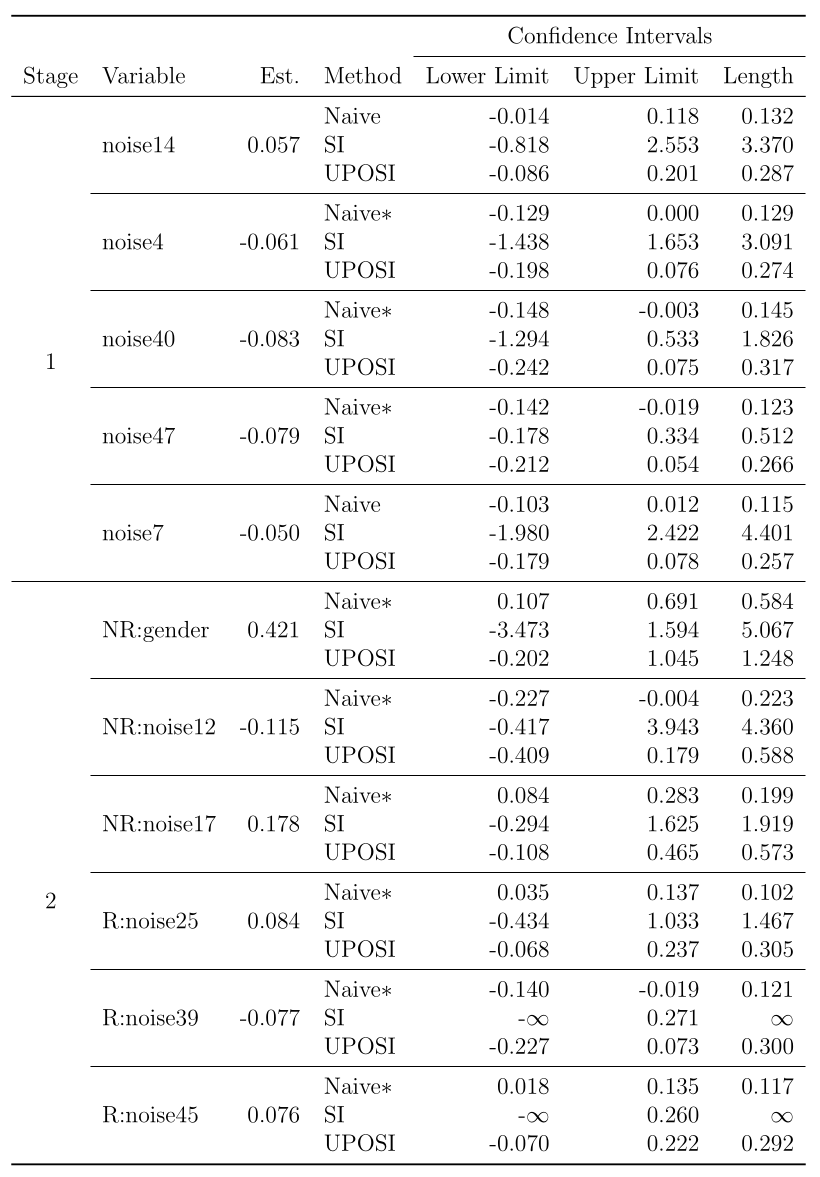}
  \caption[Inference on the selected tailoring variables in the augmented ExTENd study.]{
    Inference on the selected tailoring variables in the augmented data from the ExTENd study. Presence of $^*$ after the inference method indicates a 95\% interval that does not include zero.
  }
  \label{tab:extend-aug}
\end{table}

\FloatBarrier
\section{Equivalence of Post-Selection and Simultaneous Inference}
\label{app:equiv-sel-sim}

This result was proved in \citet{kuchibhotlaValidPostselectionInference2020}. We restate the result here and give a bit of exposition around how this fits in with the additional complications arising from multiple stages of model decisions.
The post-selection inference problem is defined for general parameters and confidence region construction methods, although it could also be viewed through the lens of hypothesis testing; we mostly focus on the former throughout. 

Let us consider a coverage event for a general parameter, which may be allowed to vary with $n$, represented by $\bar\btheta_{n,\model},$ where $\bar\btheta_{n,\model}$ is subscripted by $\model$ in order to reflect that the definition of the parameter varies by the index $\model$ taking its values in $\modelspace.$ When this index is chosen from $\modelspace$ based on the data, we obtain some $\hat\model$ and then focus on a particular $\bar\btheta_{n,\hat\model}.$ The goal of post-selection inference is to construct a confidence region $\bar{\Rcal}_{n, \hat\model}$ which contains this selected parameter at the desired confidence level. Mathematically, this may be stated:
\begin{equation}
  P\Big( \bar\btheta_{n,\hat\model} \notin \bar{\Rcal}_{n, \hat\model} \Big) \leq \alpha.
  \label{eq:fvfcp-app}
\end{equation}

The following theorem, proved in \citet{kuchibhotlaValidPostselectionInference2020}, establishes the equivalence between this criterion and a simultaneous inference result.
\begin{theorem}{\citep[Thm.~3.1]{kuchibhotlaValidPostselectionInference2020}}
  \label{thm:uposi-sim-sel}
  Let $\{ \bar{\Rcal}_{n, \model} \}$ be a family of confidence regions for a family of parameters $\bar\btheta_{n,\model}$, with both the regions and parameters indexed by the model $\model \in \modelspace.$ Let $\hat\model$ be a data-dependent model taking values almost surely in $\modelspace$. Then \eqref{eq:fvfcp-app} is equivalent to
  \begin{equation}
      P\left( \underset{{\model \in \modelspace}}{\bigcap} \bar\btheta_{n,\model} \in \bar{\Rcal}_{n, \model} \right) \geq 1-\alpha.
    \label{eq:sim-and-sel-inf}
  \end{equation}
\end{theorem}

Let us consider how this applies to the population targets defined in \Cref{sec:rql-param-oracle} in each of the two stages in Robust Q-learning.
In the second stage, the set of parameters is given by $\btheta_{20,\model_2}$ as defined in \eqref{eq:norm-eq-true-2}. These parameters are indexed by $\model_2$ taking values in $\modelspace_2$, and the parameters are fixed with $n$. The post-selection parameter is then $\bar\btheta_{n,\hat\model} \equiv \btheta_{20,\hat\model_2}.$ By \Cref{assump:model-selection-consistency}, the random parameters are chosen from $\modelspace_2(C_2),$ so that we can set $\modelspace \equiv \modelspace_2(C_2).$

In the first stage, the set of parameters is $\btheta_{10,\model_1 \model_2},$ as defined in \eqref{eq:norm-eq-true-1}. This parameter is also fixed with $n$, although there are now two models involved: $\model_1 \in \modelspace_1$ and $\model_2 \in \modelspace_2.$ The post-selection parameter is $\btheta_{10,\hat\model_1 \hat\model_2}.$ Using \Cref{assump:model-selection-consistency}, the pair of models $(\hat\model_1, \hat\model_2)$ are contained in $\modelspace_1(C_1) \times \modelspace_2(C_2)$ almost surely.
Consequently, we can set $\hat\model \defined (\hat\model_1, \hat\model_2)$ and $\modelspace \defined \modelspace_1(C_1) \times \modelspace_2(C_2),$ which gives a parameter $\bar\btheta_{n,\model} \equiv \btheta_{10,\model_1 \model_2}$ of the required form.

Also by \Cref{assump:model-selection-consistency}, the pair of models $(\hat\model_1, \hat\model_2)$ take values in $\modelspace^{\dagger} \defined \modelspace_1(C_1) \times \{ \model_2^* \}$ with probability converging to one. We can choose $\modelspace \equiv \modelspace^{\dagger}$ to obtain an asymptotic version of these probability statements that hold in in the limit. To see why this is the case, simply apply the following relationships:
\[
  0 \leq P(A) - P(A \cap B) = P(A \cap B^c) \leq P(B^c),
\]
to the events $A \equiv \left\{ \bar\btheta_{n,\hat\model} \notin \bar{\Rcal}_{n, \hat\model} \right\}$ and $B \equiv \left\{ \hat\model \in \modelspace^{\dagger} \right\}$. The rightmost inequality vanishes.

\section{Prerequisites for the Proofs}

\subsection{Some Additional Notation}
\label{app:additional-notation}

Let the ``positive part'' function be defined as $(a)_+ \defined a ~ {\Ind(a > 0)}$.
For a real matrix $\bA$, let $\| \bA \|_{2,2}$ represent its maximal singular value. This definition is used to make it clear that it is the operator norm mapping $\ell_2$ to $\ell_2$---i.e., for any vector $\bv \in \R^n$ and $\bA \in \R^{m \times n}$, the norm satisfes $\| \bA \bv \|_2 \leq \| \bA \|_{2,2} \| \bv \|_{2}.$ We will also make use of the $L_2(P_0)$ norm for vector-valued random functions: if $\bh = (h_1, \ldots, h_d)^\top$ is a vector-valued random function, then let $\LtwoPzero[ ]{\bh} = \max_{j=1,\ldots,d} \LtwoPzero[ ]{h_j}.$

Unless otherwise specified, the notation \(o_p(1)\) will be used to represent ``little-o in probability'' with respect to the $\ell_{\infty}$ norm for vector-valued random variables. That is, for vectors $\bV, \bU$, we use the expression \(
  \bU = \bV + o_p(r_n)
\)
to represent the statement \(
  r_n^{-1} \| \bU - \bV \|_{\infty} \CiP 0. 
\)
Similarly, an expression of the form $\bU = \bV + O_p(r_n)$ will represent \(
  r_n^{-1} \| \bU - \bV \|_{\infty} = O_p(1).
\)
Some random quantities defined in \Cref{sec:cf-perturb-risk} depend on both the random sample from $P_0$ as well as random multipliers $G$ from a bootstrap distribution $P_{\omega}$. The notation $O_{p^*}(1)$ and $o_{p^*}(1)$ will represent similar probability statements with respect to the product measure of the observed data and random multipliers, $P_0 \times P_{\omega}.$

We also recall that the function $\| \cdot \|_{\infty}$ for matrices is not itself a proper matrix norm; however, for a $m \times m$ matrix $\bA$, the quantity $m \| \bA \|_{\infty}$ is a proper matrix norm; see (5.6.0.4) in  \citet{horn2012matrix}. We also have the property $\| \bA \|_{\infty} \leq \| \bA \|_{2,2} \leq m \| \bA \|_{\infty}.$ This fact will be useful in some of the future arguments.

To simplify the notation, we will subscript conditional expectations and their estimates by $i$ to represent the evaluation of the function at the $i^{th}$ value. For example, if $i \in \Ik$, we will write $\mu_{2Y0}(\history{2i}) - \hat\mu_{2Y}(\history{2i}; \crossfitDataC)$ in the more compact form $\mu_{2Y0i} - \hat\mu_{2Yi}$ and similarly we will write $\mu_{1A0}(\history{1i}) - \hat\mu_{1A}(\history{1i}; \crossfitDataC)$ as $\mu_{1A0i} - \hat\mu_{1Ai}$. It will be understood that these function estimates are the product of cross-fitting, although this fact will not be explicitly notated except when necessary. For example, we will often write
\[
  \frac{1}{n} \sum_{i=1}^n (\mu_{1A0i} - \hat\mu_{1Ai})
\]
instead of the double-sum over both $k=1,\ldots,K$ and $i \in \Ik.$ The arguments involving cross-fitting more directly will draw attention to this nested structure.

Finally, several quantities were defined in the main text. To provide a more concrete reference, we list several of these quantities and relationships here.
\begin{equation} \label{eq:app-defined-quantities}
\begin{aligned}
  \bH_{\ell 0} =& \Eop{ \left[ \left\{ A_{\ell} - \mu_{\ell A 0}(\history{\ell}) \right\}^2 \left( \historybasis{\ell} \right)^{\otimes 2} \right]}
  \\
  \bHOracle{\ell} =& \frac{1}{n} \sum_{i=1}^n \left( A_{\ell i} - \mu_{\ell A 0 i}\right)^2 \left( \historybasis{\ell} \right)^{\otimes 2} 
  \\
  \hat\bH_{\ell n} =& \frac{1}{n} \sum_{i=1}^n \left( A_{\ell i} - \hat\mu_{\ell A i} \right)^2 \left( \historybasis{\ell} \right)^{\otimes 2} 
  \\
  \bG_{2 0} =& \Eop{\left[\left\{ A_{2} - \mu_{2 A 0}(\history{2}) \right\}^2 \Delta_{2}(\history{2}) \historybasis{2} \right]} 
  \\
  \tilde{\bG}_{2 n} =& \frac{1}{n} \sum_{i=1}^n (A_{2i} - \mu_{2 A0i}) (Y_i - \mu_{2 Y0i}) \historybasis{2i} 
  \\
  \hat{\bG}_{2 n} =& \frac{1}{n} \sum_{i=1}^n (A_{2i} - \hat\mu_{2 Ai}) (Y_i - \hat\mu_{2 Yi}) \historybasis{2i} 
  \\
  \bG_{1 0 \model_2} =& \Eop{ \left[\left\{ A_{1} - \mu_{1 A 0}(\history{1}) \right\}^2 \Delta_{1,\model_2}(\history{1}) \historybasis{1} \right]} 
  \\
  \tilde{\bG}_{1 n \model_2} =& \frac{1}{n} \sum_{i=1}^n (A_{1i} - \mu_{1 A0i}) (\tilde{Y}_{1 \model_2 i} - \mu_{1 Y \model_2 0i}) \historybasis{1i}
  \\ 
  \hat{\bG}_{1 n \model_2} =& \frac{1}{n} \sum_{i=1}^n (A_{1i} - \hat\mu_{1 Ai}) (\hat{Y}_{1 \model_2 i} - \hat\mu_{1 Y \model_2 i}) \historybasis{1i}
\end{aligned}
\end{equation}
Each of the terms involving sums also have a perturbation bootstrap version, which is superscripted with $b$ and involves $\bootWgt_i$ multiplying the $i^{th}$ term in the sum. For example,
\[
  \hat{\bG}_{2 n}^b = \frac{1}{n} \sum_{i=1}^n \bootWgt_i (A_{2i} - \hat\mu_{2 Ai}) (Y_i - \hat\mu_{2 Yi}) \historybasis{2i}.
\]
The target parameters in the first and second stages, respectively, are
\begin{equation}
  \label{eq:app-defined-targets}
  \begin{aligned}
    \btheta_{10,\model_{1} \model_{2}} &= \bH_{10}(\model_{1})^{-1} \bG_{10,\model_{2}}(\model_{1}) \\
    \btheta_{20,\model_2} &= \bH_{20}(\model_2)^{-1} \bG_{20}(\model_2)
  \end{aligned}
\end{equation}
and have estimators leveraging the quantities in \eqref{eq:app-defined-quantities}:
\begin{equation}
  \label{eq:app-defined-ests}
  \begin{aligned}
    \tilde{\btheta}_{1 n, \model_{1} \model_{2}} &= \bHOracle{1}(\model_{1})^{-1} \tilde{\bG}_{1 n, \model_{2}}(\model_{1}) \\
    \hat{\btheta}_{1 n, \model_{1} \model_{2}} &= \hat{\bH}_{1 n}(\model_{1})^{-1} \hat{\bG}_{1 n, \model_{2}}(\model_{1}) \\
    \tilde\btheta_{2n,\model_2} &= \bHOracle{2}(\model_2)^{-1} \tilde\bG_{2n}(\model_2)
    \\
    \hat\btheta_{2n,\model_2} &= \hat\bH_{2n}(\model_2)^{-1} \hat\bG_{2n}(\model_2).
  \end{aligned}
\end{equation}
Bootstrapped versions of these estimators may be created by using the bootstrapped versions of each of the quantities appearing here. As an illustrative example, the bootstrap analogue of $\hat{\btheta}_{1 n, \model_{1} \model_{2}}$ is given by \[
  \hat{\btheta}_{1 n, \model_{1} \model_{2}}^b = \hat{\bH}_{1 n}^b(\model_{1})^{-1} \hat{\bG}_{1 n, \model_{2}}^b(\model_{1}).
\]

\subsection{Defining Different Pseudo-Outcomes}
\label{app:defining-pseudo-outcomes}

We previously defined the relevant part of the Stage 1 Q-function in terms of an ideal pseudo-outcome $Y_{1 \model_2}.$ We will unify some notation around these pseudo outcomes by delineating the different levels of knowledge and estimation required for each.
For any particular Stage 2 model $\model_2 \in \modelspace_2$, define:
\begin{align}
  Y_{1 \model_2} = Y + \xi\{A_2, \historybasis{2}(\model_2); \btheta_{20, \model_2} \}
  \label{eq:stage-1-pseudo-outcome}
  \\
  \tilde{Y}_{1 \model_2} = Y + \xi\{A_2, \historybasis{2}(\model_2); \tilde\btheta_{2n, \model_2}\}
  \label{eq:stage-1-pseudo-outcome-tilde}
  \\
  \hat{Y}_{1 \model_2} = Y + \xi\{A_2, \historybasis{2}(\model_2); \hat\btheta_{2n, \model_2}\},
  \label{eq:stage-1-pseudo-outcome-est}
\end{align}
where $\xi$ is defined in \eqref{eq:y1-s2-def}. 
The first pseudo-outcome leverages perfect knowledge of the $\Delta_2$ function along with its projection onto a linearized model, whereas  \eqref{eq:stage-1-pseudo-outcome-tilde} only requires perfect knowledge of the $\mu_{2A0},~\mu_{2Y0}$ functions needed
to calculate the estimate $\tilde\btheta_{2n, \model_2}.$ The final pseudo-outcome \eqref{eq:stage-1-pseudo-outcome-est} is based entirely on the data using cross-fitting.


For completeness, we also define some of the bootstrap pseudo-outcome quantities used in the arguments in \Cref{sec:cf-perturb-risk}:
\begin{align}
  \tilde{Y}^b_{1 \model_2} = Y + \xi(A_2, \historybasis{2}(\model_2); \tilde\btheta_{2n, \model_2}^b)
  \label{eq:stage-1-pseudo-outcome-tilde-boot}
  \\
  \hat{Y}^b_{1 \model_2} = Y + \xi(A_2, \historybasis{2}(\model_2); \hat\btheta_{2n, \model_2}^b),
  \label{eq:stage-1-pseudo-outcome-est-boot}
\end{align}
these quantities differ from those of \eqref{eq:stage-1-pseudo-outcome-tilde,eq:stage-1-pseudo-outcome-est} through the bootstrapped second-stage estimates, which are minimizers of the bootstrapped functions defined in \Cref{sec:cf-perturb-risk}.

\subsection{Overview of the Proofs}

Because of the complexity of the multi-stage Robust Q-learning process, there are several intermediate results (i.e., lemmas) that are necessary before getting to the main theorems. Indeed, much of the hard technical work lies in the proofs of several key lemmas.
Here we will provide an overview of some of the results.

To prove \Cref{thm:uposi-coverage}, we make use of only a few results: a simple result for stochastic processes with a random index in \Cref{lem:stochastic-process-random-index}, an analogue of Lemma 4.1 in \citet{kuchibhotlaValidPostselectionInference2020} stated as \Cref{lem:kuchibhotla-uniform-model-conv}, and a result on the negligibility of cross-fitting in \Cref{lem:unif-gram-inv}.

For the proof of \Cref{thm:uposi-bootstrap-valid}, we use these results along with \Cref{lem:g1-hat-rand-model,lem:inf-g-stg-1}. The first lemma ensures that cross fitting and the second-stage model selection event do not impact the quantities being studied up to a $n^{-1/2}$ rate. The second lemma ensures that the term has an influence function representation up to this same level of approximation, in probability.

The final proof of \Cref{thm:upsoi-ci-valid} is relatively straightforward, only involving a matrix argument.

\section{Proofs of Theorems}

\subsection{Proof of \texorpdfstring{\Cref{thm:uposi-coverage}}{Theorem \ref{thm:uposi-coverage}}}
\label{sec:pf-thm-uposi-coverage}

  By Theorem 3.1 in \citet{kuchibhotlaValidPostselectionInference2020}, coverage for a random-model parameter is equivalent to simultaneous coverage over all models. We consider the Stage 1 result here, with the understanding that the arguments for the stage 2 case are similar.

  The deterministic inequality \eqref{eq:uposi-inequality-1}, which holds for any pair of models $\model_1 \in \modelspace_1,~\model_2 \in \modelspace_2,$ forms the basis for the Stage 1 UPoSI region \eqref{eq:uposi-1-dagger-star}.
  The only difference between the RHS of \eqref{eq:uposi-inequality-1} and that of the inequality within \eqref{eq:uposi-1-dagger-star} is the use of an estimator $\hat\btheta_{1n,\hat\model_1 \hat\model_2}$ for the unknown quantity. By \Cref{assump:model-selection-consistency} and \Cref{lem:stochastic-process-random-index}, $\| \btheta_{10,\model_1 \hat\model_2} - \btheta_{10,\model_1 \model_2^*} \|_1 = o_p(n^{-1/2}).$

  Using the triangle inequality, and re-arranging terms, we can bound
  \begin{align*}
    \left| \frac{\hat{D}_{1n,\hat\model_2}^G + \hat{D}_{1n}^H \| \hat\btheta_{1n,\model_1 \hat\model_2} \|_1}{\hat{D}_{1n,\model_2^*}^G + \hat{D}_{1n}^H \| \btheta_{10,\model_1 \model_2^*} \|_1} - 1 \right| 
    = \left| \frac{\hat{D}_{1n}^H (\| \hat\btheta_{1n,\model_1 \hat\model_2} \|_1 - \| \btheta_{10,\model_1 \model_2^*} \|_1)
    + (\hat{D}_{1n,\hat\model_2}^G - \hat{D}_{1n,\model_2^*}^G ) }{\hat{D}_{1n,\model_2^*}^G + \hat{D}_{1n}^H \| \btheta_{10,\model_1 \model_2^*} \|_1} \right| 
    \\
    \leq \frac{\hat{D}_{1n}^H (\| \hat\btheta_{1n,\model_1 \hat\model_2} - \btheta_{10,\model_1 \hat\model_2} \|_1 + \| \btheta_{10,\model_1 \hat\model_2} - \btheta_{10,\model_1 \model_2^*} \|_1) 
    + (\hat{D}_{1n,\hat\model_2}^G - \hat{D}_{1n,\model_2^*}^G )
    }{\hat{D}_{1n,\model_2^*}^G + \hat{D}_{1n}^H \| \btheta_{10,\model_1 \model_2^*} \|_1} \\ 
    = \frac{\hat{D}_{1n}^H }{\hat{D}_{1n,\model_2^*}^G + \hat{D}_{1n}^H \| \btheta_{10,\model_1 \model_2^*} \|_1} \left\{ \| \hat\btheta_{1n,\model_1 \hat\model_2} - \btheta_{10,\model_1 \hat\model_2} \|_1 + \| \btheta_{10,\model_1 \hat\model_2} - \btheta_{10,\model_1 \model_2^*} \|_1 \right\} + o_p(n^{-1/2}),
  \end{align*}
  where the final $o_p(n^{-1/2})$ term results from \Cref{lem:stochastic-process-random-index} applied to the term resulting from $\hat{D}_{1n,\hat\model_2}^G - \hat{D}_{1n,\model_2^*}^G $.
  Examining the final equality, the second term in curly braces is also $o_p(n^{-1/2})$ by \Cref{lem:stochastic-process-random-index}. Then we may use this fact along with \Cref{lem:kuchibhotla-uniform-model-conv} to bound
  \begin{align*}
    %
    %
    \leq \frac{\hat{D}_{1n}^H }{\hat{D}_{1n,\model_2^*}^G + \hat{D}_{1n}^H \| \btheta_{10,\model_1 \model_2^*} \|_1} ~ 
    \left\{ \frac{|\model_1| (\hat{D}_{1n,\hat\model_2}^G + \hat{D}_{1n}^H \| \btheta_{10,\model_1 \hat\model_2} \|_1)}{ \Lambda_1(C_1) - C_1 \hat{D}_{1n}^H} + \| \btheta_{10,\model_1 \hat\model_2} - \btheta_{10,\model_1 \model_2^*} \|_1 \right\}
    \\ 
    \leq 
    \frac{\hat{D}_{1n,\hat\model_2}^G + \hat{D}_{1n}^H \| \btheta_{10,\model_1 \hat\model_2} \|_1}{\hat{D}_{1n,\model_2^*}^G + \hat{D}_{1n}^H \| \btheta_{10,\model_1 \model_2^*} \|_1} ~
    \frac{ C_1 \hat{D}_{1n}^H }{ \Lambda_1(C_1) - C_1 \hat{D}_{1n}^H} + o_p(n^{-1/2}).
  \end{align*}
  Under \Cref{assump:bounded-model}, $ 1/\Lambda_1(C_1) \leq c_0^{-1}.$ Further, for fixed $p_1$, $\hat{D}_{1n}^H = O_p(n^{-1/2})$, since by \Cref{lem:unif-gram-inv} this quantity behaves up to $o_p(n^{-1/2})$ like the maximum deviation of a fixed number of mean-zero sample averages. Consequently, the second fraction on the RHS converges to zero in probability. Finally, we apply \Cref{lem:stochastic-process-random-index} to the stochastic process defined by \[
    \left\{ \frac{\hat{D}_{1n,\model_2}^G + \hat{D}_{1n}^H \| \btheta_{10,\model_1 \model_2} \|_1}{\hat{D}_{1n,\model_2^*}^G + \hat{D}_{1n}^H \| \btheta_{10,\model_1 \model_2^*} \|_1} - 1 : \model_2 \in \modelspace_2(C_2) \right\}
  \]
  in the Euclidean metric space to conclude that the RHS is $o_p(1).$
  
  The second statement within \eqref{eq:uposi-coverage-statements} follows similarly, without handling the additional random model.

\subsection{Proof of \texorpdfstring{\Cref{thm:uposi-bootstrap-valid}}{Theorem \ref{thm:uposi-bootstrap-valid}}}
\label{sec:pf-thm-uposi-bootstrap-valid}

  Let us consider the first stage, as the arguments are similar in the second stage but with fewer additional technicalities. Using 
  the definitions of the $\hat{D}$ variables 
  along with 
  \Cref{lem:g1-hat-rand-model,lem:inf-g-stg-1}, we may write 
  \begin{align*}
    \sqrt{n} \Big\| \hat\bG_{1n,\hat\model_2}^b - \hat\bG_{1n,\hat\model_2} \Big\|_{\infty} = \Big\| n^{-1/2} \sum_{i=1}^n (\bootWgt_i - 1)\mathrm{Inf}_{1G \model_2^* i}  \Big\|_{\infty} + o_{p^*}(1)
    \\
    \sqrt{n} \Big\| \hat\bG_{1n,\hat\model_2} - \bG_{10,\hat\model_2} \Big\|_{\infty} = \Big\| n^{-1/2} \sum_{i=1}^n \mathrm{Inf}_{1G \model_2^* i}  \Big\|_{\infty} + o_p(1),
  \end{align*}
  with $\Eop(\bootWgt_i - 1) =0$ and $\Eop(\bootWgt_i - 1)^2 =1.$ Consequently, these obey a central limit theorem and converge weakly to the same asymptotic distribution. For example, \[
    n^{-1/2} \sum_{i=1}^n \mathrm{Inf}_{1G \model_2^* i} \CiD N \Big(\bzero,~ \bm{\Sigma}_{1} \Big)
  \]
  where $\bm{\Sigma}_{1} \defined \Eop\mathrm{Inf}_{1G \model_2^* 1}^{\otimes 2}.$ Similarly, conditional on $\bO_1,\ldots,\bO_n,$ the Lindeberg-Feller CLT ensures the bootstrap term converges in distribution: \[
    n^{-1/2} \sum_{i=1}^n (\bootWgt_i - 1) \mathrm{Inf}_{1G \model_2^* i} \CiD N \Big(\bzero,~ \bm{\Sigma}_{1n} \Big),
  \]
  where $\bm{\Sigma}_{1n} \defined n^{-1} \sum_{i=1}^n\mathrm{Inf}_{1G \model_2^* i}^{\otimes 2} \overset{a.s.}{\longrightarrow} \bm{\Sigma}_{1}$ by the Law of Large Numbers.
  Similarly, we can use \Cref{lem:unif-gram-inv} to conclude 
  \begin{align*}
    \sqrt{n} \| \hat\bH_{1n}^b - \hat\bH_{1n} \|_{\infty} &= \sqrt{n} \| \bHOracle{1}^b - \bHOracle{1} \|_{\infty} + o_{p^*}(1) \\
    \sqrt{n} \| \hat\bH_{1n} - \bH_{10} \|_{\infty} &= \sqrt{n} \| \bHOracle{1} - \bH_{10} \|_{\infty} + o_{p}(1)
  \end{align*}
  where each of the non-remainder terms on the RHS are constructed from sample means of i.i.d. terms, with the bootstrap portion having a similar composition as above. Consequently, the CLT and LLN ensure that the bootstrap quantities converge to the same asymptotic distribution as the observed-data quantities. 

\subsection{Proof of \texorpdfstring{\Cref{thm:uposi-cond-valid}}{Theorem \ref{thm:uposi-cond-valid}}}
\begin{proof}
  This follows from \Cref{thm:uposi-sim-sel} along with the simultaneity of \eqref{eq:uposi-1-cond,eq:uposi-2-cond}.
\end{proof}

\subsection{Proof of \texorpdfstring{\Cref{thm:uposi-cond}}{Theorem \ref{thm:uposi-cond}}}
\label{sec:pf-thm-upsoi-cond}

We will show the first-stage perturbation bootstrap result \eqref{eq:uposi-mult-boot-cond-1} with the understanding that the second-stage result \eqref{eq:uposi-mult-boot-cond-2} follows similarly.
In \Cref{sec:pf-thm-uposi-bootstrap-valid}, we arrived at the asymptotically linear representations
\begin{align*}
  \sqrt{n} \Big\| \hat\bG_{1n,\hat\model_2}^b - \hat\bG_{1n,\hat\model_2} \Big\|_{\infty} = \Big\| n^{-1/2} \sum_{i=1}^n (\bootWgt_i - 1)\mathrm{Inf}_{1G \model_2^* i}  \Big\|_{\infty} + o_{p^*}(1)
  \\
  \sqrt{n} \Big\| \hat\bG_{1n,\hat\model_2} - \bG_{10,\hat\model_2} \Big\|_{\infty} = \Big\| n^{-1/2} \sum_{i=1}^n \mathrm{Inf}_{1G \model_2^* i}  \Big\|_{\infty} + o_p(1).
\end{align*}
The conditional weak convergence result would follow by the Lindeberg-Feller CLT if a similar linearization holds when replacing $\bG_{10,\hat\model_2}$ with $\bG_{1n, \hat\model_2}^{cond}.$ Identical arguments as those in the proof of \Cref{lem:g1-hat-rand-model} can be made to show that
\[
  \| \hat\bG_{1n,\hat\model_2} - \bG_{1n,\model_2^*} \|_{\infty} = \| \tilde\bG_{1n,\model_2^*} - \bG_{1n,\model_2^*} \|_{\infty} + o_p(n^{-1/2}).
\]
At this point, analogous arguments to those made in \Cref{app:lemmas-stg-1-inf} can be made to show that an asymptotically linear representation holds for $\tilde\bG_{1n,\model_2^*} - \bG_{1n,\model_2^*}$ when conditioning upon the first-stage design elements $\designElements{1}.$ Such arguments would lead to an influence function for each realization of $\designElements{1},$ $\mathrm{Inf}_{1G \model_2^*}^{cond}(\designElements{1}).$
On sets with such realizations, we could write
\begin{align*}
  \sqrt{n} \Big\| \hat\bG_{1n,\hat\model_2}^b - \hat\bG_{1n,\hat\model_2} \Big\|_{\infty} = \Big\| n^{-1/2} \sum_{i=1}^n (\bootWgt_i - 1)\mathrm{Inf}_{1G \model_2^* i}^{cond}(\designElements{1})  \Big\|_{\infty} + o_{p^*}(1)
  \\
  \sqrt{n} \Big\| \hat\bG_{1n,\hat\model_2} - \bG_{10,\hat\model_2} \Big\|_{\infty} = \Big\| n^{-1/2} \sum_{i=1}^n \mathrm{Inf}_{1G \model_2^* i}^{cond}(\designElements{1})  \Big\|_{\infty} + o_p(1).
\end{align*}

Let $\epsilon_{2}^{cond}$ and $\epsilon_{1 \model_2}^{cond}$ be the errors in models defined similarly to \eqref{eq:rob-model-y2-model-2,eq:rob-model-y1-model-1} except replacing $\btheta_{20,\model_2}$ with $\btheta_{2n,\model_2^F}^{cond}$ and $\btheta_{10,\model_1 \model_2}$ with $\btheta_{1n, \model_1^F \model_2}^{cond}$. We define the quantities
\begin{align*}
  \mathrm{Inf}_{2 \model_2^* i}^{cond} \defined \history{2i} (A_{2i} - \mu_{2A0i}) \epsilon_{2 i}^{cond}
  \\
  B_{\model_2^* i}^{cond} \defined \Ind{(\historybasis{2}(\model_2)^\top \btheta_{2n,\model_2^*}^{cond} > 0)} - A_{2}
  \\
  \bM_i^{cond} \defined B_{\model_2^* i}^{cond} \left\{ A_{1} - \mu_{1A0}(\history{1}) \right\} \historybasis{1} \historybasis{2}(\model_2^*)^\top,
  \\
  \mathrm{Inf}_{1 G \model_2^* i} \defined \historybasis{1i} (A_{1i} - \mu_{1A0i}) \epsilon_{1 \model_2^* i}^{cond} + \Eop(\bM_i^{cond} ~|~ \designElements{1}) \mathrm{Inf}_{2 \model_2^* i}^{cond}.
\end{align*}
Notice that both $\mathrm{Inf}_{2 \model_2^* i}^{cond}$ and $\mathrm{Inf}_{1 G \model_2^* i}^{cond}$ have expectation zero when conditioning on the $\designElements{1}.$ We can see the first through the tower property: since the set $\designElements{1}$ is contained within $\designElements{2},$
\[
  \Eop(\mathrm{Inf}_{2 \model_2^* i}^{cond} ~|~ \designElements{1}) = \Eop\left\{ \underbrace{\Eop(\mathrm{Inf}_{2 \model_2^* i}^{cond} ~|~ \designElements{2})}_{=0} ~\bigg|~ \designElements{1}\right\}.
\]
The second follows from the first, along with the property $\Eop(\epsilon_{1 \model_2^* i}^{cond} ~|~ \designElements{1}) = 0.$

Finally, we check Lindeberg's condition for this first-stage influence function. If the following two conditions hold, then \eqref{eq:uposi-mult-boot-cond-1} follows from the Lindeberg-Feller CLT:
\begin{align}
    \frac{1}{n} \sum_{i=1}^n \Eop\left( \| \mathrm{Inf}_{1 G \model_2^* i}^{cond} \|_{2}^2 ~\Big|~ \designElements{1} \right) ~\Ind{ \left( \| \mathrm{Inf}_{1 G \model_2^* i}^{cond} \|_{2} > \gamma \sqrt{n} \right)} \longrightarrow 0, ~\text{for all } \gamma > 0
    \label{eq:app-lindeberg-1}
    \\
    \frac{1}{n} \sum_{i=1}^n \Eop{\left( \mathrm{Inf}_{1 G \model_2^* i}^{cond} ~\big|~ \designElements{1} \right)^{\otimes 2}} \longrightarrow \Sigma_1.
    \label{eq:app-lindeberg-2}
\end{align}
Since all of the terms involved in $\mathrm{Inf}_{1 G \model_2^* i}$ have finite variance, it is straightforward to verify that \eqref{eq:app-lindeberg-1,eq:app-lindeberg-2} hold with probability converging to one over the conditioning set. The first condition follows from Chebyshev's inequality, while the second follows from the LLN.

The result for the second-stage result \eqref{eq:uposi-mult-boot-cond-2} follows very similarly, establishing the influence function $\mathrm{Inf}_{2 \model_2^* i}^{cond}$ reported above. Instead, we would modify the arguments of \Cref{app:lemmas-stg-2-inf} and use the Lindeberg-Feller CLT.

\section{Proofs of Lemmas}

\subsection{General Lemmas}

First, we state some general lemmas. The first two are simple results about random indices, which will be useful for handling random vectors that depend on a random model.
\begin{lemma}
  \label{lem:stochastic-process-random-index-bound}
  Let $\left\{ X_s : s \in \mathcal{S} \right\}$ be a stochastic process with indexing set $\mathcal{S}$. Suppose $X_s$ takes its values in a normed metric space $(\Xcal,\|\cdot\|),$ where $\|\cdot\|$ is a norm. Let $\hat{s}$ be an $\mathcal{S}-$valued random variable defined on the same probability space as the stochastic process. Then, for any fixed point $s' \in \mathcal{S}$, and any $\varepsilon \geq 0$, \[
    \Prob\left\{ \| X_{\hat{s}} - X_{s'} \| > \varepsilon \right\} \leq \Prob\left( \hat{s} \neq s' \right).
  \]
\end{lemma}
\begin{proof}
  Since the stochastic process and random index are measurable on the same probability space, we may decompose the LHS probability as \[
    \Prob\left\{ \| X_{\hat{s}} - X_{s'} \| > \varepsilon \right\} = \Prob\left\{ \| X_{\hat{s}} - X_{s'} \| > \varepsilon,~ \hat{s} = s' \right\} + \Prob\left\{ \| X_{\hat{s}} - X_{s'} \| > \varepsilon,~ \hat{s} \neq s' \right\}.
  \]
  The first term is zero, since the event $\left\{ \hat{s} = s' \right\}$ ensures that $\| X_{\hat{s}} - X_{s'} \|=0$ almost surely. The second term is bounded by $\Prob\left( \hat{s} \neq s' \right),$ completing the proof.
\end{proof}

\begin{lemma}
  \label{lem:stochastic-process-random-index}
  Under the setup of the previous lemma, suppose \( \Prob\left( \hat{s} \neq s' \right) \rightarrow 0. \) Let $f(n)$ be any non-negative, non-increasing function of $n \geq 1$. Then \(
    \| X_{\hat{s}} - X_{s'} \| = o_p(f(n)).
  \)
\end{lemma}
\begin{proof}

  Under \Cref{lem:stochastic-process-random-index-bound}, we may fix any $\varepsilon > 0$ and bound \[
    \Prob\left\{ \| X_{\hat{s}} - X_{s'} \| > \varepsilon f(n) \right\} \leq \Prob\left( \hat{s} \neq s' \right).
  \]
  
  By assumption, for any $\kappa > 0$ there exists some $N_1$ depending on $\kappa$ such that \(
    \Prob\left( \hat{s} \neq s' \right) < \kappa
  \)
  holds for all $n \geq N_1.$ Putting this with the previous inequality, we find that
  \[
    \Prob\left\{ \| X_{\hat{s}} - X_{s'} \| > \varepsilon f(n) \right\} < \kappa
  \]
  for all $n \geq N_1.$ Since both $\varepsilon$ and $\kappa$ are arbitrary, the proof is complete. Notice that $N_1$ depends on $\kappa$ but not $\varepsilon,$ reflecting the uniform rate represented by $f(n)$.
\end{proof}

The following lemma is a standard result for matrices \citep[e.g.,][]{stewartContinuityGeneralizedInverse1969}.
\begin{lemma}
  \label{lem:matrix-inversion-diff}
  Let $\hat\bM$ and $\bM$ be two matrices and $\| \cdot \|$ be any proper matrix norm. Suppose (i) both $\bM^{-1}$ and $\hat\bM^{-1}$ exist with $0 < c_0 < \| \bM^{-1} \| < c_1 < \infty$ for constants $c_0$ and $c_1$, and (ii) $\| \hat\bM - \bM \| \leq (2 \| \bM^{-1} \| )^{-1}.$ Then,
  \[
    \| \hat\bM^{-1} - \bM^{-1} \| \leq 2 c_0^{-2} \| \hat\bM - \bM \|.
  \]
\end{lemma}

Next is a useful lemma for handling the sum of a product of random variables.
\begin{lemma}
  \label{lem:sum-with-op}

  Let $X_i,Y_i$ for $i=1,\ldots,n$ be random variables on a common probability space, although not necessarily i.i.d. Let \[
    \frac{1}{n} \sum_{i=1}^n |X_i| \CiP \mu,
  \]
  and suppose $\max_{i=1,\ldots,n} |Y_i| = o_p(n^{-1/2})$.
  Then, if $\mu < \infty,$ \[
    n^{-1/2} \sum_{i=1}^n X_i Y_i \CiP 0.
  \]
\end{lemma}
\begin{proof}
  By hypothesis, for any $\epsilon,\delta >0$ there exists $N \equiv N(\epsilon,\delta)$ such that \[
    \Prob\left( \max_{i=1,\ldots,n} |Y_i| < \epsilon n^{-1/2} \right) \geq 1-\delta
  \]
  for all $n\geq N.$ Let $\Omega_1$ represent the sets on which the the condition inside the probability statement holds. By definition, $\Prob(\Omega_1) \geq 1-\delta.$ For $n\geq N$ on $\Omega_1,$ we have
  \begin{align}
    \nonumber
    \Big| n^{-1/2} \sum_{i=1}^n X_i Y_i \Big| \leq n^{-1/2} \sum_{i=1}^n |X_i| \epsilon n^{-1/2} \\
    \label{eq:lem-sum-with-op-1}
    = \epsilon n^{-1} \sum_{i=1}^n |X_i|.
  \end{align}
  By the hypothesis, for any $\gamma>0$ the sample average in \eqref{eq:lem-sum-with-op-1} converges to $\mu$ with probability no less than $1-\gamma.$
  Let $\Omega_2$ represent the sets satisfying this condition. Then on $\Omega_1 \cap \Omega_2$, the average in \eqref{eq:lem-sum-with-op-1} converges to \(
    \epsilon \mu.
  \)
  By Frechet's inequality, $\Prob(\Omega_1 \cap \Omega_2) \geq 1 - \delta - \gamma$ with $\delta,\gamma$ arbitrary. Since $\epsilon$ is also arbitrary, the proof is complete.
\end{proof}

Finally, we state this slight variation on Lemma 6.1 in \citet{chernozhukovDoubleDebiasedMachine2018f}.
\begin{lemma}
  \label{lem:cond-l2-uncond}
  Let $\bX_n$ and $\bZ_n$ be a sequence of random vectors defined on the same probability space. Suppose $\|\bX_n\| = O_p(r_n)$ conditionally on $\bZ_n$, for a sequence of positive constants $r_n.$ Then $\|\bX_n\| = O_p(r_n)$ unconditionally as well.
\end{lemma}

\subsection{Lemmas for Cross-fitted Functions}

\begin{lemma}
  \label{lem:mean-zero-rate-vec}
  Suppose $(\bW_i, \bX_i)$ are $\bO_i-$measurable random variables, where $\bO_i$ are i.i.d. from some distribution $P_0$ and $\bW_1 \in \R^d$. Suppose for $i=1,\ldots,n$ that 
  $\Eop( \bW_i ~|~ \bX_i) = \bzero \in \R^{ |d| }$ and
  $\|\Eop(\bW_1^{\otimes 2})\|_{\infty} \leq C_W < \infty$.
  Let $h(\bX_i ; \crossfitDataC)$ for any $k=1,\ldots,K$ and $i\in\Ik$ be a cross-fitted function. If $\LtwoPzero[ ]{h} = o_p(r_n)$ for some sequence $r_n,$ then
  \[
    \left\| n^{-1} \sum_{k=1}^K \sum_{i \in \Ik} \bW_i h(\bX_i ; \crossfitDataC) \right\|_{\infty} = o_p{(n^{-1/2} r_n)},
  \]
  where the constant does not depend on the dimension of $\bW$ or $\bX.$ Furthermore, if $\bootWgt_1,\ldots,\bootWgt_n$ are random multipliers independent of $\bO_1,\ldots,\bO_n$ with $\Eop(\bootWgt_1^2) < \infty$, then the previous result holds with $\bW_i$ replaced by $\bootWgt_i \bW_i$.
\end{lemma}
\begin{proof}
  Consider the $j^{th}$ element of the sum inside the max-norm:
  \[
    S_{njk} \defined \sum_{i \in \Ik} \bW_i\subj h(\bX_i ; \crossfitDataC).
  \]
  Let $\Pcal_K \defined \{\Ik\}_{k=1,\ldots,K}$ represent a particular random partition of the indices $1,\ldots,n$ into $K$ disjoint and set with a roughly-equivalent size.
  As in \citet{ertefaieRobustQLearning2021}, we have $\Eop(S_{njk})=0$ and conditional variance
  \begin{equation}
    \label{eq:lem-mean-zero-rate-vec-1}
    \Eop(S_{njk}^2 ~|~ \crossfitDataC, \Pcal_K) \leq n_k C_W \Eop{\left\{ h^2(\bX_i ; \crossfitDataC) ~|~ \crossfitDataC, \Pcal_K \right\}},
  \end{equation}
  provided by the Cauchy-Schwarz inequality.
  By Chebyshev's Inequality, we have
  \[
    \Prob \left( |S_{njk}| > \varepsilon \sqrt{n} ~|~ \crossfitDataC, \Pcal_K \right) \leq \varepsilon^{-2} n^{-1} n_k C_W  \Eop{\left\{ h^2(\bX_i ; \crossfitDataC) ~|~ \crossfitDataC, \Pcal_K \right\}}.
  \]
  Since $\Pcal_K$ is independent of the data, this conditional expectation is equal to $\LtwoPzero[ ]{h}^2.$
  Next, make the substitution $n^{-1}n_k = K^{-1} + o(1),$ where the $o(1)$ term is uniform over the index $k$ due to the finite number of folds $K$. By hypothesis,
  \[
    \Prob \left( |S_{njk}| > \varepsilon \sqrt{n} \right) \leq \left\{ K^{-1} + o(1) \right\}\varepsilon^{-2} C_W o_p(r_n).
  \]
  Now, use the upper bound \(
    \left| \sum_{K=1}^K S_{njk} \right| \leq K \max_{k=1,\ldots,K} |S_{njk}|
  \) 
  and sub-additivity of the probability measures:
  \begin{align*}
    \Prob \left( \left|\sum_{K=1}^K S_{njk} \right| > \varepsilon \sqrt{n} \right) &\leq \Prob \left( K \max_{k=1,\ldots,K} |S_{njk}| > \varepsilon \sqrt{n} \right) \\ 
    &= \Prob \left( \bigcup_{k=1,\ldots,K} \left\{ |S_{njk}| > \varepsilon \sqrt{n}/K \right\} \right) \\
    &\leq \sum_{k=1}^K \Prob \left( |S_{njk}| > \varepsilon \sqrt{n}/K \right) \\ 
    &\leq \left\{ 1 + o(1) \right\}K^2 \varepsilon^{-2} C_W \Eop{(\LtwoPzero[ ]{h}^2)}.
  \end{align*}
  The displayed result in the lemma statement follows as $C_W$ is a uniform bound on the variance over the index $j=1,\ldots,d$. An application of the sub-additivity of probability measures again demonstrates that \[
    \Prob \left( \max_{1\leq j \leq d} \left|\sum_{K=1}^K S_{njk} \right| > \varepsilon \sqrt{n} \right) \leq d \left\{ 1 + o(1) \right\}K^2 \varepsilon^{-2} C_W \Eop{(\LtwoPzero[ ]{h}^2)}.
  \]
  An application of \Cref{lem:cond-l2-uncond} completes the proof.

  The result for random multipliers holds by similar arguments. To see this, let $S_{njk}^b$ be the analogous sum with random multipliers. By independence, \(\Eop (\bootWgt_i \bW_i ~|~ \bX_i ) = 0\) and \(\Var(\bootWgt_i\bW_i ~|~ \bX_i) = \Eop(\bootWgt_i^2) \Eop(\bW_i^{\otimes 2} ~|~ \bX_i)). \) The bound $C_W$ can then be replaced by $C_W' = \Eop(\bootWgt_1^2) C_W.$ The analogue of \eqref{eq:lem-mean-zero-rate-vec-1} resulting from this is \[
    \Eop(S_{njk}^{b2} ~|~ \crossfitDataC, \Pcal_K) \leq C_W' n_k \LtwoPzero[ ]{h}^2,
  \]
  from which the remainder of the proof follows as before.
  
\end{proof}

\begin{lemma}
  \label{lem:cross-rate-vec}
  Suppose the setup of the previous lemma holds. Let $\bh_1$ and $\bh_2$ be two vector-valued cross-fitting functions such that for $k=1,\ldots,K,$ $i \in \Ik,$ and $j=1,2$, $\bh_j(\bX_i ; \crossfitDataC)$ takes values in $\R^{d_j}.$ Let these functions satisfy $\LtwoPzero[ ]{\bh_j} = o_p(r_{nj})$ for the sequences $r_{n1}$ and $r_{n2}.$ Then for both $v=1$ and $v=0$,
  \[
    \left\| n^{-1} \sum_{k=1}^K \sum_{i \in \Ik} \bootWgt_i^{v}~ \bh_1(\bX_i ; \crossfitDataC) \bh_2(\bX_i ; \crossfitDataC)^{\top} \right\|_{\infty} = O_p{(r_{n1} r_{n2})},
  \]
  where the constant does not depend on the dimension of $\bW$ or $\bX.$
\end{lemma}
\begin{proof}
  Start with the $(j, \ell)^{th}$ element of the sum inside the max-norm:
  \[
    S_{nj{\ell} k} \defined \sum_{i \in \Ik} \bootWgt_i^{v}~ h_{1 j}(\bX_i ; \crossfitDataC) h_{2 \ell}(\bX_i ; \crossfitDataC),
  \]
  where we let $h_{ab}$ be the $b^{th}$ coordinate of the $\bh_a$ function, $a=1,2,~b=1,\ldots,d_a.$ Apply the Cauchy-Schwarz inequality to this term to find
  \begin{equation}
    \label{eq:cross-prod-1}
    n_k^{-1}| S_{nj{\ell} k} | \leq \left\{ \frac{1}{n_k} \sum_{i \in \Ik} \bootWgt_i^{2v}~ h_{1 j}^2(\bX_i ; \crossfitDataC) \right\}^{1/2} \left\{ \frac{1}{n_k} \sum_{i \in \Ik} h_{2 \ell}^2(\bX_i ; \crossfitDataC) \right\}^{1/2}.
  \end{equation}
  For the second factor in curly braces, apply Markov's Inequality to find
  \begin{align*}
    \Prob\left( \frac{1}{n_k} \sum_{i \in \Ik} h_{2 \ell}^2(\bX_i ; \crossfitDataC) > a ~\bigg|~ \crossfitDataC, ~\Pcal_K \right) \leq \frac{1}{a} \Eop{\left\{ h_{2 \ell}^2(\bX_i ; \crossfitDataC) ~\big|~ \crossfitDataC, ~\Pcal_K \right\}}
  \end{align*}
  Integrating both sides over the distribution of $\crossfitDataC$ and $\Pcal_K$, the expectation on the RHS simplifies to $\Eop(\LtwoPzero[ ]{h_{2 \ell}}^2).$
  If $v=0$, then this same argument may be made for the first term in curly braces as well. This shows that \[
    n_k^{-1}| S_{nj{\ell} k} | = O_p{\left\{ \Eop{(\LtwoPzero[ ]{h_{2 \ell}})} \Eop{(\LtwoPzero[ ]{h_{1 j}})} \right\}},
  \]
  which is $o_p(r_{n1}r_{n2})$ by \Cref{lem:cond-l2-uncond}, uniformly in the indices $j$ and $\ell.$
  The number of folds $K$ and the dimensions $d_1,d_2<\infty$ are fixed. Further, our notation $\LtwoPzero[ ]{\bh_1}$ represents the maximum of the $\LtwoPzero[ ]{\cdot}$ norm over the coordinates of $\bh_1.$ 
  %
  Following similar arguments as in the proof of \Cref{lem:mean-zero-rate-vec}, we have \[
    \Prob\left( \left| \sum_{k=1}^K S_{nj{\ell}k} \right| > a \right) \leq \left\{ 1 + o(1) \right\} \frac{K}{a} o_p(r_{n1}r_{n2}).
  \]
  Taking the maximum over $d_1 \times d_2$ elements in the resulting matrix, and using the uniformity of the previous bound in $j$ and $\ell$, use the expression below to complete the $v=0$ case: \[
    \Prob\left( \max_{j \leq d_1, \ell \leq d_2} \left| \sum_{k=1}^K S_{nj{\ell}k} \right| > a \right) \leq d_1 d_2 \left\{ 1 + o(1) \right\} \frac{K}{a} o_p(r_{n1}r_{n2}).
  \]

  To handle the $v=1$ case, we need only handle the first curly-braced term in \eqref{eq:cross-prod-1}. Apply the Markov's inequality argument and leverage the independence $\bootWgt_i \independent \bO_i$ to find \[
    \Prob\left( \frac{1}{n_k} \sum_{i \in \Ik} \bootWgt_i^{2}~ h_{1 j}^2(\bX_i ; \crossfitDataC) > a ~\bigg|~ \crossfitDataC, ~\Pcal_K \right) \leq \frac{1}{a} \Eop(\bootWgt_i^2) \Eop\left\{ h_{1 j}^2(\bX_i ; \crossfitDataC) ~\bigg|~ \crossfitDataC, ~\Pcal_K \right\}.
  \]
  Since $\Eop(\bootWgt_i^2) < \infty$, the rest follows as in the $v=0$ case, completing the proof.
\end{proof}

\subsection{Lemmas Required for the Proof of \texorpdfstring{\Cref{thm:uposi-coverage}}{Theorem \ref{thm:uposi-coverage}}}
\label{app:lemmas-for-thm-1}

\begin{lemma}[Analogue of Lemma 4.1 in \citet{kuchibhotlaValidPostselectionInference2020}]
  \label{lem:kuchibhotla-uniform-model-conv}
  Suppose \Cref{assump:bounded-model} holds.
  For all models $\model_1 \in \modelspace_1(C_1),~\model_2 \in \modelspace_2(C_2)$, and a data-dependent $\hat\model_2 \in \modelspace_2,$
  \begin{align*}
    \| \hat\btheta_{2n,\model_2} - \btheta_{20,\model_2} \|_{1} \leq \frac{|\model_2| (\hat{D}_{2n}^G + \hat{D}_{2n}^H \| \btheta_{20,\model_2} \|_1)}{ \Lambda_2(C_2) - C_2 \hat{D}_{2n}^H} \\
    \| \hat\btheta_{1n,\model_1 \hat\model_2} - \btheta_{10,\model_1 \hat\model_2} \|_{1} \leq \frac{|\model_1| (\hat{D}_{1n,\hat\model_2}^G + \hat{D}_{1n}^H \| \btheta_{10,\model_1 \hat\model_2} \|_1)}{ \Lambda_1(C_1) - C_1 \hat{D}_{1n}^H}.
  \end{align*}
\end{lemma}
\begin{proof}
  These follow using the same arguments of \citet{kuchibhotlaValidPostselectionInference2020}. We will follow the more complex case given by the second inequality.

  For a particular $\model_1 \in \modelspace_1(C_1)$, we have \[
    \hat\btheta_{1n,\model_1 \hat\model_2} - \btheta_{10,\model_1 \hat\model_2} = \left\{ \hat\bH_{1n}(\model_1) \right\}^{-1} \left[ \{\hat\bG_{1n,\hat\model_2} - \bG_{10,\hat\model_2}\}(\model_1) - \{\hat\bH_{1n} - \bH_{10}\}(\model_1) \btheta_{10,\model_1 \hat\model_2} \right],
  \]
  which results from the normal equations. Examining the $\ell_2 \mapsto \ell_2$ operator norm, we can bound
  \begin{equation}
    \label{eq:deterministic-estimation-bound}
    \| \hat\bH_{1n}(\model_1) - \bH_{10}(\model_1) \|_{2,2} \leq | \model_1 | \hat{D}_{1n}^H,
  \end{equation}
  which over all $\model_1 \in \modelspace_1(C_1)$ is bounded by $C_1 \hat{D}_{1n}^H.$ Consequently, using standard results \citep[e.g., eq. (5.8.2)][]{horn2012matrix} we may bound
  \begin{align*}
    \| \hat\bH_{1n}(\model_1)^{-1} \|_{2,2} 
    \leq \frac{
      \| \bH_{10}(\model_1)^{-1} \|_{2,2}
    }{
      1 - \left\| \bH_{10}(\model_1)^{-1} \left\{ \hat\bH_{1n}(\model_1) - \bH_{10}(\model_1) \right\} \right\|_{2,2}
    } \\
    \leq \left\{ 1 / \| \bH_{10}(\model_1)^{-1} \|_{2,2} - C_1 \hat{D}_{1n}^H \right\}^{-1} \\
    \leq \left\{ \Lambda_1(C_1) - C_1 \hat{D}_{1n}^H \right\}^{-1},
  \end{align*}
  where the second inequality uses the sub-multiplicativity of the $\ell_2 \mapsto \ell_2$ operator norm along with \eqref{eq:deterministic-estimation-bound} and the final inequality uses the definition of $\Lambda_1(C_1).$ Consequently, for $C_1$ satisfying $\Lambda_1(C_1) - C_1 \hat{D}_{1n}^H > 0$, we may bound
  \begin{align*}
    \| \hat\btheta_{1n,\model_1 \hat\model_2} - \btheta_{10,\model_1 \hat\model_2} \|_2 \leq \frac{ \| \{ \hat\bG_{1n,\hat\model_2} - \bG_{10,\hat\model_2}\}(\model_1) \|_2 + \| \{\hat\bH_{1n} - \bH_{10}\}(\model_1) \btheta_{10,\model_1 \hat\model_2} \|_2
    }{\Lambda_1(C_1) - C_1 \hat{D}_{1n}^H}
    \\
    \leq \frac{ |\model_1|^{1/2} (\hat{D}_{1n}^G + \hat{D}_{1n}^H \| \btheta_{10,\model_1 \hat\model_2} \|_1)}{ \Lambda_1(C_1) - C_1 \hat{D}_{1n}^H}.
  \end{align*}
  The stated result follows by applying \[
    \| \hat\btheta_{1n,\model_1 \hat\model_2} - \btheta_{10,\model_1 \hat\model_2} \|_1 \leq |\model_1|^{1/2} \| \hat\btheta_{1n,\model_1 \hat\model_2} - \btheta_{10,\model_1 \hat\model_2} \|_2.
  \]
\end{proof}

\begin{proposition}
  \label{prop:uniform-mu-Z-rates}
  Let $\bZ_{\ell} = A_{\ell} \historybasis{\ell}$ with $\bmu_{\ell Z0} = \mu_{\ell A0}(\history{\ell}) \historybasis{\ell}$ and $\hat\bmu_{\ell Z} = \hat\mu_{\ell A}(\history{\ell}) \historybasis{\ell}$. If \Cref{assump:boundedness,assump:rate-assumps} hold, then
  \begin{align*}
    \LtwoPzero[ ]{\hat\bmu_{\ell Z} - \bmu_{\ell Z0}} = o_p(n^{-1/4}) \label{eq:uniform-mu-Z-rates-1} \\
    \LtwoPzero[ ]{\hat\bmu_{2 Z} - \bmu_{2 Z0}} \LtwoPzero[ ]{\hat\mu_{2 Y} - \mu_{2 Y0}} = o_p(n^{-1/2}) \\
    \LtwoPzero[ ]{\hat\bmu_{1 Z} - \bmu_{1 Z0}} \LtwoPzero[ ]{\hat\mu_{1 Y} - \mu_{1 Y \hat\model_2 0}} = o_p(n^{-1/2}),
  \end{align*}
  where the norm \( \LtwoPzero[ ]{\bv} = \max_{j=1,\ldots,d} \LtwoPzero[ ]{\bv\subj} \) for $d-$dimensional $\bv$.
\end{proposition}
\begin{proof}
  Simple inequalities show that
  \[
    \LtwoPzero[ ]{\hat\bmu_{\ell Z} - \bmu_{\ell Z0}} = \LtwoPzero[ ]{ (\hat\mu_{\ell A} - \mu_{\ell A0})(\history{\ell}) \historybasis{\ell}} \leq C \LtwoPzero[ ]{ \hat\mu_{\ell A} - \mu_{\ell A0}}.
  \]
  Consequently, the stated rates follow directly from \Cref{assump:rate-assumps}.
\end{proof}

\begin{lemma} \label{lem:unif-gram-inv}
  If \Cref{assump:bounded-model,assump:rate-assumps} hold, then:
  \begin{align}
    \label{eq:Hn-approx}
    \| \hat{\bH}_{\ell n} - \bHOracle{\ell} \|_{\infty} = o_p(n^{-1/2})
    \\
    \label{eq:Hn-b-approx}
    \| \hat{\bH}^b_{\ell n} - \bHOracle{\ell}^b \|_{\infty} = o_p(n^{-1/2})
    \\
    \label{eq:Hn-inv-approx}
    \max_{\model_{\ell} \in \modelspace_{\ell}} \| \hat{\bH}_{\ell n}(\model_{\ell})^{-1} - \bHOracle{\ell}(\model_{\ell})^{-1} \|_{\infty} = o_p(n^{-1/2})
    \\
    \label{eq:Hn-b-inv-approx}
    \max_{\model_{\ell} \in \modelspace_{\ell}} \| \hat{\bH}^b_{\ell n}(\model_{\ell})^{-1} - \bHOracle{\ell}^b(\model_{\ell})^{-1} \|_{\infty} = o_{p^*}(n^{-1/2})
  \end{align}
\end{lemma}
\begin{proof}
  We begin by defining the difference
  \[
    \bm{\varphi}_{\ell i}^H \defined (\bZ_{\ell i} - \hat{\bmu}_{\ell Z i})^{\otimes 2} - (\bZ_{\ell i} - \bmu_{\ell Z0i})^{\otimes 2}
  \]
  for $i=1,\ldots,n.$ By adding and subtracting terms, we may express this as
  \begin{align}
    \nonumber
    \bm{\varphi}_{\ell i}^H
    &= (\bZ_{\ell i} - \bmu_{\ell Z0i} + \hat{\bmu}_{\ell Zi} - \bmu_{\ell Z0i})^{\otimes 2} - (\bZ_{\ell i} - \bmu_{\ell Z0i})^{\otimes 2} \\
    &= (\hat{\bmu}_{\ell Zi} - \bmu_{\ell Z0i})^{\otimes 2} 
    + (\hat{\bmu}_{\ell Zi} - \bmu_{\ell Z0i})(\bZ_{\ell i} - \bmu_{\ell Z0i})^\top + (\bZ_{\ell i} - \bmu_{\ell Z0i})(\hat{\bmu}_{\ell Zi} - \bmu_{\ell Z0i})^\top. 
    \label{eq:app-phi-diff-H}
  \end{align}
  To prove \eqref{eq:Hn-approx,eq:Hn-b-approx}, re-express the differences as
  \begin{align*}
    \hat{\bH}_{\ell n} - \bHOracle{\ell} &= \frac{1}{n} \sum_{i=1}^n \bm{\varphi}_{\ell i}^H \\
    \hat{\bH}^b_{\ell n} - \bHOracle{\ell}^b &= \frac{1}{n} \sum_{i=1}^n \bootWgt_i \bm{\varphi}_{\ell i}^H.
  \end{align*}
  Since \Cref{lem:mean-zero-rate-vec,lem:cross-rate-vec} apply either in the presence or absence of these random multipliers, we focus on the argument for \eqref{eq:Hn-approx}.
  Of the three terms comprising \eqref{eq:app-phi-diff-H}, the first satisfies the conditions of \Cref{lem:cross-rate-vec} while the final two satisfy those of \Cref{lem:mean-zero-rate-vec}.
  Consequently,
  \begin{align*}
    \frac{1}{n} \sum_{i=1}^n \left\{ (\hat{\bmu}_{\ell Zi} - \bmu_{\ell Z0i})(\bZ_{\ell i} - \bmu_{\ell Z0i})^\top + (\bZ_{\ell i} - \bmu_{\ell Z0i})(\hat{\bmu}_{\ell Zi} - \bmu_{\ell Z0i})^\top \right\} = o_p(n^{-3/4}) \\
    \frac{1}{n} \sum_{i=1}^n (\hat{\bmu}_{\ell Zi} - \bmu_{\ell Z0i})^{\otimes 2} = o_p(n^{-1/2})
  \end{align*}
  under \Cref{assump:boundedness,assump:rate-assumps}, which establishes the first two equations.

  Equation \eqref{eq:Hn-inv-approx} follows from \eqref{eq:Hn-approx} along with \Cref{lem:matrix-inversion-diff}. Specifically, \Cref{assump:bounded-model} and the arguments in the proof of \Cref{lem:kuchibhotla-uniform-model-conv} ensure that
  \[
    \| \bHOracle{\ell}(\model_{\ell})^{-1} \|_{\infty} \leq \left\{ \Lambda_{\ell}(C_{\ell}) - C_{\ell} \| \hat\bH_{\ell n} - \bHOracle{\ell} \|_{\infty} \right\}^{-1},
  \]
  since $\| \bm{A} \|_{\infty} \leq \| \bm{A} \|_{2,2}$ for a matrix $\bA.$
  Condition (ii) follows from this fixed bound along with the previously established rates. Consequently, the $o_p(n^{-1/2})$ results established previously ensure that the inverses also obey this rate. 
  Using the result of \Cref{lem:matrix-inversion-diff} we obtain the first of the following series of inequalities:
  \begin{align*}
    \| \hat\bH_{\ell n}(\model_{\ell})^{-1} - \bHOracle{\ell}(\model_{\ell})^{-1} \|_{\infty} 
    &\leq 
    \frac{2 \| \hat\bH_{\ell n}(\model_{\ell}) - \bHOracle{\ell}(\model_{\ell}) \|_{\infty}
    }{
      \{ \Lambda_{\ell}(C_{\ell}) - C_{\ell} \| \hat\bH_{\ell n} - \bHOracle{\ell} \|_{\infty} \}^2
    } \\
    &\leq 
    \frac{2 \| \hat\bH_{\ell n}(\model_{\ell}) - \bHOracle{\ell}(\model_{\ell}) \|_{\infty}
    }{
      \{ c_0 - C_{\ell} \| \hat\bH_{\ell n} - \bHOracle{\ell} \|_{\infty} \}^2
    } \\
    &\leq 
    \frac{2 \| \hat\bH_{\ell n} - \bHOracle{\ell} \|_{\infty}
    }{
      \{ c_0 - C_{\ell} \| \hat\bH_{\ell n} - \bHOracle{\ell} \|_{\infty} \}^2
    },
  \end{align*}
  where the final two inequalities follow from the lower bound $c_0$ on $\Lambda_{\ell}(C_{\ell})$, and the upper bound $\| \hat\bH_{\ell n}(\model_{\ell}) - \bHOracle{\ell}(\model_{\ell}) \|_{\infty} \leq \| \hat\bH_{\ell n} - \bHOracle{\ell} \|_{\infty},$ respectively.
  In this final expression, the denominator converges to $c_0^2$ in probability and the numerator is $o_p(n^{-1/2}),$ which both follow from \eqref{eq:Hn-approx}. Consequently, the entire expression is $o_p(n^{-1/2}).$ A similar argument establishes the final result \eqref{eq:Hn-b-inv-approx} using \eqref{eq:Hn-b-approx}.
\end{proof}

\subsection{Lemmas Required for the Proof of \texorpdfstring{\Cref{thm:uposi-bootstrap-valid}}{Theorem \ref{thm:uposi-bootstrap-valid}}}
\label{app:lemmas-for-thm-2}

This section includes several results needed to prove the main lemmas used in the proof of \Cref{thm:uposi-bootstrap-valid}. To organize these results, we focus on three sets:
\begin{enumerate}
  \item those used to show that the stage 2 estimators have an influence function representation uniformly over the stage 2 models $\model_2 \in \modelspace_2(C_2)$ (\Cref{app:lemmas-stg-2-inf}),
  \item those used to establish the asymptotic negligibility of cross-fitting in stage 1 (\Cref{app:lemmas-stg-1-cf}); and
  \item those used to show that the stage 1 estimators have an influence function representation uniformly over the stage 1 models $\model_1 \in \modelspace_1(C_1)$ (\Cref{app:lemmas-stg-1-inf}).
\end{enumerate}

\subsubsection{Results Establishing an Influence Function in the Second Stage}
\label{app:lemmas-stg-2-inf}

\begin{proposition}
  \label{prop:unif-g2}
  Under \Cref{assump:boundedness,assump:rate-assumps}, 
  \begin{align}
    \label{eq:app-g2-hat-tilde-diff}
    \| \hat\bG_{2n} - \tilde\bG_{2n} \|_{\infty} &= o_p(n^{-1/2}) \\
    \label{eq:app-g2-hat-tilde-diff-b}
    \| \hat\bG_{2n}^b - \tilde\bG_{2n}^b \|_{\infty} &= o_{p^*}(n^{-1/2}). 
  \end{align}
\end{proposition}
\begin{proof}
  As in the proof of \Cref{lem:unif-gram-inv}, begin by writing \eqref{eq:app-g2-hat-tilde-diff,eq:app-g2-hat-tilde-diff-b} as
  \begin{align*}
    \hat\bG_{2n} - \tilde\bG_{2n} &= \frac{1}{n} \sum_{i=1}^n \bm{\varphi}_{2i}^G
    \\
    \hat\bG_{2n}^b - \tilde\bG_{2n}^b &= \frac{1}{n} \sum_{i=1}^n \bootWgt_i \bm{\varphi}_{2i}^G,
  \end{align*}
  where we define for $i=1,\ldots,n,$ the term inside the sum as
  \[
    \bm{\varphi}_{2i}^G \defined (\bZ_{2i} - \hat\bmu_{2Zi}) (Y_i - \hat\mu_{2Yi}) - (\bZ_{2i} - \bmu_{2Z0i}) (Y_i - \mu_{2Y0i}).
  \]
  By adding and subtracting terms, we arrive at the expression
  \begin{align}
    \bm{\varphi}_{2i}^G
    =~& (\bZ_{2i} - \bmu_{2Z0i}) (\mu_{2Y0i} - \hat\mu_{2Yi}) \label{eq:unif-g2-1} \\
    &+ (\mu_{2A0i} - \hat\mu_{2Ai}) \left\{ \historybasis{2i} (Y_i - \mu_{2Y0i}) \right\} \label{eq:unif-g2-2} \\
    &+ (\bmu_{2Z0i} - \hat\bmu_{2Zi}) (\mu_{2Y0i} - \hat\mu_{2Yi}), \label{eq:unif-g2-3}
  \end{align}
  where we've substituted the identity $\bmu_{2Z0i} - \hat\bmu_{2Zi} = (\mu_{2A0i} - \hat\mu_{2Ai}) \historybasis{2i}$ in \eqref{eq:unif-g2-2}.

  As in the proof of \Cref{lem:unif-gram-inv}, we plan to use \Cref{lem:mean-zero-rate-vec,lem:cross-rate-vec} to show the results \eqref{eq:app-g2-hat-tilde-diff,eq:app-g2-hat-tilde-diff-b}. Consequently, we focus only on the former result, as the random multiplier does not impact the argument.
  The proof will be complete if we show that each of these labeled terms is $o_p(n^{-1/2})$ when averaged over $i=1,\ldots,n.$ 
  
  Notice that \Cref{lem:mean-zero-rate-vec} applies to \eqref{eq:unif-g2-1}, since the first vector factor has conditional mean zero on $\history{2i},$ and the remaining term involves cross-fitting. Equation \eqref{eq:unif-g2-2} also fits this condition, since the term in curly braces has mean zero conditioned upon $\history{2i}.$ Consequently, \Cref{assump:rate-assumps} ensures that
  \[
    \frac{1}{n} \sum_{i=1}^n \left\{ (\bZ_{2i} - \bmu_{2Z0i}) (\mu_{2Y0i} - \hat\mu_{2Yi}) + (\mu_{2A0i} - \hat\mu_{2Ai}) \history{2i} (Y_i - \mu_{2Y0i}) \right\} = o_p(n^{-1/2}).
  \]
  For the final term \eqref{eq:unif-g2-3}, we may apply \Cref{prop:uniform-mu-Z-rates} with \Cref{lem:cross-rate-vec} to show
  \[
    \frac{1}{n} \sum_{i=1}^n (\bmu_{2Z0i} - \hat\bmu_{2Zi}) (\mu_{2Y0i} - \hat\mu_{2Yi}) = o_p(n^{-1/2}),
  \]
  concluding the proof of \eqref{eq:app-g2-hat-tilde-diff}. The result \eqref{eq:app-g2-hat-tilde-diff-b} follows similarly.
\end{proof}

\begin{lemma}
  \label{lem:unif-consistency-both-stages}
  Under \Cref{assump:bounded-model,assump:model-selection-consistency}, the following inequalities hold for any $\model_1 \in \modelspace_1(C_1)$ and $\model_2 \in \modelspace_2(C_2)$:
  \begin{align*}
    \| \hat\btheta_{2n,\model_2} - \tilde\btheta_{2n,\model_2} \|_{1} \leq \frac{|\model_2| ( \| \hat\bG_{2n} - \tilde\bG_{2n} \|_{\infty} + \| \hat\bH_{2n} - \bHOracle{2} \|_{\infty} \| \tilde\btheta_{2n,\model_2} \|_1 )}{ \Lambda_2(C_2) - C_2 \hat{D}_{2n}^H} \\
    \| \hat\btheta_{1n,\model_1 \hat\model_2} - \tilde\btheta_{1n,\model_1 \hat\model_2} \|_{1} \leq \frac{|\model_1| ( \| \hat\bG_{1n,\model_2^*} - \tilde\bG_{1n,\model_2^*} \|_{\infty} + \| \hat\bH_{1n} - \bHOracle{1} \|_{\infty} \| \tilde\btheta_{1n,\model_1 \model_2^*} \|_1 )}{ \Lambda_1(C_1) - C_1 \hat{D}_{1n}^H}.
  \end{align*}
  Consequently, under the conditions of \Cref{thm:uposi-coverage}, cross-fitting approximates the oracle estimators uniformly in probability over sparse models:
  \begin{align*}
    \max_{\model_2 \in \modelspace_2(C_2)} \| \hat\btheta_{2n,\model_2} - \tilde\btheta_{2n,\model_2} \|_{1} = o_p(n^{-1/2}) \\
    \max_{\model_1 \in \modelspace_1(C_1)} \| \hat\btheta_{1n,\model_1 \hat\model_2} - \tilde\btheta_{1n,\model_1 \hat\model_2} \|_{1} = o_p(n^{-1/2}).
  \end{align*}
\end{lemma}
\begin{proof}
  We may prove this using a similar argument as in \Cref{lem:kuchibhotla-uniform-model-conv}. We sketch out a few relevant expressions in the argument for the inequality in Stage 1.
  First, use \Cref{lem:stochastic-process-random-index} along with \Cref{assump:model-selection-consistency} to conclude that \[
    \| \hat\btheta_{1n,\model_1 \hat\model_2} - \tilde\btheta_{1n,\model_1 \hat\model_2} \|_1 = \| \hat\btheta_{1n,\model_1 \model_2^*} - \tilde\btheta_{1n,\model_1 \model_2^*} \|_1 + o_p\left(n^{-1/2}\right).
  \]
  The first-order equations ensure:
  \[
    \hat\btheta_{1n,\model_1 \model_2^*} - \tilde\btheta_{1n,\model_1 \model_2^*} = \left\{ \hat\bH_{1n}(\model_1) \right\}^{-1}
    \left[ 
      \{ \hat\bG_{1n,\model_2^*} - \tilde\bG_{1n,\model_2^*} \}(\model_1) - \{ \hat\bH_{1n} - \bHOracle{1} \}(\model_1) \tilde\btheta_{1n,\model_1 \model_2^*}
     \right].
  \]
  Similarly, the $\ell_2 \mapsto \ell_2$ operator norm bound below follows:
  \[
    \| \hat\bH_{1n}(\model_1) - \bHOracle{1}(\model_1) \|_{2,2} \leq | \model_1 | \| \hat\bH_{1n} - \bHOracle{1} \|_{\infty}.
  \]
  Using the previously-established bound on $\| \hat\bH_{1n}(\model_2)^{-1} \|_{2,2}, $ we can derive the expression 
  \begin{equation}
    \label{eq:app-ineq-crossfit-1}
    \| \hat\btheta_{1n,\model_1 \model_2^*} - \tilde\btheta_{1n,\model_1 \model_2^*} \|_{2} \leq |\model_1|^{1/2} \frac{
      \| \hat\bG_{1n,\model_2^*} - \tilde\bG_{1n,\model_2^*} \|_{\infty} + \| \hat\bH_{1n} - \bHOracle{1} \|_{\infty} \| \tilde\btheta_{1n,\model_1 \model_2^*} \|_1
    }{
      \Lambda_1(C_1) - C_1 \hat{D}_{1n}^H
    },
  \end{equation}
  which allows us to establish the first result using the relationship between $\ell_1$ and $\ell_2$ norms. 
  
  Now we may establish the uniform rate on $\| \hat\btheta_{1n,\model_1 \model_2^*} - \tilde\btheta_{1n,\model_1 \model_2^*} \|_{1}$ using this inequality.
  The term $\max_{\model_1 \in \modelspace_1(C_1)}\| \tilde\btheta_{1n,\model_1 \model_2^*} \|_1 = O_p(1)$ using Lemma 4.1 in \citet{kuchibhotlaValidPostselectionInference2020} and the two remaining terms in the numerator are $o_p(n^{-1/2})$ by \Cref{lem:unif-gram-inv,lem:approximation-cross-fit}, uniformly over $\modelspace_1(C_1).$ Finally, the denominator converges in probability to a term bounded from below by $c_0$ according to \Cref{assump:bounded-model}. The continuous mapping theorem implies that the RHS of \eqref{eq:app-ineq-crossfit-1} is $o_p(n^{-1/2}),$ completing the proof.
\end{proof}

\begin{lemma}
  \label{lem:uniform-influence-function-stg2}
  Under the conditions for \Cref{lem:unif-gram-inv,lem:unif-consistency-both-stages}, the following property holds for any $\model_2 \in \modelspace_2(C_2)$: 
  \begin{align}
    \label{eq:uniform-influence-function-stg2-1}
    \Big\| (\hat\btheta_{2n,\model_2} - \btheta_{20,\model_2}) - \frac{1}{n} \sum_{i=1}^n \mathrm{Inf}_{2 \model_2 i} \Big\|_{\infty} &= o_p\left( n^{-1/2} \right) 
    \\
    \label{eq:uniform-influence-function-stg2-2}
    \Big\| (\hat\btheta_{2n,\model_2}^b - \btheta_{20,\model_2}) - \frac{1}{n} \sum_{i=1}^n \bootWgt_i \mathrm{Inf}_{2 \model_2 i} \Big\|_{\infty} &= o_{p^*}\left( n^{-1/2} \right)
    \\
    \label{eq:uniform-influence-function-stg2-3}
    \Big\| (\hat\btheta_{2n,\model_2}^b - \hat\btheta_{2n,\model_2}) - \frac{1}{n} \sum_{i=1}^n (\bootWgt_i - 1) \mathrm{Inf}_{2 \model_2 i} \Big\|_{\infty} &= o_{p^*}\left( n^{-1/2} \right)
    ,
  \end{align}
  where we define the function \[
    \mathrm{Inf}_{2 \model_2 i} = \bH_{20}(\model_2)^{-1} (A_{2i} - \mu_{2A0i}) \historybasis{2i}(\model_2) \left\{ Y_i - \mu_{2Y0i} - (A_{2i} - \mu_{2A0i}) \historybasis{2i}(\model_2)^\top \btheta_{20,\model_2} \right\}.
  \]
\end{lemma}
\begin{proof}
  Begin with \eqref{eq:uniform-influence-function-stg2-1}.
  Under \Cref{lem:unif-consistency-both-stages}, we may replace $\hat\btheta_{2n,\model_2}$ in this expression by $\tilde\btheta_{2n,\model_2}$ without affecting the remainder. Doing so, we write
  \begin{equation*}
    \tilde\btheta_{2n,\model_2} - \btheta_{20,\model_2} = \bHOracle{2}(\model_2)^{-1} \tilde\bG_{2n}(\model_2) - \btheta_{20,\model_2}
    = \bHOracle{2}(\model_2)^{-1} \left\{ \tilde\bG_{2n}(\model_2) - \bHOracle{2}(\model_2) \btheta_{20,\model_2}  \right\}
  \end{equation*}
  Next, we express the sum
  \[
    \frac{1}{n} \sum_{i=1}^n \mathrm{Inf}_{2 \model_2 i} = \bH_{20}(\model_2)^{-1} \left\{ \tilde\bG_{2n}(\model_2) - \bHOracle{2}(\model_2) \btheta_{20,\model_2} \right\},
  \]
  and hence $R_{n, \model_2} = \tilde\btheta_{2n,\model_2} - \btheta_{20,\model_2} - \frac{1}{n} \sum_{i=1}^n \mathrm{Inf}_{2 \model_2 i}$ may be written
  \begin{equation}
    R_{n, \model_2} = \left\{ \bHOracle{2}(\model_2)^{-1} - \bH_{20}(\model_2)^{-1} \right\} \left\{ \tilde\bG_{2n}(\model_2) - \bHOracle{2}(\model_2) \btheta_{20,\model_2} \right\}.
    \label{eq:app-infl-remainder}
  \end{equation}
  The $\ell_0$ norm of the second quantity in curly braces is bounded by $C_2$ by the definition of $\modelspace_2(C_2).$
  Using the relationship $\|\bm{A}\bm{b}\|_{\infty} \leq \| \bm{A} \|_{\infty} \| \bm{b} \|_0 \| \bm{b} \|_{\infty}$ for matrix $\bm{A}$ and vector $\bm{b}$, we may bound the remainder \eqref{eq:app-infl-remainder} by \[
    C_2 \| \bHOracle{2}(\model_2)^{-1} - \bH_{20}(\model_2)^{-1} \|_{\infty} \| \tilde\bG_{2n}(\model_2) - \bHOracle{2}(\model_2) \btheta_{20,\model_2} \|_{\infty}. 
  \]
  Under \Cref{lem:matrix-inversion-diff}, the first normed quantity is $O_p(n^{-1/2})$ and so the proof will be complete if $\| \tilde\bG_{2n}(\model_2) - \bHOracle{2}(\model_2) \btheta_{20,\model_2} \|_{\infty} = o_p(1).$
  To this end, express the quantity \[
    \tilde\bG_{2n}(\model_2) - \bHOracle{2}(\model_2) \btheta_{20,\model_2} = \frac{1}{n} \sum_{i=1}^n (A_{2i} - \mu_{2A0i}) \history{2}(\model_2) \epsilon_{2,\model_2},
  \]
  which is a sample average of $|\model_2|-$dimensional random variables having mean zero. Consequently, the above term is $O_p(n^{-1/2})$.
  Consequently, the remainder term is $O_p(n^{-1}) = o_p(n^{-1/2}).$

  Next, we use similar arguments to handle \eqref{eq:uniform-influence-function-stg2-2}. Again, we may replace $\hat\btheta_{2n,\model_2}^b$ by the oracle term $\tilde\btheta_{2n,\model_2}^b.$ Writing the bootstrap error
  \[
    \tilde\btheta_{2n,\model_2}^b - \btheta_{20,\model_2} = \bHOracle{2}^b(\model_2)^{-1} \left\{ \tilde\tilde\bG_{2n}^b(\model_2) - \bHOracle{2}^b(\model_2) \btheta_{20,\model_2} \right\},
  \]
  along with the desired influence function-based sample average:
  \[
    \frac{1}{n} \sum_{i=1}^n \bootWgt_i \mathrm{Inf}_{2 \model_2 i} = \bH_{20}(\model_2)^{-1} \left\{ \tilde\bG_{2n}^b(\model_2) - \bHOracle{2}^b(\model_2) \btheta_{20,\model_2} \right\},
  \]
  we can express their difference with a similar remainder term:
  \begin{align*}
    \tilde\btheta_{2n,\model_2}^b - \btheta_{20,\model_2} &- \frac{1}{n} \sum_{i=1}^n \bootWgt_i \mathrm{Inf}_{2 \model_2 i} 
    \\ 
    &= \left\{ \bHOracle{2}^b(\model_2)^{-1} - \bH_{20}(\model_2)^{-1} \right\} \left\{ \tilde\bG_{2n}^b(\model_2) - \bHOracle{2}^b(\model_2) \btheta_{20,\model_2} \right\}
    .
  \end{align*}
  By similar arguments as in the \eqref{eq:uniform-influence-function-stg2-1} case, this remainder can be shown to be $o_{p^*}(n^{-1/2}).$

  The final expression \eqref{eq:uniform-influence-function-stg2-3} is a direct consequence of the two previous expressions, along with the triangle inequality.
\end{proof}

\subsubsection{Results Establishing the Asymptotic Negligibility of Cross-fitting in the First Stage}
\label{app:lemmas-stg-1-cf}

\begin{lemma}
  \label{lem:approximation-cross-fit}
  Under \Cref{assump:boundedness,assump:bounded-model,assump:rate-assumps}, the following rates hold:
  \begin{align}
    &\Vert \hat{\bG}_{1n,\model_2^*} - \tilde{\bG}_{1n,\model_2^*} \Vert_{\infty} = o_p(n^{-1/2})
    \label{eq:G1-approx} \\
    &\Vert \hat{\bG}^b_{1n,\model_2^*} - \tilde{\bG}^b_{1n,\model_2^*} \Vert_{\infty} = o_{p^*}(n^{-1/2}).
    \label{eq:G1-approx-boot}
  \end{align}
\end{lemma}
\begin{proof}

  We will follow the general strategy as in the proof of \Cref{prop:unif-g2}, which demonstrated similar properties in the second stage. Additionally, we must handle complexities arising from the pseudo-outcomes. To establish \eqref{eq:G1-approx}, we first write the sums
  \begin{align*}
    \hat\bG_{1n,\model_2^*} - \tilde\bG_{1n,\model_2^*} &= \frac{1}{n} \sum_{i=1}^n \bm{\varphi}_{1i}^G
    \\
    \hat\bG_{1n,\model_2^*}^b - \tilde\bG_{1n,\model_2^*}^b &= \frac{1}{n} \sum_{i=1}^n \bootWgt_i \bm{\varphi}_{1i}^G,
  \end{align*}
  where we define for $i=1,\ldots,n,$ the term inside the sum as
  \[
    \bm{\varphi}_{1i}^G \defined (\bZ_{1i} - \hat\bmu_{1Zi}) (\hat{Y}_{1 \model_2^* i} - \hat\mu_{1Y \model_2^* i}) - (\bZ_{1i} - \bmu_{1Z0i}) (\tilde{Y}_{1 \model_2^* i} - \mu_{1Y \model_2^* 0i}).
  \]
  Following along the previous proof, we expand the representation of $\bm{\varphi}_{1i}^G$ to find
  \begin{align}
    \nonumber
    \bm{\varphi}_{1i}^G = 
    (\bZ_{1i} - \bmu_{1Z0i})(\hat{Y}_{1 \model_2^* i} - \tilde{Y}_{1 \model_2^* i}) + (\bZ_{1i} - \bmu_{1Z0 i})( \mu_{1 Y \model_2^* 0} - \hat{\mu}_{1 Y \model_2^* i}) 
    \\
    + \label{eq:app-phi-diff-G-1}
    (\bmu_{1Z0 i} - \hat{\bmu}_{1Z i}) (\hat{Y}_{1 \model_2^* i} - \hat{\mu}_{1 Y \model_2^* i}).
  \end{align}
  This last line involves both the estimated conditional expectation as well as the estimated pseudo-outcome. Expand this problematic term into three summands: \[
    \hat{Y}_{1 \model_2^* i} - \hat{\mu}_{1 Y \model_2^* i} = (\hat{Y}_{1 \model_2^* i} - \tilde{Y}_{1 \model_2^* i}) + (\tilde{Y}_{1 \model_2^* i} - \mu_{1 Y \model_2^* 0 i}) + (\mu_{1 Y \model_2^* 0 i} - \hat{\mu}_{1 Y \model_2^* i}).
  \]
  Substituting this into \eqref{eq:app-phi-diff-G-1}, we obtain
  \begin{align}
    \bm{\varphi}_{1i}^G =~&
    (\bZ_{1i} - \bmu_{1Z0i})(\hat{Y}_{1 \model_2^* i} - \tilde{Y}_{1 \model_2^* i}) \label{eq:app-uposi-g-diff-1}
    \\
    &+ (\bmu_{1Z0 i} - \hat{\bmu}_{1Z i})(\hat{Y}_{1 \model_2^* i} - \tilde{Y}_{1 \model_2^* i}) \label{eq:app-uposi-g-diff-5}
    \\
    &+ 
    (\bZ_{1i} - \bmu_{1Z0 i})( \mu_{1 Y \model_2^* 0} - \hat{\mu}_{1 Y \model_2^* i}) \label{eq:app-uposi-g-diff-2}
    \\
    &+ \left\{ \historybasis{1i} (\mu_{1A0i} - \hat\mu_{1Ai}) \right\}(\tilde{Y}_{1 \model_2^* i} - \mu_{1 Y \model_2^* 0 i}) \label{eq:app-uposi-g-diff-3}
    \\
    &+ (\bmu_{1Z0 i} - \hat{\bmu}_{1Z i})(\mu_{1 Y \model_2^* 0 i} - \hat{\mu}_{1 Y \model_2^* i}), \label{eq:app-uposi-g-diff-4}
  \end{align}
  where we made the additional substitution $(\bmu_{1Z0 i} - \hat{\bmu}_{1Z i}) = \historybasis{1i} (\mu_{1A0i} - \hat\mu_{1Ai})$ in \eqref{eq:app-uposi-g-diff-3}. Analogous to the proof of \Cref{prop:unif-g2}, 
  \Cref{lem:mean-zero-rate-vec} applies to \eqref{eq:app-uposi-g-diff-2,eq:app-uposi-g-diff-3}, while \Cref{lem:cross-rate-vec} applies to \eqref{eq:app-uposi-g-diff-4}. \Cref{assump:rate-assumps} as well as \Cref{prop:uniform-mu-Z-rates} ensure that \eqref{eq:app-uposi-g-diff-2}-\eqref{eq:app-uposi-g-diff-4} are $o_p(n^{-1/2})$ when averaged over $i=1,\ldots,n.$ The outcome of these lemmas are not impacted by the random multipliers.

  To examine the rates of \eqref{eq:app-uposi-g-diff-1} and \eqref{eq:app-uposi-g-diff-5}, we re-write the expression
  \[
    \hat{Y}_{1 \model_2^*} - \tilde{Y}_{1 \model_2^*} = - A_2 \history{2}^\top (\hat\btheta_{2n,\model_2^*} - \tilde\btheta_{2n,\model_2^*}) + \{\historybasis{2}(\model_2^*)^\top \hat\btheta_{2n,\model_2^*}\}_+ - \{\historybasis{2}(\model_2^*)^\top \tilde\btheta_{2n,\model_2^*}\}_+.
  \]
  Now we use the triangle inequality and the inequality $| A_+ - B_+ | \leq 2| A - B |$:
  \begin{align*}
    | \hat{Y}_{1 \model_2^*} - \tilde{Y}_{1 \model_2^*} | \leq | A_2 \historybasis{2}(\model_2^*)^\top (\hat\btheta_{2n,\model_2^*} - \tilde\btheta_{2n,\model_2^*})| + |\{\historybasis{2}(\model_2^*)^\top \hat\btheta_{2n,\model_2^*}\}_+ - \{\historybasis{2}(\model_2^*)^\top \tilde\btheta_{2n,\model_2^*}\}_+| \\
    \leq C \| \hat\btheta_{2n,\model_2^*} - \tilde\btheta_{2n,\model_2^*} \|_1 + 2 C \| \hat\btheta_{2n,\model_2^*} - \tilde\btheta_{2n,\model_2^*} \|_1.
  \end{align*}
  This is a uniform rate over all the observed samples; as such, the maximum $\max_{i=1,\ldots,n} |\hat{Y}_{1 \model_2^* i} - \tilde{Y}_{1 \model_2^* i}|$ is uniformly bounded by some constant multiple of the difference $\| \hat\btheta_{2n,\model_2^*} - \tilde\btheta_{2n,\model_2^*} \|_1.$
  Use \Cref{lem:unif-consistency-both-stages} to conclude that $\max_{i=1,\ldots,n} |\hat{Y}_{1 \model_2^* i} - \tilde{Y}_{1 \model_2^* i}| = o_p(n^{-1/2}).$
  Multiply the sample average of the terms in \eqref{eq:app-uposi-g-diff-1} by $\sqrt{n}$ and take the modulus to obtain
  \begin{equation}
    \label{eq:G1-approx-1-rhs}
    \left| n^{-1/2} \sum_{i=1}^n (\bZ_{1i} - \bmu_{1Z0i}) (\hat{Y}_{1 \model_2^* i} - \tilde{Y}_{1 \model_2^* i}) \right| \leq n^{-1/2} \sum_{i=1}^n \left\| \bZ_{1i} - \bmu_{1Z0i} \right\|_{\infty} |\hat{Y}_{1 \model_2^* i} - \tilde{Y}_{1 \model_2^* i} |
  \end{equation}
  Because $\|\bZ_{1i} - \bmu_{1Z0i}\|_{\infty} \leq C$ and the maximum is taken over finitely-many elements, its sample mean converges. Furthermore, the bound on the maximum of the remaining factor of the RHS ensures that the conditions of \Cref{lem:sum-with-op} are satisfied. Consequently, \eqref{eq:G1-approx-1-rhs} is $o_p(1).$ The conclusion of this lemma likewise does not change in the presence of random multipliers; to see this, notice that the average \[
    \frac{1}{n}\sum_{i=1}^n \bootWgt_i \|\bZ_{1i} - \bmu_{1Z0i}\|_{\infty}
  \]
  also converges in probability due to independence of $\bootWgt_i$ and the other term.
  
  A very similar argument can be used for \eqref{eq:app-uposi-g-diff-5}. Writing the desired term and applying similar bounds as in \eqref{eq:G1-approx-1-rhs}, we need to show that
  \begin{equation*}
    n^{-1/2} \sum_{i=1}^n \left\| \bmu_{1Z0 i} - \hat{\bmu}_{1Z i} \right\|_{\infty} |\hat{Y}_{1 \model_2^* i} - \tilde{Y}_{1 \model_2^* i} |
  \end{equation*}
  converges in probability to zero, which would follow from \Cref{lem:sum-with-op} if
  \[
    \frac{1}{n} \sum_{i=1}^n \left\| (\mu_{1A0 i} - \hat\mu_{1A i}) \historybasis{1i} \right\|_{\infty} \leq C \frac{1}{n} \sum_{k=1}^K \sum_{i \in \Ik} | \mu_{1A0 i} - \hat\mu_{1A i} |
  \]
  converges in probability to some constant, where we have re-written the sum on the RHS to draw attention to the fold structure. We can apply Markov's inequality to each of the $K$ folds to ensure that this piece is $o_p(1)$. This also holds in the presence of random multipliers: for the $k^{th}$ fold, we have \[
    \frac{1}{n_k} \sum_{i \in \Ik} |\bootWgt_i| | \mu_{1A0 i} - \hat\mu_{1A i} | \overset{p^*}{\longrightarrow} \Eop |G|     \Eop\left\{ |\mu_{1A0} - \hat\mu_{1A}| ~\big|~\crossfitDataC \right\},
  \]
  where we recall that $n_k \defined |\Ik|$ and $n_k / n = K + o_p(1).$ By \Cref{assump:rate-assumps}, this last quantity converges to zero.

  We have handled each of the terms \eqref{eq:app-uposi-g-diff-1,eq:app-uposi-g-diff-2,eq:app-uposi-g-diff-3,eq:app-uposi-g-diff-4,eq:app-uposi-g-diff-5} when averaged over the $n$ indices, also accounting for the addition random multipliers. This completes the proof.
\end{proof}



\begin{lemma}
  \label{lem:g1-hat-rand-model}
  Under \Cref{assump:boundedness,assump:bounded-model,assump:rate-assumps,assump:model-selection-consistency,assump:regularity}, we have
  \begin{align*}
    \| \hat\bG_{1n,\hat\model_2} - \bG_{10,\model_2^*} \|_{\infty} = \| \tilde\bG_{1n,\model_2^*} - \bG_{10,\model_2^*} \|_{\infty} + o_p(n^{-1/2}) \\
    \| \hat\bG^b_{1n,\hat\model_2} - \hat\bG_{1n,\hat\model_2} \|_{\infty} = \| \tilde\bG^b_{1n,\model_2^*} - \tilde\bG_{1n,\model_2^*} \|_{\infty} + o_{p^*}(n^{-1/2})
  \end{align*}
\end{lemma}
\begin{proof}
  Decompose the first term as \[
    \| \hat\bG_{1n,\hat\model_2} - \bG_{10,\model_2^*} \|_{\infty} = \| \hat\bG_{1n,\hat\model_2} - \hat\bG_{1n,\model_2^*} \|_{\infty} + \| \hat\bG_{1n,\model_2^*} - \tilde\bG_{1n,\model_2^*} \|_{\infty} + \| \tilde\bG_{1n,\model_2^*} - \bG_{10,\model_2^*} \|_{\infty}.
  \]
  We need to show that the first two terms are $o_p(n^{-1/2}).$ Viewing $\left\{ \hat\bG_{1n,\model_2} : \model_2 \in \modelspace_2 \right\}$ as a stochastic process taking values in the norm-induced metric space $(\R^{p_1}, \| \cdot \|_{\infty}),$ notice that it is measurable given the observed data $\bO_1,\ldots,\bO_n,$ the partition $\Pcal_K,$ and any additional randomness in the machine learning procedures which yield cross-fitted estimates. Given all of this information, the model selection procedure $\hat\model_2$ is also measurable. Consequently, \Cref{lem:stochastic-process-random-index} applies, establishing that the first term is $o_p(n^{-1/2}).$
  The second term is $o_p(n^{-1/2})$ by \Cref{lem:approximation-cross-fit}, which gives the desired result.

  This same argument establishes the result for the bootstrap version.
\end{proof}

\subsubsection{Results Establishing an Influence Function in the First Stage}
\label{app:lemmas-stg-1-inf}

Now, we need to show that the UPoSI bootstrap results apply in stage 1. To this end, recall the definition in \Cref{app:defining-pseudo-outcomes} of $B_{\model_2} \defined \Ind{(\historybasis{2}(\model_2)^\top \btheta_{20,\model_2} > 0)} - A_{2}.$ First, we need the blip function to behave like a smooth function:
\begin{lemma}[The blip function does not have unsmooth behavior]
  \label{lem:blip-inf}
  Let $\bootWgt_i$ be i.i.d. random multipliers as described in \Cref{sec:cf-perturb-risk}. Further, let the random variable $\bM \in \R^{p_1 \times | \model_2^* |}$ be defined according to the expression \[
    \bM \defined B_{\model_2^*} \left\{ A_{1} - \mu_{1A0}(\history{1}) \right\} \historybasis{1} \historybasis{2}(\model_2^*)^\top,
  \]
  with $\bM_i,~i=1,\ldots,n$ the observed realizations from $\bO_1,\ldots,\bO_n.$
  For both $v=0$ and $v=1$ under \Cref{assump:boundedness,assump:bounded-model,assump:rate-assumps,assump:model-selection-consistency,assump:regularity},
  \begin{align}
    \nonumber
    \frac{1}{n} \sum_{i=1}^n \bootWgt_i^v \historybasis{1i} (A_{1i} - \mu_{1A0i}) 
    \left[
      \xi\left\{A_{2i}, \historybasis{2i}(\model_2^*); \tilde\btheta_{2n,\model_2^*} \right\} - \xi\left\{A_{2i}, \historybasis{2i}(\model_2^*); \btheta_{20,\model_2^*} \right\}
    \right] \\ 
    \label{eq:blip-inf-1}
    = \frac{1}{n} \sum_{i=1}^n \bootWgt_i^v \bM_i \left( \tilde\btheta_{2n,\model_2^*} - \btheta_{20,\model_2^*} \right) +  o_{p^*}(n^{-1/2}) \\
    \label{eq:blip-inf-2}
    = \Eop( \bM) \frac{1}{n} \sum_{i=1}^n \mathrm{Inf}_{2 \model_2^* i} +  o_{p^*}(n^{-1/2}).
  \end{align}
  Similarly, the perturbation bootstrap version satisfies
  \begin{align}
    \nonumber
    \frac{1}{n} \sum_{i=1}^n \bootWgt_i^v \historybasis{1i} (A_{1i} - \mu_{1A0i}) \left[
      \xi\left\{A_{2i}, \historybasis{2i}(\model_2^*); \tilde\btheta_{2n,\model_2^*}^b \right\} - \xi\left\{A_{2i}, \historybasis{2i}(\model_2^*); \btheta_{20,\model_2^*} \right\}
    \right]
    \\
    \label{eq:blip-inf-3}
    = \Eop( \bM) \frac{1}{n} \sum_{i=1}^n \bootWgt_i \mathrm{Inf}_{2 \model_2^* i} +  o_{p^*}(n^{-1/2}).
  \end{align}
\end{lemma}
\begin{proof}
  The arguments for \eqref{eq:blip-inf-1} are exactly the same as those around equations (38) and (39) in the Supplement of \citet{ertefaieRobustQLearning2021}. 
  
  Using the definitions of $\tilde{R}_{ni} = \Ind{ \left( \historybasis{2i}(\model_2^*)^\top \tilde\btheta_{2n,\model_2^*} > 0 \right)} - \Ind{ \left( \historybasis{2i}(\model_2^*)^\top \btheta_{20,\model_2^*} > 0 \right)}$ and $R_{ni} = \Ind{\left\{ 0 \leq | \historybasis{2i}(\model_2^*)^\top\btheta_{20,\model_2^*} | \leq |\historybasis{2i}(\model_2^*)^\top (\tilde\btheta_{2n,\model_2^*} - \btheta_{20,\model_2^*})| \right\}}$, which satisfies $|\tilde{R}_{ni}|\leq R_{ni}$, we expand the first term:
  \begin{align}
    \nonumber
    \frac{1}{n} \sum_{i=1}^n& \bootWgt_i^v \historybasis{1i} (A_{1i} - \mu_{1A0i}) \left[
      \xi\left\{A_{2i}, \historybasis{2i}(\model_2^*); \tilde\btheta_{2n,\model_2^*} \right\} - \xi\left\{A_{2i}, \historybasis{2i}(\model_2^*); \btheta_{20,\model_2^*} \right\}
    \right]
    \\
    \label{eq:blip-inf-eq-1}
    &= \frac{1}{n} \sum_{i=1}^n \bootWgt_i^v \historybasis{1i} (A_{1i} - \mu_{1A0i}) B_{\model_2^* i} \historybasis{2i}(\model_2^*)^\top (\tilde\btheta_{2n,\model_2^*} - \btheta_{20,\model_2^*})
    \\
    \label{eq:blip-inf-eq-2}
    &+ \frac{1}{n} \sum_{i=1}^n \bootWgt_i^v \historybasis{1i} (A_{1i} - \mu_{1A0i}) \historybasis{2i}(\model_2^*)^\top\btheta_{20,\model_2^*} \tilde{R}_{ni}
    \\
    \label{eq:blip-inf-eq-3}
    &+ \frac{1}{n} \sum_{i=1}^n \bootWgt_i^v \historybasis{1i} (A_{1i} - \mu_{1A0i}) \historybasis{2i}(\model_2^*)^\top (\tilde\btheta_{2n,\model_2^*} - \btheta_{20,\model_2^*}) \tilde{R}_{ni}.
  \end{align}
  Equation \eqref{eq:blip-inf-eq-2} can be bounded using 
  \begin{align*}
    \left\| \frac{1}{n} \sum_{i=1}^n \bootWgt_i^v \historybasis{1i} (A_{1i} - \mu_{1A0i}) \historybasis{2i}(\model_2^*)^\top\btheta_{20,\model_2^*} \tilde{R}_{ni} \right\|_{\infty}
    \\ 
    \leq C \frac{1}{n} \sum_{i=1}^n |\bootWgt_i^v| |\historybasis{2i}(\model_2^*)^\top (\tilde\btheta_{2n,\model_2^*} - \btheta_{20,\model_2^*})| R_{ni},
  \end{align*}
  which follows from the triangle inequality, $\| \historybasis{1i} (A_{1i} - \mu_{1A0i}) \|_{\infty} \leq C,$ and \[
    |\historybasis{2i}(\model_2^*)^\top\btheta_{20,\model_2^*} \tilde{R}_{ni} | \leq |\historybasis{2i}(\model_2^*)^\top (\tilde\btheta_{2n,\model_2^*} - \btheta_{20,\model_2^*})| R_{ni}.
  \]
  A similar argument bounds \eqref{eq:blip-inf-eq-3}: 
  \begin{align*}
    \left\| \frac{1}{n} \sum_{i=1}^n \bootWgt_i^v \historybasis{1i} (A_{1i} - \mu_{1A0i}) \historybasis{2i}(\model_2^*)^\top (\tilde\btheta_{2n,\model_2^*} - \btheta_{20,\model_2^*}) \tilde{R}_{ni} \right\|_{\infty} 
    \\
    \leq C \frac{1}{n} \sum_{i=1}^n |\bootWgt_i^v| |\historybasis{2i}(\model_2^*)^\top (\tilde\btheta_{2n,\model_2^*} - \btheta_{20,\model_2^*})| R_{ni}.
  \end{align*}
  Apply the Cauchy-Schwarz inequality to find
  \begin{align*}
    \frac{1}{n} \sum_{i=1}^n |\bootWgt_i^v| |\historybasis{2i}(\model_2^*)^\top (\tilde\btheta_{2n,\model_2^*} - \btheta_{20,\model_2^*})| R_{ni} 
    \\
    \leq C \left[  \frac{1}{n} \sum_{i=1}^n |\bootWgt_i^v| |\historybasis{2i}(\model_2^*)^\top (\tilde\btheta_{2n,\model_2^*} - \btheta_{20,\model_2^*})|^2 \right]^{1/2} \left[ \frac{1}{n} \sum_{i=1}^n R_{ni}^2 \right]^{1/2}.
  \end{align*}
  Using Holder's inequality, we can bound $|\historybasis{2i}(\model_2^*)^\top (\tilde\btheta_{2n,\model_2^*} - \btheta_{20,\model_2^*})| \leq \| \historybasis{2i}(\model_2^*)^\top \|_{\infty} \| \tilde\btheta_{2n,\model_2^*} - \btheta_{20,\model_2^*} \|_1.$ Markov's inequality and the independence of $\bootWgt_i$ and $\bO_i$ can be used to show that the first bracketed quantity is $O_{p^*}(\| \tilde\btheta_{2n,\model_2^*} - \btheta_{20,\model_2^*} \|_1)$. By the arguments on page 14 of the supplement of \citet{ertefaieRobustQLearning2021}, the second term is $o_{p^*}(1).$ Consequently, the rate \(\| \tilde\btheta_{2n,\model_2^*} - \btheta_{20,\model_2^*} \|_1 = O_{p^*}(n^{-1/2})\) demonstrates that \eqref{eq:blip-inf-eq-2}-\eqref{eq:blip-inf-eq-3} are each $o_{p^*}(n^{-1/2}).$ This demonstrates the desired remainder term \eqref{eq:blip-inf-1}.

  Next, we establish \eqref{eq:blip-inf-2}.
  We can write \eqref{eq:blip-inf-eq-1} as an i.i.d. sum of a matrix-valued random variable $\bM_i \in \R^{p_1 \times | \model_2^* |}$ with $\| \bM \|_{\infty} \leq C^2$:
  \[
    \left( n^{-1}\sum_{i=1}^n \bootWgt_i^v \bM_i \right) \left( \tilde\btheta_{2n,\model_2^*} - \btheta_{20,\model_2^*} \right).
  \]
  Adding and subtracting \(\Eop\left( G^v \bM \right) \frac{1}{n} \sum_{i=1}^n \mathrm{Inf}_{2 \model_2^* i},\) equation \eqref{eq:blip-inf-eq-1} equals:
  \begin{align*}
    \Eop\left( G^v \bM \right) \frac{1}{n} \sum_{i=1}^n \mathrm{Inf}_{2 \model_2^* i} + \bR_{n G1 \model_2^*} \\ 
    \bR_{n G1 \model_2^*} \defined \left( n^{-1}\sum_{i=1}^n \bootWgt_i^v \bM_i \right) \left( \tilde\btheta_{2n,\model_2^*} - \btheta_{20,\model_2^*} \right) - \Eop\left( G^v \bM \right) \frac{1}{n} \sum_{i=1}^n \mathrm{Inf}_{2 \model_2^* i}.
  \end{align*}
  We arrive at the result if we show the remainder satisfies $\| \bR_{n G1 \model_2^*} \|_{\infty} = o_{p^*}(n^{-1/2}).$ Expanding the remainder, we have
  \begin{align*}
    \bR_{n G1 \model_2^*} = 
    \left\{ n^{-1}\sum_{i=1}^n \bootWgt_i^v \bM_i - \Eop\left( G^v \bM \right) \right\} \frac{1}{n} \sum_{i=1}^n \mathrm{Inf}_{2 \model_2^* i} \\
    + \Eop\left( G^v \bM \right) \left( \tilde\btheta_{2n,\model_2^*} - \btheta_{20,\model_2^*} - \frac{1}{n} \sum_{i=1}^n \mathrm{Inf}_{2 \model_2^* i} \right) \\
    + \left\{ n^{-1}\sum_{i=1}^n \bootWgt_i^v \bM_i - \Eop\left( G^v \bM \right) \right\} \left( \tilde\btheta_{2n,\model_2^*} - \btheta_{20,\model_2^*} - \frac{1}{n} \sum_{i=1}^n \mathrm{Inf}_{2 \model_2^* i} \right).
  \end{align*}
  Next, apply the triangle inequality to the $\ell_{\infty}$ norm and use
  the property $\| \bM \bv \|_{\infty} \leq \| \bM \|_{\infty} \| \bv \|_{1} \leq \| \bM \|_{\infty} \| \bv \|_{0} \| \bv \|_{\infty}$ for any matrix $\bM$ and vector $\bv$ to find
  \begin{align*}
    \left\| \bR_{n G1 \model_2^*} \right\|_{\infty} = |\model_2^*|
    \left\| n^{-1}\sum_{i=1}^n \bootWgt_i^v \bM_i - \Eop\left( G^v \bM \right) \right\|_{\infty} \left\| \frac{1}{n} \sum_{i=1}^n \mathrm{Inf}_{2 \model_2^* i} \right\|_{\infty} \\
    + |\model_2^*| \| \Eop\left( G^v \bM \right) \|_{\infty} \left\| \tilde\btheta_{2n,\model_2^*} - \btheta_{20,\model_2^*} - \frac{1}{n} \sum_{i=1}^n \mathrm{Inf}_{2 \model_2^* i} \right\|_{\infty} \\
    + |\model_2^*| \left\| n^{-1}\sum_{i=1}^n \bootWgt_i^v \bM_i - \Eop\left( G^v \bM \right) \right\|_{\infty} \left\| \tilde\btheta_{2n,\model_2^*} - \btheta_{20,\model_2^*} - \frac{1}{n} \sum_{i=1}^n \mathrm{Inf}_{2 \model_2^* i} \right\|_{\infty}.
  \end{align*}

  By \Cref{lem:uniform-influence-function-stg2}, \(\left\| \tilde\btheta_{2n,\model_2^*} - \btheta_{20,\model_2^*} - \frac{1}{n} \sum_{i=1}^n \mathrm{Inf}_{2 \model_2^* i} \right\|_{\infty} = o_{p^*}(n^{-1/2}).\) By the law of large numbers, \(\left\| \frac{1}{n} \sum_{i=1}^n \mathrm{Inf}_{2 \model_2^* i} \right\|_{\infty} = o_{p^*}(1).\) If \(\left\| n^{-1}\sum_{i=1}^n \bootWgt_i^v \bM_i - \Eop\left( G^v \bM \right) \right\|_{\infty} = O_{p^*}(n^{-1/2}),\) then $\| \bR_{n G1 \model_2^*} \|_{\infty} = o_{p^*}(n^{-1/2}).$ Since $\| \bM \|_{\infty} \leq C^2,$ the term being analyzed is the maximum over a finite number of sample means with expectation zero and finite variance. Consequently, the required rate holds, establishing \eqref{eq:blip-inf-2}.

  The equation \eqref{eq:blip-inf-3} follows by similar arguments. Notably, we can define $\tilde{R}_{ni}^b$ and $R_{ni}^b$ which are defined similarly to the non-bootstrapped versions $\tilde{R}_{ni}$ and $R_{ni}$, respectively. The definitions of these bootstrap analogs, replace $\tilde\btheta_{2n,\model_2^*}$ wherever it appears by $\tilde\btheta_{2n,\model_2^*}^b.$ Following through with the arguments, we merely need to show that \[
    \frac{1}{n} \sum_{i=1}^n R_{ni}^b = o_{p^*}(1).
  \]
  Using Markov's inequality, this follows if $\Eop(R_{n1}^b) \rightarrow 0$. Given its definition,
  \begin{align*}
    \Eop\left( R_{n1}^b  \right) 
    = \Prob\left\{ 0 \leq |\historybasis{2}(\model_2^*)^\top \btheta_{20,\model_2}| \leq |\historybasis{2}(\model_2^*)^\top (\tilde\btheta_{2n,\model_2}^b - \btheta_{20,\model_2})| \right\}.
  \end{align*}
  By the influence function representation of \Cref{lem:uniform-influence-function-stg2}, $\tilde\btheta_{2n,\model_2}^b - \btheta_{20,\model_2} = o_{p^*}(1)$. Consequently, if $|\historybasis{2}(\model_2^*)^\top \btheta_{20,\model_2}| > 0$ almost surely, then for any $\delta >0$ we may find $\gamma>0$ such that \(
    \Prob\left( |\historybasis{2}(\model_2^*)^\top \btheta_{20,\model_2}| < \gamma \right) \leq \delta.
  \) Notice that we are assured $\gamma>0$ since $|\historybasis{2}(\model_2^*)^\top \btheta_{20,\model_2}| > 0$ almost surely.
  Then bound the expectation \[
    \Eop\left( R_{n1}^b  \right) \leq \Prob\left\{ 0 \leq |\historybasis{2}(\model_2^*)^\top \btheta_{20,\model_2}| \leq |\historybasis{2}(\model_2^*)^\top (\tilde\btheta_{2n,\model_2}^b - \btheta_{20,\model_2})| , ~ |\historybasis{2}(\model_2^*)^\top \btheta_{20,\model_2}| \geq \gamma\right\} + \delta.
  \]
  The probability on the RHS is bounded by \[
    \Prob\left\{ 0 \leq \gamma \leq |\historybasis{2}(\model_2^*)^\top (\tilde\btheta_{2n,\model_2}^b - \btheta_{20,\model_2})| , ~ |\historybasis{2}(\model_2^*)^\top \btheta_{20,\model_2}| \geq \gamma\right\} \rightarrow 0,
  \]
  where the convergence results from the consistency of the bootstrap estimator.
  Since $\delta>0$ is arbitrary, this ensures that $\Eop\left( R_{n1}^b  \right) \rightarrow 0.$

  Continuing along in the analogous argument, we have shown that the top line of \eqref{eq:blip-inf-3} is equal to \[
    \left( n^{-1}\sum_{i=1}^n \bootWgt_i^v \bM_i \right) \left( \tilde\btheta^b_{2n,\model_2^*} - \btheta_{20,\model_2^*} \right).
  \]
  The remaining arguments follow as before, using the influence function representation of $n^{-1} \sum_{i=1}^n \bootWgt_i \mathrm{Inf}_{2 \model_2^* i}$ resuling from \Cref{lem:uniform-influence-function-stg2}.

\end{proof}

\begin{lemma}
  \label{lem:inf-g-stg-1}
  Let the random variable $\bM$ be defined as in \Cref{lem:blip-inf}.
  Define the function \[ 
    \mathrm{Inf}_{1 G \model_2^* i} \defined \historybasis{1i} (A_{1i} - \mu_{1A0i}) \epsilon_{1 \model_2^* i} + \Eop(\bM) \mathrm{Inf}_{2 \model_2^* i}.
    \]
  Under \Cref{assump:boundedness,assump:bounded-model,assump:rate-assumps,assump:model-selection-consistency,assump:regularity},
  \begin{align}
    \label{eq:inf-g-stg-1-statement-1}
    \tilde\bG_{1n,\model_2^*} - \bG_{10,\model_2^*} = \frac{1}{n} \sum_{i=1}^n \mathrm{Inf}_{1 G \model_2^* i} + o_p(n^{-1/2}) \\
    \label{eq:inf-g-stg-1-statement-2}
    \tilde\bG_{1n,\model_2^*}^b - \tilde\bG_{1n,\model_2^*} = \frac{1}{n} \sum_{i=1}^n (\bootWgt_i - 1) \mathrm{Inf}_{1 G \model_2^* i} + o_{p^*}(n^{-1/2})
  \end{align}
\end{lemma}
\begin{proof}
  Represent the quantity $\tilde\bG_{1n,\model_2^*} - \bG_{10,\model_2^*}$ as a sum:
  \begin{equation}
    \label{eq:inf-g-stg-1-1}
    \tilde\bG_{1n,\model_2^*} - \bG_{10,\model_2^*} = \frac{1}{n} \sum_{i=1}^n \historybasis{1i} (A_{1i} - \mu_{1A0i}) \left\{ \epsilon_{1 \model_2^* i} + (\tilde{Y}_{1 \model_2^* i} - Y_{1 \model_2^* i}) \right\}.
  \end{equation}
  The difference \( \tilde{Y}_{1 \model_2^* i} - Y_{1 \model_2^* i} \) can be represented as the difference in the blips: \[
    \tilde{Y}_{1 \model_2^* i} - Y_{1 \model_2^* i} = 
    \xi\left\{A_{2i}, \historybasis{2i}(\model_2^*); \tilde\btheta_{2n,\model_2^*} \right\} - \xi\left\{A_{2i}, \historybasis{2i}(\model_2^*); \btheta_{20,\model_2^*} \right\}.
  \]
  The term resulting from this blip difference is analyzed in \Cref{lem:blip-inf}, resulting in the second term in the definition of $\mathrm{Inf}_{1G \model_2^* i}$, along with an $o_p(n^{-1/2})$ remainder. Plugging this in \eqref{eq:inf-g-stg-1-1} yields the first result, \eqref{eq:inf-g-stg-1-statement-1}.

  Similarly, we consider the bootstrap version \[
    \tilde\bG_{1n,\model_2^*}^b - \bG_{10,\model_2^*} = \frac{1}{n} \sum_{i=1}^n \bootWgt_i \historybasis{1i} (A_{1i} - \mu_{1A0i}) \left\{ \epsilon_{1 \model_2^* i} + (\tilde{Y}_{1 \model_2^* i}^b - Y_{1 \model_2^* i}) \right\},
  \]
  where the bootstrap pseudo-outcomes \( \tilde{Y}_{1 \model_2^* i}^b := Y + \xi(\historybasis{2i}(\model_2^*); \tilde\btheta_{2n,\model_2^*}^b) \) use the bootstrapped estimator $\tilde\btheta_{2n,\model_2^*}^b.$ Subtracting off \eqref{eq:inf-g-stg-1-1} from the previous display, we arrive at the representation 
  \begin{align}
    \nonumber
    \tilde\bG_{1n,\model_2^*}^b - \tilde\bG_{1n,\model_2^*} 
    = \frac{1}{n} \sum_{i=1}^n (\bootWgt_i - 1) \historybasis{1i} (A_{1i} - \mu_{1A0i}) \epsilon_{1 \model_2^* i} \\
    \label{eq:inf-g-stg-1-2}
    + \frac{1}{n} \sum_{i=1}^n \historybasis{1i} (A_{1i} - \mu_{1A0i}) \left\{ \bootWgt_i (\tilde{Y}_{1 \model_2^* i}^b - Y_{1 \model_2^* i}) - (\tilde{Y}_{1 \model_2^* i} - Y_{1 \model_2^* i}) \right\}
  \end{align}
  In light of \eqref{eq:blip-inf-2} in \Cref{lem:blip-inf}, and using the definition of $\bM$ therein,
  \[
    \frac{1}{n} \sum_{i=1}^n \historybasis{1i} (A_{1i} - \mu_{1A0i}) (\tilde{Y}_{1 \model_2^* i} - Y_{1 \model_2^* i}) = (\Eop\bM) \frac{1}{n} \sum_{i=1}^n \mathrm{Inf}_{2 \model_2^* i} + o_{p^*}(n^{-1/2}).
  \]
  Similarly, \eqref{eq:blip-inf-3} in this same lemma ensures that 
  \[
    \frac{1}{n} \sum_{i=1}^n \bootWgt_i \historybasis{1i} (A_{1i} - \mu_{1A0i}) (\tilde{Y}_{1 \model_2^* i}^b - Y_{1 \model_2^* i}) = (\Eop\bM) \frac{1}{n} \sum_{i=1}^n \bootWgt_i \mathrm{Inf}_{2 \model_2^* i} + o_{p^*}(n^{-1/2}).
  \]
  Subtracting these representations, we obtain
  \begin{align*}
    \frac{1}{n} \sum_{i=1}^n \historybasis{1i} (A_{1i} - \mu_{1A0i}) \left\{ \bootWgt_i (\tilde{Y}_{1 \model_2^* i}^b - Y_{1 \model_2^* i}) - (\tilde{Y}_{1 \model_2^* i} - Y_{1 \model_2^* i}) \right\}\\ 
    = (\Eop\bM) \frac{1}{n} \sum_{i=1}^n (\bootWgt_i - 1) \mathrm{Inf}_{2 \model_2^* i} + o_{p^*}(n^{-1/2}).
  \end{align*}
  The result \eqref{eq:inf-g-stg-1-statement-2} can be found by plugging the above representation into \eqref{eq:inf-g-stg-1-2}.

\end{proof}

\section{Notes Regarding the Risk Functions in Stage 1}
\label{app:sec-stg1-risk-deriv}

In this section we provide more details in our justification of the limiting estimating equation \eqref{eq:norm-eq-true-1} in \Cref{sec:rql-param-oracle}. This derivation requires some additional care due to the dependence of the pseudo-outcomes $\tilde{Y}_{1\model_2}$ on the entire sample through $\tilde\btheta_{2n,\model_2}$.

In the Stage 2 problem, we began with the oracle risk functions $R_{2n,\model_2}(\btheta_{2,\model_2})$ and determined the limiting parameter $\btheta_{20,\model_2}$ by minimizing the expected risk---namely, we derived a closed-form representation for the function
\[
  \Eop{ \left[R_{2n,\model_2}(\btheta_{2,\model_2}) - R_{2n,\model_2}(\bzero) \right]  },
\]
and found a corresponding minimizer $\btheta_{20,\model_2}$ of this function in the $\R^{|\model_2|}-$valued argument $\btheta_{2,\model_2}$. This line of argument benefitted from the fact that the ``expected function'' did not change with $n$, and that a closed form was available.

In the Stage 1 problem, the analogous ``expected risk'' function
\[
  \Eop{ \left[ R_{1n,\model_1,\model_2}(\btheta_{1,\model_1}) - R_{1n,\model_1,\model_2}(\bzero) \right] }
\]
is not constant in $n$, and does not admit a simple closed-form solution. Nonetheless, the argument is not altered significantly if we instead take the limiting expectation as $n\rightarrow\infty$. In fact, the following lemma shows that such a line of reasoning leads to a desirable form which can be directly minimized, leading to
the ``expected normal equation'' system given in \eqref{eq:norm-eq-true-1}.

\begin{lemma}
  \label{lem:pop-risk-stg1}
  Suppose \Cref{assump:bounded-model,assump:boundedness} hold and that $\model_2 \in \modelspace_2$ satisfies \Cref{assump:regularity}. Then
  \[
    \lim_{n\rightarrow \infty} \Eop{ \left[ R_{1n,\model_1,\model_2}(\btheta_{1,\model_1}) - R_{1n,\model_1,\model_2}(\bzero) \right] } = -2 \bG_{10,\model_2}(\model_1)^\top \btheta_{1,\model_2} + \btheta_{1,\model_2}^\top \bH_{10}(\model_1) \btheta_{1,\model_2}.
  \]
\end{lemma}


\begin{proof}
Beginning with \eqref{eq:oracle-risk-decomp-1}, notice that the result is shown if we can determine that $\lim_{n\rightarrow\infty} \Eop \tilde\bG_{1n,\model_2} = \bG_{10,\model_2}$ and $\lim_{n\rightarrow\infty}\Eop \tilde\bH_{1n} = \bH_{10}$. The latter is true, since $\Eop \tilde\bH_{1n} = \bH_{10}$ for all $n$. In the remainder, we derive the result for $\bG_{10,\model_2}$.

Using similar derivations as in the stage 2 case, we have
\begin{align*}
  Y_{1 \model_2} - \mu_{1Y \model_2 0}(\history{1}) = \left\{ A_1 - \mu_{1A0}(\history{1}) \right\} \Delta_{1,\model_2}(\history{1}) + \varepsilon_{1,\model_2}
\end{align*}
with $\Eop(\varepsilon_{1,\model_2} ~|~ \history{1}) = 0.$ Next, we multiply both sides by $\left\{ A_1 - \mu_{1A0}(\history{1}) \right\} \historybasis{1}(\model_1)$. Taking expectations with respect to both sides of the resulting equation, we find
\begin{align*}
  \Eop\left[\left\{ Y_{1 \model_2} - \mu_{1Y \model_2 0}(\history{1}) \right\} \left\{ A_1 - \mu_{1A0}(\history{1}) \right\} \historybasis{1}(\model_1)\right] = \Eop\left[\left\{ A_1 - \mu_{1A0}(\history{1}) \right\}^2 \Delta_{1,\model_2}(\history{1})\historybasis{1}(\model_1)\right] + \bzero,
\end{align*}
where the final zero term comes from the conditional expectation of the error in the previous system. 

Consequently, the stated limiting representation for $\tilde\bG_{1n,\model_2}$ holds if we establish that the quantity $\lim_{n\rightarrow \infty}\Eop\left[ \left\{ \tilde{Y}_{1 \model_2} - Y_{1 \model_2} \right\} \left\{ A_1 - \mu_{1A0}(\history{1}) \right\} \historybasis{1}(\model_1) \right] = \bzero.$

The two pseudo-outcomes only differ based on the parameter used inside the $\xi$ functions. As such, we may write the equality
\begin{align*}
  \Eop\left[\left\{ \tilde Y_{1 \model_2} - \mu_{1Y \model_2 0}(\history{1}) \right\} \left\{ A_1 - \mu_{1A0}(\history{1}) \right\} \historybasis{1}(\model_1)\right] 
  \\
  = \Eop\left[\left\{ \xi\left(A_{2i}, \historybasis{2i}(\model_2); \tilde\btheta_{2n,\model_2} \right) - \xi\left(A_{2i}, \historybasis{2i}(\model_2); \btheta_{20,\model_2} \right) \right\} \left\{ A_1 - \mu_{1A0}(\history{1}) \right\} \historybasis{1}(\model_1)\right].
\end{align*}
This is exactly the piece that was analyzed in \Cref{lem:blip-inf}, although the lemma specialized to $\model_2^*$ rather than a general $\model_2$. This was because the proof generally required the regularity condition \Cref{assump:regularity} to hold. Using the result of the lemma, we find this line equal to
\begin{align*}
  \Eop( \bM) \Eop(\mathrm{Inf}_{2 \model_2 i}) + o(n^{-1/2}),
\end{align*}
where $\mathrm{Inf}_{2 \model_2 i}$ is the influence function for $\tilde\btheta_{2n,\model_2}$ (which is the same as that for $\hat\btheta_{2n,\model_2}$) derived in \Cref{lem:uniform-influence-function-stg2}, and the $o(n^{-1/2})$ remainder follows due to the uniform integrability implied by \Cref{assump:bounded-model,assump:boundedness}. Consequently, $\Eop(\mathrm{Inf}_{2 \model_2 i})=\bzero$, and the remainder term vanishes in the limit, establishing the desired result.

\end{proof}


